\documentclass{article}

\usepackage[letterpaper, left=1in, right=1in, top=1in, bottom=1in]{geometry}
\usepackage{paralist}
\usepackage{xparse}

\usepackage[dvipsnames]{xcolor}
\usepackage[colorlinks=true, linkcolor=blue!70!black, citecolor=blue!70!black,urlcolor=black,breaklinks=true]{hyperref}
\usepackage{microtype}

\usepackage{algorithm}

 \usepackage{natbib}
 \bibpunct{(}{)}{;}{a}{,}{,}

\usepackage{amsthm}
\usepackage{mathtools}
\usepackage{amsmath}
\usepackage{bbm}
\usepackage{amsfonts}
\usepackage{amssymb}
\let\vec\undefined

\usepackage{xpatch}

\theoremstyle{definition}  %

\newtheorem{claim}{Claim}
\newtheorem{lemma}{Lemma}

\newtheorem{corollary}{Corollary}
\newtheorem{proposition}{Proposition}
\newtheorem{fact}{Fact}

\theoremstyle{plain}
\newtheorem{remark}{Remark}

\newtheorem{theorem}{Theorem}
\newtheorem{definition}{Definition}

\numberwithin{claim}{section}
\numberwithin{fact}{section}
\numberwithin{lemma}{section}
\numberwithin{proposition}{section}
\numberwithin{theorem}{section}
\numberwithin{corollary}{section}
\numberwithin{definition}{section}

\xpatchcmd{\proof}{\itshape}{\normalfont\proofnameformat}{}{}
\newcommand{\proofnameformat}{\bfseries}

\makeatletter
\newcommand{\neutralize}[1]{\expandafter\let\csname c@#1\endcsname\count@}
\makeatother

\usepackage{prettyref}
\newcommand{\pref}[1]{\prettyref{#1}}

\newcommand{\savehyperref}[2]{\texorpdfstring{\hyperref[#1]{#2}}{#2}}
\newrefformat{eq}{\savehyperref{#1}{\textup{(\ref*{#1})}}}
\newrefformat{eqn}{\savehyperref{#1}{Equation~\ref*{#1}}}
\newrefformat{lem}{\savehyperref{#1}{Lemma~\ref*{#1}}}
\newrefformat{def}{\savehyperref{#1}{Definition~\ref*{#1}}}
\newrefformat{line}{\savehyperref{#1}{line~\ref*{#1}}}
\newrefformat{thm}{\savehyperref{#1}{Theorem~\ref*{#1}}}
\newrefformat{corr}{\savehyperref{#1}{Corollary~\ref*{#1}}}
\newrefformat{cor}{\savehyperref{#1}{Corollary~\ref*{#1}}}
\newrefformat{sec}{\savehyperref{#1}{Section~\ref*{#1}}}
\newrefformat{app}{\savehyperref{#1}{Appendix~\ref*{#1}}}
\newrefformat{ass}{\savehyperref{#1}{Assumption~\ref*{#1}}}
\newrefformat{asm}{\savehyperref{#1}{Assumption~\ref*{#1}}}
\newrefformat{ex}{\savehyperref{#1}{Example~\ref*{#1}}}
\newrefformat{fig}{\savehyperref{#1}{Figure~\ref*{#1}}}
\newrefformat{alg}{\savehyperref{#1}{Algorithm~\ref*{#1}}}
\newrefformat{rem}{\savehyperref{#1}{Remark~\ref*{#1}}}
\newrefformat{conj}{\savehyperref{#1}{Conjecture~\ref*{#1}}}
\newrefformat{prop}{\savehyperref{#1}{Proposition~\ref*{#1}}}
\newrefformat{proto}{\savehyperref{#1}{Protocol~\ref*{#1}}}
\newrefformat{prob}{\savehyperref{#1}{Problem~\ref*{#1}}}
\newrefformat{claim}{\savehyperref{#1}{Claim~\ref*{#1}}}
\newrefformat{op}{\savehyperref{#1}{Open Problem~\ref*{#1}}}
\newrefformat{fact}{\savehyperref{#1}{Fact~\ref*{#1}}}

\newrefformat{prog}{\savehyperref{#1}{\textup{(\ref*{#1})}}}

\usepackage{complexity}

\usepackage{multirow,array}

\def\1{\bm{1}}

\def\vzero{{\bm{0}}}
\def\vone{{\bm{1}}}
\def\va{{\bm{a}}}

\def\vc{{\bm{c}}}
\def\vd{{\bm{d}}}
\def\ve{{\bm{e}}}

\def\vr{{\bm{r}}}

\def\vv{{\bm{v}}}
\def\vw{{\bm{w}}}
\def\vx{{\bm{x}}}
\def\vy{{\bm{y}}}
\def\vz{{\bm{z}}}

\DeclareMathAlphabet{\mathsfit}{\encodingdefault}{\sfdefault}{m}{sl}
\SetMathAlphabet{\mathsfit}{bold}{\encodingdefault}{\sfdefault}{bx}{n}

\newcommand{\calA}{\ensuremath{\mathcal{A}}}
\newcommand{\calB}{\ensuremath{\mathcal{B}}}

\newcommand{\calD}{\ensuremath{\mathcal{D}}}

\newcommand{\calG}{\ensuremath{\mathcal{G}}}

\newcommand{\calL}{\ensuremath{\mathcal{L}}}

\newcommand{\calN}{\ensuremath{\mathcal{N}}}

\newcommand{\calS}{\ensuremath{\mathcal{S}}}

\newcommand{\calX}{\ensuremath{\mathcal{X}}}
\newcommand{\calY}{\ensuremath{\mathcal{Y}}}
\newcommand{\calZ}{\ensuremath{\mathcal{Z}}}

\let\E\undefined
\let\R\undefined

\newcommand{\E}{\mathbb{E}}

\newcommand{\R}{\mathbb{R}}

\renewcommand{\bar}[1]{\overline{#1}}

\newcommand{\tstar}{t^{\star}}

\newcommand{\step}[1]{^{(#1)} }

\usepackage[short,nocomma]{optidef}

\newcommand{\px}{\bm{x}}

\newcommand{\py}{\bm{y}}

\newcommand{\policy}[1]{\boldsymbol{\pi}_{{#1}}}
\newcommand{\adv}{\text{adv}}

\newcommand{\norm}[1]{\left\| #1 \right\|}

\newcommand{\vzeta}{\boldsymbol{\zeta} }

\newcommand{\vlambda}{\boldsymbol{\lambda} }

\newcommand{\vxi}{\boldsymbol{\xi} }

\newcommand{\vrho}{\boldsymbol{\rho} }

\newcommand{\vpsi}{\boldsymbol{\psi} }
\newcommand{\vomega}{\boldsymbol{\omega} }

\DeclareMathOperator*{\argmin}{arg\,min}

\DeclareMathOperator{\dist}{dist}
\newcommand{\prox}{\mathrm{prox}}

\usepackage{amsthm}
\usepackage{amssymb}
\usepackage{thm-restate}
\usepackage{stmaryrd}
\usepackage[capitalize,noabbrev]{cleveref}
\crefname{claim}{claim}{claims}
\crefname{fact}{fact}{facts}
\crefname{line}{Lines}{lines}

\usepackage{bm}

\usepackage{algpseudocode,algorithm,algorithmicx}

\newcommand*\Let[2]{\State #1 $\gets$ #2}
\algrenewcommand\algorithmicrequire{\textbf{Input:}}
\algrenewcommand\algorithmicensure{\textbf{Output:}}

\newcommand{\projA}[1]{\mathrm{Proj}_{#1}}
\newcommand{\projB}[2]{\mathrm{Proj}_{#1}\left( #2 \right)}

\NewDocumentCommand\proj{ m g }{
  \IfNoValueTF{#2}{\projA{#1}}{\projB{#1}{#2}}
}

\usepackage{xspace}
\newcommand{\team}{\text{team}}
\newcommand{\hatpolicy}{\hat{\boldsymbol{\pi}}}
\newcommand{\GDmax}{$\textsc{GDmax}$\xspace}
\newcommand{\IPGmax}{$\textsc{IPGmax}$\xspace}
\newcommand{\advnashpolicy}{\textup{\texttt{AdvNashPolicy}}\xspace}

\newcommand{\CLS}{\ComplexityFont{CLS}}

\usepackage{soul}

\newcommand{\tlambda}{\tilde{\lambda}}
\newcommand{\tpsi}{\tilde{\psi}}
\newcommand{\tomega}{\tilde{\omega}}
\newcommand{\tzeta}{\tilde{\zeta}}

\usepackage{tabularx}
\usepackage{booktabs}
\usepackage{soul}

\usepackage{tikz}
\date{}

\newcommand{\defeq}{\coloneqq}

\newcommand{\vec}[1]{\bm{#1}}
\newcommand{\mat}[1]{\mathbf{#1}}

\include{macros/nlp-env}
\usepackage{lpform}
\usepackage{subfigure}

\DeclareMathOperator{\pr}{\mathbb{P}}

\usepackage{array}		
\usepackage{booktabs}		
\usepackage[inline,shortlabels]{enumitem}		
\setenumerate{itemsep=0pt,parsep=\smallskipamount,topsep=\smallskipamount,left=\parindent}

\usepackage[utf8]{inputenc} 
\usepackage[T1]{fontenc}    

\usepackage{url}            
\usepackage{booktabs}       
\usepackage{amsfonts}       
\usepackage{nicefrac}       
\usepackage{microtype}      
\usepackage{xcolor}         

\usepackage{autonum}

\usepackage{hyperref}       

\newcommand{\declarecolor}[2]{\definecolor{#1}{RGB}{#2}\expandafter\newcommand\csname #1\endcsname[1]{\textcolor{#1}{##1}}}
\declarecolor{Navy}{0, 0, 128}
\declarecolor{Purple}{126, 50, 171}
\hypersetup{
colorlinks=true,
linktocpage=true,
pdfstartview=FitH,
breaklinks=true,
pdfpagemode=UseNone,
pageanchor=true,
pdfpagemode=UseOutlines,
plainpages=false,
bookmarksnumbered,
bookmarksopen=false,
bookmarksopenlevel=1,
hypertexnames=true,
pdfhighlight=/O,
urlcolor=red,linkcolor=Maroon,citecolor=Purple,	
pdftitle={},
pdfauthor={},
pdfsubject={},
pdfkeywords={},
pdfcreator={pdfLaTeX},
pdfproducer={LaTeX with hyperref}
}

\title{Efficiently Computing Nash Equilibria in Adversarial Team Markov Games}

\author{Fivos Kalogiannis\\\
University of California, Irvine
\and
Ioannis Anagnostides \\
Carnegie Mellon University \\
\and
Ioannis Panageas\footnote{Correspondence to \texttt{ipanagea@ics.uci.edu}.} \\
University of California, Irvine
\and
Emmanouil-Vasileios Vlatakis-Gkaragkounis \\
Columbia University
\and
Vaggos Chatziafratis \\
University of California, Santa Cruz
\and
Stelios Stavroulakis \\
University of California, Irvine
}

\usepackage{tikz}

\newcounter{qst}
\crefname{qst}{Question}{Questions}

\usepackage{lipsum}  
\begin{document}

\maketitle

\begin{abstract}
    Computing Nash equilibrium policies is a central problem in multi-agent reinforcement learning that has received extensive attention both in theory and in practice. However, in light of computational intractability barriers in general-sum games, provable guarantees have been thus far either limited to fully competitive or cooperative scenarios or impose strong assumptions that are difficult to meet in most practical applications.
    
    In this work, we depart from those prior results by investigating infinite-horizon \emph{adversarial team Markov games}, a natural and well-motivated class of games in which a team of identically-interested players---in the absence of any explicit coordination or communication---is competing against an adversarial player. This setting allows for a unifying treatment of zero-sum Markov games and Markov potential games, and serves as a step to model more realistic strategic interactions that feature both competing and cooperative interests. Our main contribution is the first algorithm for computing stationary $\epsilon$-approximate Nash equilibria in adversarial team Markov games with computational complexity that is polynomial in all the natural parameters of the game, as well as $1/\epsilon$.
    
    The proposed algorithm is particularly natural and practical, and it is based on performing independent policy gradient steps for each player in the team, in tandem with best responses from the side of the adversary; in turn, the policy for the adversary is then obtained by solving a carefully constructed linear program. Our analysis leverages non-standard techniques to establish the KKT optimality conditions for a nonlinear program with nonconvex constraints, thereby leading to a natural interpretation of the induced Lagrange multipliers. Along the way, we significantly extend an important characterization of optimal policies in adversarial (normal-form) team games due to Von Stengel and Koller (GEB `97).
\end{abstract}

\pagenumbering{gobble}

\clearpage

\tableofcontents

\clearpage

\pagenumbering{arabic}

\section{Introduction}

Multi-agent reinforcement learning (MARL) offers a principled framework for analyzing competitive interactions in dynamic and stateful environments in which agents' actions affect both the state of the world and the rewards of the other players. Strategic reasoning in such complex multi-agent settings has been guided by game-theoretic principles, leading to many recent landmark results in benchmark domains in AI~\citep{Bowling15:Limit,Silver17:Mastering,Vinyals19:Grandmaster,Mora17:DeepStack,Brown19:Superhuman,Brown18:Superhuman,Brown20:Combining,Perolat22:Mastering}. Most of these remarkable advances rely on scalable and often decentralized algorithms for computing \emph{Nash equilibria}~\citep{nash1951non}---a standard game-theoretic notion of rationality---in two-player zero-sum games.

Nevertheless, while single-agent RL has enjoyed rapid theoretical progress over the last few years (\textit{e.g.}, see \citep{Jin18:Is,Agarwal20:Optimality,Li21:Settling,Luo19:Algorithmic,Sidford18:Near}, and references therein), a comprehensive understanding of the multi-agent landscape still remains elusive. Indeed, provable guarantees for efficiently computing Nash equilibria have been thus far limited to either fully competitive settings, such as two-player zero-sum games~\citep{daskalakis2020independent,Wei21:Last,Sayin21:Decentralized,Cen21:Fast,Sayin20:Fictitious,Condon93:On}, or environments in which agents are endeavoring to coordinate towards a common global objective~\citep{Claus98:The,Wang02:Reinforcement,leonardos2021global,ding2022independent,zhang2021gradient,Chen22:Convergence,Maheshari22:Independent,Fox22:Independent}.

However, many real-world applications feature both shared and competing interests between the agents. Efficient algorithms for computing Nash equilibria in such settings are much more scarce, and typically impose restrictive assumptions that are difficult to meet in most applications~\citep{Hu03:Nash,Bowling00:Convergence}. In fact, even in \emph{stateless} two-player (normal-form) games, computing approximate Nash equilibria is computationally intractable~\citep{Daskalakis09:The,Rubinstein17:Settling,Chen09:Settling,Etessami10:On}---subject to well-believed complexity-theoretic assumptions. As a result, it is common to investigate equilibrium concepts that are more permissive than Nash equilibria, such as \emph{coarse correlated equilibria (CCE)}~\citep{Aumann74:Subjectivity,Moulin78:Strategically}. Unfortunately, recent work has established strong lower bounds for computing even approximate (stationary) CCEs in turn-based stochastic two-player games~\citep{Daskalakis22:The,jin2022complexity}. Those negative results raise the following central question:
\\[3mm]
\begin{tikzpicture}
    \node[text width=15cm,align=center,inner sep=0pt] at (0,0) {\emph{Are there natural multi-agent environments incorporating both competing and shared interests for which we can establish efficient algorithms for computing (stationary) Nash equilibria?}};
    \node at (7.8cm,0) {\refstepcounter{qst}\label{qst:1}($\bigstar$)};
    \node at (-8cm, 0) {};
\end{tikzpicture}
\\[2mm]
Our work makes concrete progress in this fundamental direction. Specifically, we establish the first efficient algorithm leading to Nash equilibria in \emph{adversarial team Markov games}, a well-motivated and natural multi-agent setting in which a team of agents with a common objective is facing a competing adversary.

\subsection{Our Results}

Before we state our main result, let us first briefly introduce the setting of adversarial team Markov games; a more precise description is deferred to \Cref{sec:Markov-ATGs}. To address \Cref{qst:1}, we study an infinite-horizon Markov (stochastic) game with a finite state space $\mathcal{S}$ in which a team of agents $\calN_A \defeq [n]$ with a common objective function is competing against a \emph{single} adversary with opposing interests. Every agent $k \in [n]$ has a (finite) set of available actions $\mathcal{A}_k$, while $\mathcal{B}$ represents the adversary's set of actions. We will also let $\gamma \in [0,1)$ be the \emph{discounting factor}. Our goal will be to compute an (approximate) Nash equilibrium; that is, a strategy profile so that neither player can improve via a unilateral deviation (see \cref{def:Nash}). In this context, our main contribution is the first polynomial time algorithm for computing Nash equilibria in adversarial team Markov games:

\begin{theorem}[Informal]
\label{thm:informal}
There is an algorithm (\IPGmax) that, for any $\epsilon > 0$, computes an $\epsilon$-approximate stationary Nash equilibrium policy profile in adversarial team Markov games, and runs in time 
$$\poly\left(|\calS|, \sum_{k = 1}^n |\calA_k|+ |\calB|, \frac{1}{1 - \gamma}, \frac{1}{\epsilon} \right).$$
\end{theorem}

A few remarks are in order. First, our guarantee significantly extends and unifies prior results that only applied to either \emph{two-player} zero-sum Markov games or to \emph{Markov potential games}; both of those settings can be cast as special cases of adversarial team Markov games (see \Cref{sec:Markov}). Further, the complexity of our algorithm, specified in \Cref{thm:informal}, scales only with $\sum_{k \in \mathcal{N}_{\calA}} |\mathcal{A}_k|$ instead of $\prod_{k \in \mathcal{N}_{\calA}} |\mathcal{A}_k|$, bypassing what is often referred to as the \emph{curse of multi-agents}~\citep{jin2021v}. Indeed, viewing the team as a single ``meta-player'' would induce an action space of size $\prod_{k \in \mathcal{N}_{\calA}} |\mathcal{A}_k|$, which is \emph{exponential} even if each agent in the team has only two actions. In fact, our algorithm operates without requiring any (explicit) form of coordination or communication between the members of the team (beyond the structure of the game itself), a feature that has been motivated in a series of practical applications~\citep{VonStengel97:Team}. Namely, scenarios in which communication or coordination between the members of the team is either overly expensive, or even infeasible; for a more in depth discussion regarding this point we refer to the work of~\citep{schulman2017duality}.


\subsection{Overview of Techniques}

To establish \Cref{thm:informal}, we propose a natural and decentraliezd algorithm we refer to as \emph{Independent Policy GradientMax} (\IPGmax). \IPGmax works in turns. First, each player in the team performs one independent policy gradient step on their value function with an appropriately selected learning rate $\eta > 0$. In turn, the adversary best responds to the current policy of the team. This exchange is repeated for a sufficiently large number of iterations $T$. Finally, \IPGmax includes an auxiliary subroutine, namely $\texttt{AdvNashPolicy}()$, which computes the Nash policy of the adversary; this will be justified by \Cref{prop:extend} we describe below.

Our analysis builds on the techniques of~\citet{lin2020gradient}---developed for the saddle-point problem $\min_{\Vec{x} \in \calX} \max_{\Vec{y} \in \calY} f(\Vec{x}, \Vec{y})$---for characterizing \GDmax. Specifically, \GDmax consists of performing gradient descent steps, but on the function $\phi(\Vec{x}) \defeq \max_{\py \in \mathcal{Y}} f(\vec{x},\py)$. \citet{lin2020gradient} showed that \GDmax converges to a point $\big( \hat{\px},\py^*(\hat \px) \big)$ such that $\hat{\px}$ is an approximate first-order stationary point of the \emph{Moreau envelope} (see \Cref{def:moreau}) of $\phi(\vec{x})$, while $\py^*(\hat\px)$ is a best response to $\hat{\px}$. Now if $f(\px, \cdot)$ is \emph{strongly-concave}, one can show (by Dantzig's theorem) that $\big(\hat{\px},\py^*(\px)\big)$ is an approximate first-order stationary point of $f$. However, our setting introduces further challenges since the value function $V_{\vrho}(\policy{\team},\policy{\adv})$ is nonconvex-nonconcave.

For this reason, we take a more refined approach. We first show in \Cref{lem:ipgmax-convergence-lemma} that \IPGmax is guaranteed to converge to a policy profile $\big(\hatpolicy_{\team}, \cdot \big)$ such that $\hatpolicy_{\team}$ is an $\epsilon$-nearly stationary point of $\max_{\policy{\adv}} V_\rho( \policy{\team}, \policy{\adv})$. Then, the next key step and the crux in the analysis is to show that $\hatpolicy_{\team}$ \emph{can be extended to an $O(\epsilon)$-approximate Nash equilibrium policy:}

\begin{proposition}[Informal]
    \label{prop:extend}
If $\hatpolicy_{\team}$ is an $\epsilon$-nearly stationary point of $\max_{\policy{\adv}} V_\rho( \policy{\team},\policy{\adv})$, there exists a policy for the adversary $\hatpolicy_\adv$ so that $(\hatpolicy_\team,\hatpolicy_\adv)$ is an $O(\epsilon)$-approximate Nash equilibrium policy.
\end{proposition}

In the special case of normal-form games, a weaker extension theorem for \emph{team-maxmin equilibria} was established in the seminal work of~\citet[pp. 315]{von1997team}. In particular, that result was derived by employing fairly standard linear programming techniques, analyzing the structural properties of the (exact) \emph{global optimum} of a suitable LP. On the other hand, our more general setting introduces several new challenges, not least due to the nonconvexity-nonconcavity of the objective function.

Indeed, our analysis leverages more refined techniques from nonlinear programming. More precisely, while we make use of standard policy gradient properties, similar to the single-agent MDP setting~\citep{agarwal2021theory,xiao2022convergence},  
our analysis does not rely on the so-called \emph{gradient-dominance} property~\citep{bhandari2019global}, as that property does not hold in a team-wise sense. Instead, inspired by an alternative proof of Shapley's theorem~\citep{shapley1953stochastic} for two-person zero-sum Markov games~\citep[Chapter 3]{filar2012competitive}, we employ mathematical programming. One of the central challenges is that the induced nonlinear program has a nonconvex set of constraints. As such, even the existence of (nonnegative) Lagrange multipliers satisfying the KKT conditions is not guaranteed, thereby necessitating more refined analysis techniques.

To this end, we employ the \emph{Arrow-Hurwiz-Uzawa constraint qualification} (\Cref{theorem:AHU}) in order to establish that the local optima are contained in the set of KKT points (\Cref{corollary:KKT-sat}). Then, we leverage the structure of adversarial team Markov games to characterize the induced Lagrange multipliers, showing that a subset of these can be used to establish \Cref{prop:extend}; incidentally, this also leads to an efficient algorithm for computing a (near-)optimal policy of the adversary. Finally, we also remark that controlling the approximation error---an inherent barrier under policy gradient methods---in \Cref{prop:extend} turns out to be challenging. We bypass this issue by constructing ``relaxed'' programs that incorporate some imprecision in the constraints. A more detailed overview of our algorithm and the analysis is given in \Cref{sec:mainn}.

\section{Preliminaries}
\label{sec:prelims}

In this section, we introduce the relevant background and our notation. \Cref{sec:Markov-ATGs} describes adversarial team Markov games. \Cref{sec:pol} then defines some key concepts from multi-agent MDPs, while \Cref{sec:Markov} describes a generalization of adversarial team Markov games, beyond identically-interested team players, allowing for a richer structure in the utilities of the team---namely, adversarial Markov potential games.

\paragraph{Notation.} We let $[n] \defeq \{1, \dots, n\}$. We use superscripts to denote the (discrete) time index, and subscripts to index the players. We use boldface for vectors and matrices; scalars will be denoted by lightface variables. We denote by $\|\cdot\| \defeq \|\cdot\|_2$ the Euclidean norm. For simplicity in the exposition, we may sometimes use the $O(\cdot)$ notation to suppress dependencies that are polynomial in the natural parameters of the game; precise statements are given in the Appendix. For the convenience of the reader, a comprehensive overview of our notation is given in \ref{table:notation}.

\subsection{Adversarial Team Markov Games}
    \label{sec:Markov-ATGs}
    An \emph{adversarial team Markov game} (or an adversarial team \emph{stochastic} game) is the Markov game extension of static, normal-form adversarial team games~\citep{von1997team}. The game is assumed to take place in an infinite-horizon discounted setting in which a team of identically-interested agents gain what the adversary loses. Formally, the game $\calG$ is represented by a tuple $\calG = (\calS, \calN, \calA, \calB,  r, \pr, \gamma, \rho )$ whose components are defined as follows.
    
    \begin{itemize}[topsep=0pt,leftmargin=5ex]        \setlength{\itemsep}{1pt}
        \item $\calS$ is a finite and nonempty set of \emph{states}, with cardinality $S \defeq | \calS |$;
        \item $\calN$ is the set of players, partitioned into a set of $n$ team agents $\calN_A \defeq [n]$ and a single \emph{adversary}
        \item $\calA_k$ is the action space of each player in the team $k \in [n]$, so that $\calA \defeq \bigtimes_{k \in [n]} \calA_k$, while $\calB$ is the action space of the adversary. We also let $A_k \defeq |\calA_k |$ and $B \defeq | \calB|$;\footnote{To ease the notation, and without any essential loss of generality, we will assume throughout that the action space does not depend on the state.}
        \item $r : \calS \times \calA \times  \calB \to (0,1)$ is the (deterministic) instantaneous \emph{reward function}\footnote{Assuming that the reward is positive is without any loss of generality (see \Cref{claim:stratequiv}).} representing the (normalized) payoff of the adversary, so that for any $(s, \vec{a}, b) \in \calS \times \calA \times \calB$,
        \begin{equation}
            \label{eq:adv}
        r(s, \vec{a}, b) + \sum_{k =1}^n r_k(s, \vec{a}, b) = 0,
        \end{equation}
        and for any $k \in [n]$,
        \begin{equation}
            \label{eq:ident-int}
        r_k(s, \vec{a}, b) = r_{\team}(s, \vec{a}, b).
        \end{equation}
        \item $\pr:  \calS \times \calA \times \calB \to \Delta(\calS)$ is the \emph{transition probability function}, so that $\pr( s'| s, \vec{a}, b )$ denotes the probability of transitioning to state $s' \in \calS$ when the current state is $s \in \calS$ under the action profile $(\vec{a}, b) \in \calA \times \calB$;
        \item $\gamma \in [0, 1)$ is the \emph{discount factor}; and
        \item $\vrho \in \Delta(\calS)$ is the \emph{initial state distribution} over the state space. We will assume that $\vec{\rho}$ is full-support, meaning that $\rho(s) > 0$ for all $s \in \calS$.
    \end{itemize}
    
    In other words, an adversarial team Markov game is a subclass of general-sum infinite-horizon multi-agent discounted MDPs under the restriction that all but a single (adversarial) player have identical interests (see~\eqref{eq:ident-int}), and the game is globally zero-sum---in the sense of \eqref{eq:adv}. As we point out in \Cref{sec:Markov}, \eqref{eq:ident-int} can be relaxed in order to capture \emph{(adversarial) Markov potential games} (\Cref{def:potential}), without qualitatively altering our results.
    
    \subsection{Policies, Value Function, and Nash Equilibria}
    \label{sec:pol}

    \paragraph{Policies.} A \emph{stationary}---that is, time-invariant---policy $\policy{k}$ for an agent $k$ is a function mapping a given state to a distribution over available actions, $\policy{k} : \calS \ni s \mapsto \policy{k}(\cdot|s) \in \Delta( \calA_k)$. We will say that $\policy{k}$ is \emph{deterministic} if for every state there is some action that is selected with probability $1$ under policy $\policy{k}$. For convenience, we will let $\Pi_\team : \calS \rightarrow \Delta(\calA)$ and $\Pi_\adv : \calS \rightarrow \Delta(\calB)$ denote the policy space for the team and the adversary respectively. We may also write $\Pi: \calS \rightarrow \Delta(\calA) \times \Delta(\calB)$ to denote the joint policy space of all agents.
    
    \paragraph{Direct Parametrization.} Throughout this paper we will assume that players employ \emph{direct policy parametrization}. That is, for each player $k \in [n]$, we let $\calX_k \defeq \Delta(\calA_k)^S$ and $\policy{k} = \px_{k}$ so that $x_{k,s,a} = \policy{k}(a|s)$. Similarly, for the adversary, we let $\calY \defeq \Delta(\calB)^S$ and $\policy{\adv} = \py$ so that $y_{s,a} = \policy{\adv}(a|s)$. (Extending our results to other policy parameterizations, such as soft-max~\citep{agarwal2021theory}, is left for future work.)
    


    \paragraph{Value Function.} The \emph{value function} $V_s : \Pi \ni (\policy{1}, \dots, \policy{n}, \policy{\adv} ) \mapsto \R$ is defined as the expected cumulative discounted reward received by the adversary under the joint policy $(\policy{\team}, \policy{\adv}) \in \Pi$ and the initial state $s \in \calS$, where $\policy{\team} \defeq (\policy{1}, \dots, \policy{n})$. In symbols,
    \begin{equation}
        V_{s}(\policy{\team}, \policy{\adv}) \defeq \E_{(\policy{\team}, \policy{\adv})}\left[ \sum_{t=0}^{\infty} \gamma^t r(s^{(t)}, \vec{a}^{(t)}, b^{(t)})  \big| s_0 = s \right],
        \label{eq:value-func-def}
    \end{equation}
    where the expectation is taken over the trajectory distribution induced by $\policy{\team}$ and $\policy{\adv}$. When the initial state is drawn from a distribution $\vec{\rho}$, the value function takes the form $V_{\vec{\rho}}(\policy{\team}, \policy{\adv}) \defeq \E_{s \sim \vec{\rho} } \Big[ V_s(\policy{\team}, \policy{\adv}) \Big]$. 

    \paragraph{Nash Equilibrium.} Our main goal is to compute a joint policy profile that is an (approximate) \emph{Nash equilibrium}, a standard equilibrium concept in game theory formalized below.
    
    \begin{definition}[Nash equilibrium]
        \label{def:Nash}
        A joint policy profile  $\big( \policy{\team}^\star, \policy{\adv}^\star \big) \in \Pi$ is an $\varepsilon$-approximate Nash equilibrium, for $\epsilon \geq 0$, if
    \begin{align}
        \setlength\arraycolsep{3pt}
        \left\{
        \begin{array}{rlll}
            V_{\vec{\rho}}  ( \policy{\team}^\star, \policy{\adv}^\star) & \leq
            V_{\vec{\rho}} ( (\policy{k}', \policy{-k}^\star), \policy{\adv}^\star  \big) + \varepsilon, 
            & \forall k \in [n], \forall{\policy{k}' } \in \Pi_k,
            \\
            V_{\vec{\rho}} ( \policy{\team}^\star, \policy{\adv}^\star) & \geq
            V_{\vec{\rho}} ( \policy{\team}^\star, \policy{\adv}') - \varepsilon, 
            &\forall \policy{\adv}' \in \Pi_{\adv}.
            \end{array}
        \right.
    \end{align}
    \end{definition}
    That is, a joint policy profile is an (approximate) Nash equilibrium if no unilateral deviation from a player can result in a non-negligible---more than additive $\epsilon$---improvement for that player. Nash equilibria always exist in multi-agent stochastic games~\citep{Fink64:Equilibrium}; our main result implies an (efficient) constructive proof of that fact for the special case of adversarial team Markov games.
    
    \subsection{Adversarial Markov Potential Games}
        \label{sec:Markov}
        
       A recent line of work has extended the fundamental class of potential normal-form games~\citep{monderer1996potential} to \emph{Markov potential games}~\citep{marden2012state,macua2018learning,leonardos2021global,ding2022independent,zhang2021gradient,Chen22:Convergence,Maheshari22:Independent,Fox22:Independent}. Importantly, our results readily carry over even if players in the team are not necessarily identically interested, but instead, there is some underlying potential function for the team; we will refer to such games as \emph{adversarial Markov potential games}, formally introduced below.
    \begin{definition}
        \label{def:potential}
        An adversarial Markov potential game $\calG = (\calS, \calN, \calA, \calB, \{r_k\}_{k \in [n]}, \pr, \gamma, \rho )$ is a multi-agent discounted MDP that shares all the properties of adversarial team Markov games (\Cref{sec:Markov-ATGs}), with the exception that \eqref{eq:ident-int} is relaxed in that there exists a potential function $\Phi_s, ~\forall s \in \calS  $, such that for any $\policy{\adv} \in \Pi_{\adv}$,
        \begin{equation}
            \Phi_s ( \policy{k}, \policy{-k}; \policy{\adv}) - \Phi_s(\policy{k}', \policy{-k}; \policy{\adv}) = V_{ k,s}( \policy{k}, \policy{-k}; \policy{\adv}) - V_{ k,s} (\policy{k}', \policy{-k}; \policy{\adv}),
        \end{equation}
        for every agent $k \in [n]$, every state $s \in \calS$, and all policies $\policy{k}, \policy{k'} \in \Pi_k$ and $\policy{-k} \in \Pi_{-k}$.
    \end{definition}

\section{Main Result}
\label{sec:mainn}

In this section, we sketch the main pieces required in the proof of our main result, \Cref{thm:informal}. We begin by describing our algorithm in \Cref{sec:algo}. Next, in \Cref{sec:IPGmax}, we characterize the strategy $\hat{\vx} \in \calX$ for the team returned by \IPGmax, while \Cref{sec:extendibility-poly-time} completes the proof by establishing that $\hat{\vx}$ can be efficiently extended to an approximate Nash equilibrium. The formal proof of \Cref{thm:informal} is deferred to the Appendix.

\subsection{Our Algorithm}
\label{sec:algo}

In this subsection, we describe in more detail \IPGmax, our algorithm for computing $\epsilon$-approximate Nash equilibria in adversarial team Markov games (\Cref{alg:ipg-max}). \IPGmax takes as input a precision parameter $\epsilon > 0$ (Line \ref{line:prec}) and an initial strategy for the team $(\vx_1^{(0)}, \dots, \vx_n^{(0)}) = \vx^{(0)} \in \calX \defeq \bigtimes_{k=1}^n \calX_k$ (Line \ref{line:initstrar}). The algorithm then proceeds in two phases:
\begin{itemize}
    \item In the first phase the team players are performing independent policy gradient steps (Line \ref{line:grad}) with learning rate $\eta$, as defined in Line \ref{line:lr}, while the adversary is then best responding to their strategy (Line \ref{line:bestres}). This process is repeated for $T$ iterations, with $T$ as defined in Line~\ref{line:T}. We note that $\proj{}{\cdot}$ in Line~\ref{line:grad} stands for the Euclidean projection, ensuring that each player selects a valid strategy. The first phase is completed in Line~\ref{line:setx}, where we set $\hat{\vx}$ according to the iterate at time $\tstar$, for some $0 \leq \tstar \leq T -1$. As we explain in \Cref{sec:IPGmax}, selecting uniformly at random is a practical and theoretically sound way of setting $\tstar$.
    \item In the second phase we are fixing the strategy of the team $\hat{\vx} \in \calX$, and the main goal is to determine a strategy $\hat{\vy} \in \calY$ so that $(\hat{\vx}, \hat{\vy})$ is an $O(\epsilon)$-approximate Nash equilibrium. This is accomplished in the subroutine \texttt{AdvNashPolicy}$(\hat{\vx})$, which consists of solving a linear program---from the perspective of the adversary---that has polynomial size. Our analysis of the second phase of \IPGmax can be found in \Cref{sec:extendibility-poly-time}.
\end{itemize}

\label{sec:main}
\begin{algorithm}
  \caption{Independent Policy GradientMax (\IPGmax) \label{alg:ipg-max}}
  \begin{algorithmic}[1]
   \State Precision $\epsilon > 0$ \label{line:prec}
   \State Initial Strategy $\vx^{(0)} \in \calX$ \label{line:initstrar}
   \State Learning rate $\eta \defeq \frac{ \epsilon^2 (1 - \gamma )^9}{ 32 S^4 D^2  \left( \sum_{k=1}^n A_k + B \right)^3 } $ \label{line:lr}
   \State Number of iterations $T \defeq    \frac{ 512 S^8 D^4  \left( \sum_{k=1}^n A_k + B \right)^4 }{ \epsilon^4 (1-\gamma)^{12} } $ \label{line:T}
    \For{$t \gets 1,2, \dots, T$}
      \Let{$\vy^{(t)}$}{$ \arg\max_{\vy \in \calY} V_{\vrho} \big( \vx^{(t-1)}, \vy \big)$} \label{line:bestres}
      \Let{$\vx_k^{(t)}$}{$
                            \proj{\calX_k}{
                                \vx_k^{(t-1)}- \eta \nabla_{\vx_k}  V_{\vrho}\big( \vx^{(t-1)}, \vy^{(t)} \big)
                                }
                        $} \label{line:grad} \Comment{for all agents $i \in [n]$}
      \EndFor
      \Let{$\hat{\vx}$}{$\vx^{(\tstar)}$} \label{line:setx}
      \Let{$\hat{\py}$}{$\texttt{AdvNashPolicy}(\hat{\px})$} \Comment{defined in \Cref{alg:adv-nash-policy}}
      \State \Return{$(\hat{\px},\hat{\py})$}
  \end{algorithmic}
\end{algorithm}

\subsection{Analyzing Independent Policy GradientMax}
\label{sec:IPGmax}

In this subsection, we establish that \IPGmax finds an $\epsilon$-nearly stationary point $\hat{\vx}$ of $\phi(\Vec{x}) \defeq \max_{\vy \in \calY} V_{\vrho}(\vx, \vy)$ in a number of iterations $T$ that is polynomial in the natural parameters of the game, as well as $1/\epsilon$; this is formalized in \Cref{lem:ipgmax-convergence-lemma}. 

First, we note the by-now standard property that the value function $V_{\vrho}$ is $L$-Lipschitz continuous and $\ell$-smooth, where $L \defeq \frac{\sqrt{\sum_{k=1}^n A_k +B}}{(1-\gamma)^2}$ and $\ell \defeq \frac{2 \left(\sum_{k=1}^n A_{k}+B\right)}{(1-\gamma)^3}$ (\Cref{lem:smoothness}). An important observation for the analysis is that \IPGmax is essentially performing gradient descent steps on $\phi(\vx)$. However, the challenge is that $\phi(\vx)$ is not necessarily differentiable; thus, our analysis relies on the \emph{Moreau envelope} of $\phi$, defined as follows.

\begin{definition}[Moreau Envelope]\label{def:moreau} 
    Let $\phi(\vx) \defeq \max_{\vy \in \calY} V_{\vrho}(\vx, \vy)$. For any $0 < \lambda < \frac{1}{\ell}$ the Moreau envelope $\phi_\lambda$ of $\phi$ is defined as 
\begin{equation}
    \label{eq:Moreau}
    \phi_{\lambda}(\px) \defeq \min_{\px' \in \calX} \left\{ \phi(\px') + \frac{1}{2\lambda}\norm{\px-\px'}^2 \right\}.
\end{equation}
We will let $\lambda \defeq \frac{1}{2\ell}$.
\end{definition}

Crucially, the Moreau envelope $\phi_\lambda$, as introduced in \eqref{eq:Moreau}, is $\ell$-strongly convex; this follows immediately from the fact that $\phi(\vx)$ is \emph{$\ell$-weakly convex}, in the sense that $\phi(\vx) + \frac{\ell}{2} \|\vx\|^2$ is convex (see \Cref{lemma:max-weakly}). A related notion that will be useful to measure the progress of \IPGmax is the \emph{proximal mapping} of a function $f$, defined as $\prox_f : \calX \ni \vx \mapsto \argmin_{x' \in \calX} \left\{ f(\vx') + \frac{1}{2} \|\vx' - \vx\|^2 \right\}$; the proximal of $\phi/(2\ell)$ is well-defined since $\phi$ is $\ell$-weakly convex (\Cref{prop:welldefined}). We are now ready to state the convergence guarantee of \IPGmax.

\begin{restatable}{proposition}{ipgmx}
  \label{lem:ipgmax-convergence-lemma}
Consider any $\epsilon > 0$. If $\eta = 2\epsilon^2  ( 1-\gamma)$ and  $T=\frac{(1-\gamma)^4}{8 \epsilon^4 (\sum_{k=1}^n A_k+B)^2}$, there exists an iterate $\tstar$, with $0 \leq \tstar \leq T -1 $, such that $\norm{\vx^{(\tstar)}-\tilde{\vx}^{(\tstar)}}_2 \leq \epsilon$, where $\tilde{\vx}^{(\tstar)} \defeq \prox_{\phi/(2\ell)}(\vx^{(\tstar)})$.
\end{restatable}

The proof relies on the techniques of~\citet{lin2020gradient}, and it is deferred to \Cref{sec:convergence}. The main takeaway is that $O(1/\epsilon^4)$ iterations suffice in order to reach an $\epsilon$-nearly stationary point of $\phi$---in the sense that it is $\epsilon$-far in $\ell_2$ distance from its proximal point. A delicate issue here is that \Cref{lem:ipgmax-convergence-lemma} only gives a best-iterate guarantee, and identifying that iterate might introduce a substantial computational overhead. To address this, we also show in \Cref{cor:highprob} that by randomly selecting $\lceil \log(1/\delta)\rceil$ iterates over the $T$ repetitions of \IPGmax, we are guaranteed to recover an $\epsilon$-nearly stationary point with probability at least $1 - \delta$, for any $\delta > 0$.


\subsection{Efficient Extension to Nash Equilibria}
\label{sec:extendibility-poly-time}

In this subsection, we establish that any $\epsilon$-nearly stationary point $\hat{\vx}$ of $\phi$, can be \emph{extended} to an $O(\epsilon)$-approximate Nash equilibrium $(\hat{\vx}, \hat{\vy})$ for any adversarial team Markov game, where $\hat{\vy} \in \calY$ is the strategy for the adversary. Further, we show that $\hat{\vy}$ can be computed in polynomial time through a carefully constructed linear program. This ``extendibility'' argument significantly extends a seminal characterization of~\citet{von1997team}, and it is the crux in the analysis towards establishing our main result, \Cref{thm:informal}.

To this end, the techniques we leverage are more involved compared to~\citep{von1997team}, and revolve around nonlinear programming. Specifically, in the spirit of~\citep[Chapter 3]{filar2012competitive}, the starting point of our argument is the following nonlinear program with variables $(\vx,\vv) \in \calX \times \R^S$:

\begin{subequations}
        \makeatletter
        \def\@currentlabel{\text{Q-}\mathrm{NLP}}
        \makeatother
        \renewcommand{\theequation}{$Q$\arabic{equation}}
        \begin{tagblock}[tagname={$\text{Q-}\mathrm{NLP}$},content={\label{prog:xinlp}}]
        \begin{alignat}{3}
         & \mathrlap{ \min ~ \sum_{s \in \calS} \rho(s) v(s) + \ell
         \|\px - \hat{\px} \|^2} \label{eq:xiobj} \\
        \label{eq:xibr}
        & \text{s.t.}~ & r(s, \px, b) + \gamma \sum_{s' \in \calS} \pr(s' | s, \px, b) v(s')  \leq v(s), & \quad  \forall  (s,b) \in \calS \times \calB;    \label{eq:xibr} \\
         & & \px_{k,s}^\top \vone  = 1, & \quad \forall (k,s) \in [n] \times \calS; \text{ and} \label{eq:xieqone} 
         \\
         & & x_{k,s,a}  \geq 0, & \quad \forall k \in [n], (s,a) \in \calS \times \calA_k.
         \label{eq:xinonneg}
        \end{alignat}
        \end{tagblock}
    \end{subequations}
    Here, we have overloaded notation so that $r(s,\vx, b) \defeq \E_{\Vec{a} \sim \vx_s}[r(s, \Vec{a}, b]$ and $\pr(s' | s, \vx, b)) \defeq \E_{\Vec{a} \sim \vx_s}[\pr(s'|s, \vec{a}, b)]$. For a fixed strategy $\vx \in \calX$ for the team, this program describes the (discounted) MDP faced by the adversary. A central challenge in this formulation lies in the nonconvexity-nonconcavity of the constraint functions, witnessed by the multilinear constraint~\eqref{eq:xibr}. Importantly, unlike standard MDP formulations, we have incorporated a quadratic regularizer in the objective function; this term ensures the following property.
    
    \begin{proposition}
    For any fixed $\vx \in \calX$, there is a unique optimal solution $\vv^\star$ to \eqref{prog:xinlp}. Further, if $\tilde{\vx} \defeq \prox_{\phi/(2\ell)}(\hat{\vx})$, and $\tilde{\vv} \in \R^S$ is the corresponding optimal, then $(\tilde{\vx}, \tilde{\vv})$ is the global optimum of \eqref{prog:xinlp}.
    \end{proposition}
    
    The uniqueness of the associated value vector is a consequence of Bellman's optimality equation, while the optimality of the proximal point follows by realizing that \eqref{prog:xinlp} is an equivalent formulation of the proximal mapping. These steps are formalized in~\Cref{sec:nlp+phi}. Having established the optimality of $(\tilde{\vx}, \tilde{\vv})$, the next step is to show the existence of nonnegative Lagrange multipliers satisfying the KKT conditions (recall \Cref{def:KKT}); this is non-trivial due to the nonconvexity of the feasibility set of \eqref{prog:xinlp}.
    
    To do so, we leverage the so-called \emph{Arrow-Hurwicz-Uzawa constraint qualification} (\Cref{theorem:AHU})---a form of ``regularity condition'' for a nonconvex program. Indeed, in \Cref{lemma:AHU} we show that any feasible point of \eqref{prog:xinlp} satisfies that constraint qualification, thereby implying the existence of nonnegative Lagrange multipliers satisfying the KKT conditions for any local optimum (\Cref{corollary:KKT-sat}), and in particular for $(\tilde{\vx}, \tilde{\vv})$:
    
    \begin{proposition}
        \label{prop:info-kkt}
    There exist nonnegative Lagrange multipliers satisfying the KKT conditions at $(\tilde{\vx}, \tilde{\vv})$.
    \end{proposition}
    
    Now the upshot is that a subset of those Lagrange multipliers $\tilde{\vlambda} \in \R^{S \times B}$ can be used to establish the extendability of $\hat{\vx}$ to a Nash equilibrium. Indeed, our next step makes this explicit: We construct a linear program whose sole goal is to identify such multipliers, which in turn will allow us to efficiently compute an admissible strategy for the adversary $\hat{\vy}$. However, determining $\tilde{\vlambda}$ exactly seems too ambitious. For one, \IPGmax only granted us access to $\hat{\vx}$, but not to $\tilde{\vx}$. On the other hand, the Lagrange multipliers $\tilde{\vlambda}$ are induced by $(\tilde{\vx}, \tilde{\vv})$. To address this, the constraints of our linear program are phrased in terms of $(\hat{\vx}, \hat{\vv})$, instead of $(\tilde{\vx}, \tilde{\vv})$, while to guarantee feasibility we appropriately relax all the constraints of the linear program; this relaxation does not introduce too much error since $\|\hat{\vx} - \tilde{\vx}\| \leq \epsilon$ (\Cref{lem:ipgmax-convergence-lemma}), and the underlying constraint functions are Lipschitz continuous---with constants that depend favorably on the game $\calG$; we formalize that in \Cref{lem:lpadv-feasible}. This leads us to our main theorem, summarized below (see \Cref{theorem:approx-Nash} for a precise statement).
    
    \begin{theorem}
        \label{theorem:main}
    Let $\hat{\vx}$ be an $\epsilon$-nearly stationary point of $\phi$. There exist a linear program, \eqref{prog:lpadv}, such that:
    \begin{enumerate}[(i)]
        \item It has size that is polynomial in $\calG$, and all the coefficients depend on the (single-agent) MDP faced by the adversary when the team is playing a fixed strategy $\hat{\vx}$; and
        \item It is always feasible, and any solution induces a strategy $\hat{\vy}$ such that $(\hat{\vx}, \hat{\vy})$ is an $O(\epsilon)$-approximate Nash equilibrium.
    \end{enumerate}
    \end{theorem}
    
    The proof of this theorem carefully leverages the structure of adversarial team Markov games, along with the KKT conditions we previously established in~\Cref{prop:info-kkt}. The algorithm for computing the policy for the adversary is summarized in \Cref{alg:adv-nash-policy} of \Cref{sec:extendibility}. A delicate issue with \Cref{theorem:main}, and in particular with the solution of \eqref{prog:lpadv}, is whether one can indeed \emph{efficiently simulate} the environment faced by the adversary. Indeed, in the absence of any structure, determining the coefficients of the linear program could scale exponentially with the number of players; this is related to a well-known issue in computational game theory, revolving around the exponential blow-up of the input space as the number of players increases~\citep{Papadimitriou08:Computing}. As is standard, we bypass this by assuming access to natural oracles that ensure we can efficiently simulate the environment faced by the adversary (\Cref{remark:oracle}).
    


\section{Further Related Work}
\label{sec:related}

In this section, we highlight certain key lines of work that relate to our results in the context of adversarial team Markov games. We stress that the related literature on multi-agent reinforcement learning (MARL) is too vast to even attempt to faithfully cover here. For some excellent recent overviews of the area, we refer the interested reader to~\citep{yang2020overview,zhang2021multi} and the extensive lists of references therein.

\paragraph{Team Games.}

The study of team games has been a prolific topic of research in economic theory and group decision theory for many decades; see, \textit{e.g.}, \citep{Marschak55:Elements,Groves73:Incentives,Radner62:Team,Ho72:Team}. A more modern key reference point to our work is the seminal paper of~\citet{von1997team} that introduced the notion of \emph{team-maxmin equilibrium (TME)} in the context of normal-form games. A TME profile is a mixed strategy for each team member so that the minimal expected team payoff over all possible responses of the adversary---who potentially knows the play of the team---is the maximum possible. While TME's enjoy a number of compelling properties, being the optimal equilibria for the team given the lack of coordination, they suffer from computational intractability even in $3$-player team games~\citep{hansen2008approximability,Borgs10:The}.\footnote{\citet{hansen2008approximability,Borgs10:The} establish $\FNP$-hardness and inapproximability for general $3$-player games, but their argument readily applies to $3$-player team games as well.} Nevertheless, practical algorithms have been recently proposed and studied for computing them in multiplayer games~\citep{Converging20:Zhang,Zhang20:Computing,Basilico17:Team}. It is worth pointing out that team equilibria are also useful for extensive-form two-player zero-sum games where one of the players has \emph{imperfect recall}~\citep{Piccione97:On}.

The intractability of TME has motivated the study of a relaxed equilibrium concept that incorporates a \emph{correlation device}~\citep{Farina18:Ex,celli2018computational,Basilico17:Team,zhang2020,Zhang21:Team,Zhang22:Team,Carminati22:A}; namely, \emph{TMECor}. In TMECor players are allowed to select \emph{correlated strategies}. Despite the many compelling aspects of TMECor as a solution concept in team games, even \emph{ex ante} coordination or correlated randomization---beyond the structure of the game itself---can be overly expensive or even infeasible in many applications~\citep{von1997team}. Further, even TMECor is $\NP$-hard to compute (in the worst-case) for \emph{imperfect-information} extensive-form games (EFGs)~\citep{Chu01:On}, although fixed-parameter-tractable ($\FPT$) algorithms have recently emerged for natural classes of EFGs~\citep{Zhang21:Team,Zhang22:Team}.

On the other hand, the computational aspects of the standard Nash equilibrium (NE) in adversarial team games is not well-understood, even in normal-form games. In fact, it is worth pointing out that Von Neumann's celebrated \emph{minimax theorem}~\citep{vonNeumannMorgenstern+2007} does not apply in team games, rendering traditional techniques employed in two-player zero-sum games of little use. Indeed, \citet{schulman2017duality} provided a precise characterization of the \emph{duality gap} between the two teams based on the natural parameters of the problem, while \citet{kalogiannis2021teamwork} showed that standard no-regret learning dynamics such as gradient descent and optimistic Hedge could fail to stabilize to mixed NE even in binary-action adversarial team games. Finally, we should also point out that although from a complexity-theoretic standpoint our main result (\Cref{thm:informal}) establishes a \emph{fully polynomial time approximate scheme (FPTAS)}, since the dependence on the approximation error $\epsilon$ is $\poly(1/\epsilon)$, an improvement to $\poly(\log(1/\epsilon))$ is precluded even in normal-form games unless $\CLS \subseteq \P$ (an unlikely event); this follows since adversarial team games immediately capture potential games~\citep{kalogiannis2021teamwork}, wherein computing mixed Nash equilibria is known to be complete for the class $\CLS = \PPAD \cap \PLS$~\citep{Babichenko21:Settling}.


\paragraph{Multi-agent RL.} Computing Nash equilibria has been a central endeavor in multi-agent RL. While some algorithms have been proposed, perhaps most notably the Nash-Q algorithm~\citep{Hu98:Multiagent,Hu03:Nash}, convergence to Nash equilibria is only guaranteed under severe restrictions on the game. More broadly, the long-term behavior of independent policy gradient methods~\citep{Schulman15:Trust} is still not well-understood. Before all else, from the impossibility result of \citeauthor{hart2003uncoupled}, universal convergence to Nash equilibria is precluded even for normal-form games; this is aligned with the computational intractability ($\PPAD$-completeness) of Nash equilibria even in two-player general-sum games~\citep{Daskalakis09:The,Chen09:Settling}. Surprisingly, recent work has also established hardness results in turn-based stochastic games, rendering even the weaker notion of (stationary) CCEs intractable~\citep{Daskalakis22:The,jin2022complexity}.

As a result, the existing literature has inevitably focused on specific classes of games, such as Markov potential games~\citep{leonardos2021global,ding2022independent,zhang2021gradient,Chen22:Convergence,Maheshari22:Independent,Fox22:Independent} or two-player zero-sum Markov games~\citep{daskalakis2020independent,Wei21:Last,Sayin21:Decentralized,Cen21:Fast,Sayin20:Fictitious}. As we pointed out earlier, adversarial Markov team games can unify and extend those settings (\Cref{sec:Markov}). More broadly, identifying multi-agent settings for which Nash equilibria are provably efficiently computable is recognized as an important open problem in the literature (see, \textit{e.g.}, \citep{daskalakis2020independent}), boiling down to one of the main research question of this paper (\Cref{qst:1}). We also remark that certain guarantees for convergence to Nash equilibria have been recently obtained in a class of symmetric games~\citep{Emmons22:For}---including symmetric team games. Finally, weaker solution concepts relaxing either the Markovian or the stationarity properties have also recently attracted attention~\citep{Daskalakis22:The,jin2021v}.

\section{Conclusions}
    \label{sec:Conclusion}
    Our main contribution in this paper is the first polynomial algorithm for computing (stationary) Nash equilibria in adversarial team Markov games, an important class of games in which a team of uncoordinated but identically-interested players is competing against an adversarial player. We argued that this setting serves as a step to model more realistic multi-agent applications that feature both competing and cooperative interests. One caveat of our main algorithm ($\IPGmax$) is that it requires a separate subroutine for computing the optimal policy of the adversary. It is plausible that a carefully designed two-timescale policy gradient method can efficiently reach to a Nash equilibrium, towards fully model-free algorithms for solving adversarial team Markov games. Exploring different solution concepts--- beyond Nash equilibria---could also be a fruitful avenue for the future. Indeed, allowing some limited form of correlation between the players in the team could lead to more efficient algorithms; whether that form of coordination is justified (arguably) depends to a large extent on the application at hand. Finally, returning to \Cref{qst:1}, a more ambitious agenda revolves around understanding the fundamental structure of games for which computing Nash equilibria is provably computationally tractable.

\section*{Acknowledgments}
Ioannis Anagnostides is grateful to Gabriele Farina and Brian H. Zhang for helpful discussions. Ioannis Panageas would like to acknowledge a start-up grant. Part of this project was done while he was a visiting research scientist at the Simons Institute for the Theory of Computing for the program ``Learning and Games''. Vaggos Chatziafratis was supported by a start-up grant of UC Santa Cruz, the Foundations of Data Science Institute (FODSI) fellowship at MIT and Northeastern, and part of this work was carried out at the Simons Institute for the Theory of Computing. Emmanouil V. Vlatakis-Gkaragkounis is grateful for financial support by the Google-Simons Fellowship, Pancretan Association of America and Simons Collaboration on Algorithms and Geometry. This project was completed while he was a visiting research fellow at the Simons Institute for the Theory of Computing. Additionally, he would like to acknowledge the following series of NSF-CCF grants under the numbers 1763970/2107187/1563155/1814873.
\bibliography{main}

\newpage

\appendix


\newcommand{\matlip}{{\textstyle\sqrt{\sum_{k=1}^n{A_k}}}}
\newcommand{\rewlip}{{\textstyle\sqrt{\sum_{k=1}^n{A_k}}}}
\newcommand{\maxrew}{M_r}
\newcommand{\minrew}{m_r}
\newcommand{\distht}{{\epsilon}}

\section{Additional Preliminaries}

\subsection{Background on Nonlinear Programming}
    In this subsection, we provide additional background on the theory of nonlinear programming~\citep{mangasarian1994nonlinear}.
%
For the purposes of our work,
given a constrained minimization problem, we are interested in conditions that establish the existence of (nonnegative) Lagrange multipliers that satisfy the \emph{Kaursh-Kuhn-Tucker (KKT)} conditions. In the unconstrained case, this property corresponds to the nullification of the gradient at point of a local optimum (Fermat's Theorem). In the constrained case  however, further regularity conditions have to be met with respect to the underlying feasible set; this is formalized via so-called \emph{constraint qualifications}~\citep{bazaraa1972constraint,giorgi2018guided}. For our purposes, we will use the so-called \emph{Arrow-Hurwicz-Uzawa constraint qualification}~\citep{arrow1961constraint,mangasarian1994nonlinear} (see \Cref{theorem:AHU}) to show that the set of (local) optima of a particular constrained optimization problem is contained within the set of KKT points (\Cref{lemma:AHU}).
    
We first define the \emph{nonlinear program} that encodes a \emph{constrained minimization problem}. Then, we state the Karush-Kuhn-Tucker optimality conditions for a given feasible point of the problem.
    
    \paragraph{Constrained optimization problems.}
    In a \emph{constrained optimization problem} on a Euclidean space $\R^d$, for $d \in \mathbb{N}$, we are interested in optimizing a given function $f: \R^d \to \R$ over a given nonempty set $\mathcal{D}\subseteq \R^d$. The function $f$ is called
    the \emph{objective function}, and the set $\mathcal{D}$ the \emph{constraint} or \emph{feasibility} set. Notationally, we will represent such problems by
    \begin{center}
        ``Minimize $f(\vz) $ subject to $\vz\in \mathcal{D}$'', or more compactly as ``$\min\left\{f(\vz)\medskip\ |\medskip\ \vz\in\calD\right\}$.''
    \end{center}
    A global solution to such a problem is a point $\vz^*$ in $\calD$ such that $f(\vz^*) \leq f(\vz)$ for all $\vz \in \calD$; the existence of such a solution is typically guaranteed by Weierstrass' theorem.
    %
    Relaxing the requirement of global optimality, below we clarify the meaning of a local minimum with following definition:
    \begin{definition}[Local minimum]
        Let a function $f : \R^d \rightarrow \R$ and a point $\vz_0$. A \emph{constrained local minimum} occurs at $\vz_0\in\calD$, with $\calD \subseteq \R^d$, if there exists $\delta > 0$ such that
        \begin{equation}
            f(\vz_0) \leq f(\vz), ~ \forall \vz \in \{ B(\vz_0,\delta) \} \cap \calZ,
        \end{equation}
        where $B(\vz_0,\delta) $ denotes the set of all points belonging to the open ball with radius $\delta$ and center at $\vz_0$.
    \end{definition}
    
    We now turn to study constrained optimization problems with feasible sets defined by inequality constraints. Namely, the constraint set will have the form 
    \begin{equation}
        \label{eq:setD}
      \calD = \{\vz \in U \ |\ g_i(\vz) \leq 0, \forall i=1,\ldots, m\},  
    \end{equation}
    where $U \subseteq \R^d$ is an open set in $\R^d$, and $m$ is the number of 
    the necessary inequalities to describe the feasible set $\calD$. The minimization problem can now be written as follows.
    \begin{alignat}{2}
        \begin{array}{ll}
            \min & f(\vz)  \\
            \mathrm{s.t.} & g_i(\vz) \leq 0, ~ \forall i \in [m]. 
        \end{array}
        \label{eq:minimization-problem}
        \tag{MP}
    \end{alignat}
    
    In the sequel, we say
    that an inequality constraint $ g_i(\vz) \leq 0 $ is \emph{active} at a point $\vz^*$ if the constraint holds as an equality at $\vz^*$, that is, we have $g_i(\vz^*)=0$; otherwise, it is called \emph{inactive}. Below we introduce the KKT conditions (\emph{e.g.}, see \citep[Chapter 5.5.3]{boyd2004convex}). 
    
    \begin{definition}[Karush-Kuhn-Tucker Conditions]
    \label{def:KKT}
    Suppose that $f: U \to \R$ and $g_i: U \to \R$ are differentiable functions, for any $i=1,\ldots,m$. Further, let $  \calL(\vz,\vlambda) \defeq f(\vz) + \sum_{i = 1}^m \lambda_i g_i(\vx)$ be the associated Lagrangian function. We say that a point $(\vec{z}^*, \lambda^*)$ satisfies the KKT conditions if  
        \begin{equation}
   \label{eq:KKT}
        \begin{aligned}
& &\lambda_i^* g_i(\vz^*) = 0, \quad \forall i = 1, \dots, m; & && & \text{(Complementary Slackness)}\\
& & g_i(\vz^*) \leq 0, \quad \forall i = 1, \dots, m; & && & \text{(Primal Feasibility)} \\
& & \lambda_i^* \geq 0, \quad \forall i = 1, \dots, m; \text{ and} & && & \text{(Dual Feasibility)} \\
& &\nabla_{\vz} \calL(\vz^*,\vlambda^*)=\nabla_{\vz} f(\vz^*) + \sum_{i=1}^m \lambda_i^*\nabla_{\vz} g_i(\vz^*) = \vzero. & && & \text{(First-Order Stationarity)}\\
        \end{aligned}
         \tag{KKT}
    \end{equation}
    \end{definition}

    In general, while these conditions are necessary for optimality, they are not necessarily sufficient. We also remark that for the unconstrained case, \emph{i.e.}, $\{g_i(\vz)\equiv 0\}$, the \eqref{eq:KKT} conditions generalize the necessary condition of a gradient equal to zero.


    \paragraph{The Arrow-Hurwicz-Uzawa constraint qualification.} The establish the KKT conditions under nonconvex constraints, a number of different constraint qualifications have been developed~\citep{bazaraa1972constraint,giorgi2018guided}. We recall that constraint qualifications ensure that all the local minimizers acquire a set of (nonnegative) Lagrange multipliers that (jointly) satisfy the KKT conditions (\Cref{def:KKT}). For our purposes, we will use the Arrow-Hurwicz-Uzawa constraint qualification (henceforth AHU-CQ for brevity), which is recalled below (see \citep[Ch. 7]{mangasarian1994nonlinear}). 
    
    \begin{theorem}[AHU-CQ~\citep{mangasarian1994nonlinear}]
    \label{theorem:AHU}
        Consider a constrained minimization problem with a feasibility set $\calD$ given in \eqref{eq:setD}. Further, let $\vz_0$ be a feasible point and let $A(\vz_0)$ denote the set of \emph{active constraints} at $\vz$.  We differentiate between concave $A' (\vz_0)$ and nonconcave $A''(\vz_0)$ active constraints, so that $A (\vz_0) = A' (\vz_0) \cup A'' (\vz_0)$. If there exists a vector $\vw \in \R^d$ such that
        \begin{equation}
        \label{eq:CQ}
            \begin{cases}
                    \vw^\top \nabla_{\vz} g_i(\vz_0) \geq 0, & \forall i \in A'; \text{ and}  \\
                    \vw^\top \nabla_{\vz} g_i(\vz_0) > 0, & \forall i \in A'',
            \end{cases}
        \end{equation}
        then, the Arrow-Hurwicz-Uzawa constraint qualification at point $\vz_0$ is satisfied.
    \end{theorem}
    
    The importance of this theorem lies in the following implication, which provides sufficient conditions for the satisfaction of the KKT conditions.
    
    \begin{corollary}
        \label{cor:ahu-cq}
        Consider a local minimum $\vz_0$ of \eqref{eq:minimization-problem}. If the Arrow-Hurwicz-Uzawa constraint qualification is satisfied at $\vz_0$, there exist (nonnegative) Lagrange multipliers satisfying the \eqref{eq:KKT} conditions of \Cref{def:KKT}.
    \end{corollary}
    
    It is important to stress that the Arrow-Hurwicz-Uzawa constraint qualification---see \eqref{eq:CQ}---does \emph{not} involve the objective function; this is the case more broadly for constraint qualifications.
\subsection{Weak Convexity, the Moreau Envelope, and Near-Stationarity}
In this subsection, we provide some necessary background on optimizing nonsmooth functions. We refer the interested reader to~\citep{davis2019stochastic} for a more complete discussion on the subject. 

Throughout this subsection, we will tacitly assume that $\calX$ and $\calY$ are nonempty, convex and compact subsets of a Euclidean space. We will also denote by $\dist$ the distance between a vector $\vx$ and $\calY$, defined as follows.
\begin{equation}
    \dist(\vx; \calY) = \min_{\vy \in \calY} \| \vx - \vy \|_2.
\end{equation}


\begin{definition}[Weak Convexity]
A function $f:\R^d \rightarrow \R$ is said to be convex if for any $\vx_1,\vx_2\in \R^d$ and any $t\in[0,1]$, it holds that $f(t \vx_1 + (1-t)\vx_2 )\leq t f(\vx_1)+ (1-t) f(\vx_2)$. Additionally, a function $f:\R^d \rightarrow \R$ is said to be $\lambda$-weakly convex if the function $f(\vx) + \frac{\lambda}{2} \| \vx \|^2$ is convex.
\end{definition}

The following corollary
is an immediate consequence of the definition of weak convexity, and the fact that the function $\frac{\lambda}{2} \|\vx\|^2$ is $\lambda$-strongly convex.

\begin{corollary}
    \label{cor:weak-conv}
    Let $f:\calX \ni \vx \mapsto \R$ be a $\lambda$-weakly convex function. Then, the function $f(\vx) + \lambda \| \vx \|^2$ is $\lambda$-strongly convex.
\end{corollary}

A notion closely related to weak convexity within optimization literature is the \emph{Moreau envelope} (also knwon as Moreau-Yosida regularization). Namely, the Moreau envelope of a function is defined as follows for $\lambda > 0$.
\begin{equation}
    f_{\lambda}(\vx) \defeq \min_{\vx' \in \calX} \left\{ f(\vx') + \frac{1}{2\lambda} \|\vx - \vx' \|^2 \right\}.
\end{equation}

Moreover, when $\lambda < \frac{1}{\ell}$, with $\ell$ being the corresponding parameter of weak convex, the Moreau envelope $f_{\lambda}$ is $C^1$-smooth, and its gradient given by $\nabla f_{\lambda} = \lambda^{-1}(\vx - \prox_{\lambda f}(\vx) )$~\citep[Theorem 31.5]{rockafellar1970convex}, where $\prox_{\lambda f}(\cdot)$ is the \emph{proximal mapping}. Namely, for a convex and continuous function $f : \calX \rightarrow \R$ we define its proximal operator $\prox_{f} : \R^d \rightarrow \R^d$ as follows.
\begin{equation}
    \label{eq:prox}
    \prox_{f} (\vx) = \argmin_{\vx' \in \calX} \left\{ f(\vx') + \frac{1}{2}\|\vx - \vx' \|^2 \right\}.
\end{equation}
The point $\tilde{\vx}\defeq \prox_{f}(\vx)$ that results by applying the proximal operator~\eqref{eq:prox} on $\vx$ is called the \textit{proximal point} of $\vx$. The proximal point of the scaled function $\lambda f$ coincides with the solution of the minimization problem needed in order to determine the Moreau envelope of $f$ at $\vx$. The proximal operator of an $\ell$-weakly convex function is well-defined, as as long as $\lambda$ is sufficiently small:

\begin{proposition}[\cite{lin2020gradient}]
    \label{prop:welldefined}
    Let $\phi$ be a $\ell$-weakly convex function. Then, $\prox_{\phi/(2\ell)}(\vx)$ is well-defined.
\end{proposition}




\paragraph{Minimization of weakly convex functions.}
Generally, in a minimization problem we are interested in computing minima of a function subject to constraints. If no convexity assumption holds for the objective function, even computing local minima is \NP-hard~\citep{murty1985some}. Instead, one is often interested in computing an approximate stationary point of the objective function.


More precisely, an $\epsilon$-approximate stationary point $\vx_0$ of a nondifferentiable function is a point such that $\dist(\vzero; \partial f(\vx_0)) \leq \epsilon$ where $\partial f(\vx_0)$ is the \emph{subdifferential} of $f$ at $\vx_0$ (see \citep[Sec. 2.2]{davis2019stochastic}). 
However, such a measure of stationarity for nonsmooth objective functions is so restrictive that, in fact, it can be shown as difficult as solving the optimization problem exactly---\textit{e.g.}, if we let $f(x) = |x|$ then $x=0$ is the only $\epsilon$-approximate stationary point for $\epsilon\in[0,1)$.

The alternative notion of \emph{near stationarity} for a nonsmooth function $f(\vx)$, contributed by~\citet{davis2019stochastic}, has become standard ({see Propositions 4.11 and 4.12 in \citep{lin2020gradient}}) for optimization of weakly convex functions. (For a more in depth discussion see \citep[Section 4.1]{drusvyatskiy2019efficiency}.) More precisely, we measure stationarity by means of the proximal operator:

\begin{definition}[$\epsilon$-nearly stationary point]
    \label{def:nearst}
    Let $f : \calX \rightarrow \R$ be a continuous, nonsmooth function, and some $\epsilon > 0$. We say that a point $\vx_0 \in \calX$ is $\epsilon$-nearly stationary if
    \begin{equation}
        \| \vx_0 - \tilde{\vx}_{0} \|_2 \leq \epsilon,
    \end{equation}
    where $\tilde{\vx}_0 \defeq \prox_{\lambda f}(\vx_0)$ is the proximal point of $\vx_0$.
\end{definition}

The Moreau envelope of $f$ offers a number of useful properties for the analysis of convergence to near stationarity, as formalized below. 
\begin{fact}[\citep{davis2019stochastic}]
Let $f:\calX\rightarrow\R$ be an $\ell$-weakly convex function  and $\lambda < \frac{1}{\ell}$. Further, let $\vx \in \calX$ and $\tilde{\vx}\defeq \prox_{\lambda f}(\vx)$ be its proximal point. Then,
\begin{equation}
    \left\{
    \begin{array}{rcl}
        \| \vx - \tilde{\vx} \|_2 &\leq& \lambda \| \nabla f_{\lambda} (\vx) \|; \\
        f(\tilde{\vx}) &\leq& f(\vx); \\
        \dist(\vzero ; \partial f (\vx) ) &\leq & \| \nabla f_{\lambda} (\vx)  \|.
    \end{array}
    \right.
\end{equation}
\end{fact}
\begin{remark}
    An $\frac{\epsilon}{\lambda}$-approximate first-order order stationary point of $f_{\lambda}$ is an $\epsilon$-near stationary point of $f$.
\end{remark}

\paragraph{Properties of the max function.} In our analysis of $\IPGmax$, we will measure progress based on the function $\phi(\vx) = \max_{\py \in \calY}f(\px, \py)$, where $f$ corresponds to the value function in our setting; using $\phi$ is fairly common in the context of min-max optimization. The following lemma points out some useful properties of $\phi$.

\begin{lemma}[\cite{lin2020gradient}]
    \label{lemma:max-weakly}
    Let $f:\calX\times\calY\rightarrow\R$ be $L$-Lipschitz and $\ell$-smooth. Then, the function $\phi(\vx) = \max_{\vy \in \calY} f(\vx, \vy) $ is
    \begin{itemize}
        \item $L$-Lipschitz continuous; and
        \item $\ell$-weakly convex.
    \end{itemize}
    \label{lem:max-weakly-convex}
\end{lemma}

\subsection{Further Background on Markov Decision Processes}

Additionally, we will need some further preliminaries on Markov decision processes (MDPs). First, the \emph{(discounted) state visitation measure} effectively measures the ``discounted'' expected amount of time
that the Markov chain---induced by fixing the players' policies---spends at a state $s$ given that it starts from an initial state $\bar{s}$. That is, every visit is multiplied by a discount factor $\gamma^t$, where $t$ is the time of the visit. We note that~\citet{agarwal2021theory} use the definition that makes it a probability measure, in the sense that for a given initial state distribution $\vrho$ the discounted state visitation distribution sums to $1$. For convenience, we will work with the unnormalized definition found in \citep[Chapter 6.10]{puterman2014markov} that instead sums to $\frac{1}{1-\gamma}$; this is the reason why we use the term \textit{measure} instead of \textit{distribution}.

\begin{definition}
    Consider an initial state distribution $\vrho\in \Delta(\calS)$ and a stationary joint policy $\policy{} \in \Pi$.
    The state visitation measure $d_{\bar{s}}^{\policy{}}$ is defined as
    \begin{equation}
        d^{\policy{} }_{\bar{s}} (s) = \sum_{t=0}^{\infty}\gamma^t \pr(s\step{t} = s |\policy{}, s\step{0} = \bar{s}).
    \end{equation}
    Further, overloading notation, we let
    \begin{equation}
        d^{\policy{} }_{\vrho} (s) = \E_{\bar{s}\sim \vrho} \left[ d^{\policy{}}_{\bar{s}} (s) \right].
    \end{equation}
\end{definition}

With a slight abuse of notation, we will also write $d^{\vx, \vy}_{\vrho} (s)$ to denote the state visitation measure induced by strategies $(\vx, \vy) \in \calX \times \calY$. 

\begin{definition}[Distribution Mismatch Coefficient]
    Let $\vrho \in \Delta(\calS)$ be a full-support distribution over states, and $\Pi$ be the joint set of policies. We define the distribution mismatch coefficient $D$ as
    \begin{equation}
        D \defeq \sup_{\policy{} \in \Pi } \left\|\frac{\vd^{\policy{}}_{\vrho} } {\vrho}\right\|_\infty,
    \end{equation}
    where $\frac{\vd^{\policy{}}_{\vrho}}{\vrho}$ denotes element-wise division.
    \label{def:mismatch}
\end{definition}

Below, for the sake of mathematical rigor, we discuss a mathematical technicality that appears repeatedly in policy gradient methods in both single- and multi-agent RL. 
{

\begin{remark}[Empty Interior]
In multi-agent Markov games the joint policy profiles of all agents belong to  $\mathcal{X}\times\mathcal{Y}=\prod_{k\in\mathcal{N}_A}\Delta(\mathcal{A}_k)^{S}\times\Delta(\mathcal{B})^{S}\subseteq \mathbb{R}^{S \times(\sum_{k\in\mathcal{N}_A} A_k + B)}$. In other words, the value function is defined on a product of orthogonal simplices, and as such on a space that does not have an interior.  Thus, a more careful discussion is necessary regarding the well-posedness of $\nabla_{\vx_k} V_{\vrho}(\vx,\vy)$,$\nabla_{\vy} V_{\vrho}(\vx,\vy)$. 

By definition,
the gradient of an arbitrary function $f$ is defined as the unique vector field whose dot product with any vector $\vv$ at each point $\vx$ is the directional derivative of $f$ along $\vv$, \text{i.e.,} $\langle \nabla f(\vx),\vv\rangle=D_{\vv} f(\vx)$ (the right-hand side is the directional derivative). 
However, the derivative at a point $\vx$ is a notion that involves a certain limit. This limit requires the existence of a topological net, \text{i.e.}, a partially-ordered infinite collection of open sets whose intersection converges to $\{\vx\}$. The latter cannot exist as $\mathcal{X}_k,\mathcal{Y}$ lack an interior. 

To alleviate this issue, we can use a standard continuity argument from Clarke's generalized gradient for semi-algebraic functions (see \citep{drusvyatskiy2015clarke} and references therein). Roughly speaking, we dilate the domain set by a $\delta$-ball and use continuity to establish that as $\delta$ approaches $0$ the gradients remain well-defined.

This approach was originally discussed for normal-form games in \citep{aubin1981locally} and for stochastic games in \citep{daskalakis2020independent}). More analytically, by the compactness of $\mathcal{X}$ and $\mathcal{Y}$ for any $\delta>0$, the outer $\delta$-parallel bodies $B(\mathcal{X},\delta)$ and 
 $B(\mathcal{Y},\delta)$ are closed, convex, and compactly contain $\mathcal{X},\mathcal{Y}$ correspondingly. Additionally, these bodies are full dimensional and contain a dense interior. Hence, the gradients can be well-defined on these extensions of the domain. Another important note is that thanks to the semi-algebraic form of the value function, $V_{\vrho}(\vx,\vy)$ defined on these full dimensional smooth manifolds is a $L+\mathrm{fun}_{Lip}(\delta)$ Lipschitz and $\ell+\mathrm{fun}_{Smooth}(\delta)$ smooth for some $\mathrm{fun}_{Smooth},\mathrm{fun}_{Lip}$ which are uniformly continuous functions on $\mathcal{X},\mathcal{Y}$ such that $\displaystyle\lim_{\delta\to 0}\mathrm{fun}_{Lip}(\delta)=\displaystyle\lim_{\delta\to 0}\mathrm{fun}_{Smooth}(\delta)=0$.
Thus, we can extend the notion of the (sub)gradient of $V_{\vrho}$ at a point $(\px,\py)\in \mathcal{X}\times\mathcal{Y}$ as the limit of the corresponding (sub)gradients of $V_{\vrho}$ defined on $B(\mathcal{X},\delta)\times
 B(\mathcal{Y},\delta)$ while $\delta\downarrow 0$. For a further extensive discussion, see \citep{drusvyatskiy2013semi}.
 
%

\end{remark}
}
\newpage
To further enhance the paper's readability, we include below a table listing our main notation.
\begin{table}[H]
    \label{table:notation}
    \caption{Notation}
    \begin{tabularx}{\textwidth}{p{0.22\textwidth}X}
    \toprule[1pt]
    \midrule[0.3pt]
  \multicolumn{2}{l}{{{Parameters of the model:}}}    \\
  \midrule

      $\calS$                   & State space\\
      $\calN$                   & Set of players\\
      $r$                       & Reward function of the adversary\\
      $n$                       & Number of players in the team \\
      $\calA_k$                 & Action space of player $k$ of the team \\
      $\calA$                   & Team's joint action space \\
      $\calB$                   & Action space of the adversary \\
      $A_k$                     & Number of actions available to player $k$ of the team \\
      $B$                       & Number of actions available to the adversary \\
      $\calX_k$                 & The set of feasible directly parameterized policies of player $k$: $\calX_k \defeq \Delta(\calA_k)^S$ \\
      $\calX$                   & The set of feasible directly parameterized policies of the team: $\calX \defeq \bigtimes_{k=1}^n \calX_k$ \\
      $\calY$                   & The set of feasible directly parameterized policies of the adversary: $\calY \defeq \Delta(\calB)^S$ \\
      $\gamma$                  & Discount factor \\
      $\pr(s'|s, \Vec{a}, b )$  & Probability of transitioning from state $s$ to $s'$ under the action profile $(\Vec{a}, b)$ \\
      $\pr( \px, \py )$         & The (row-stochastic) transition matrix of the Markov chain induced by $(\px, \py)$ \\
      $\textstyle\mat{V}(\px, \py), V_{\vrho}(\px, \py)$ & The value vector per-state, the expected value under initial distribution $\vrho$\\
  \multicolumn{2}{l}{{{Parameters:}}}    \\ 
    \midrule 
      $L$                       & Lipschitz constant of the value function $V_{\vrho}(\cdot, \cdot)$\\
      $\ell$                    & Smoothness constant of the value function $V_{\vrho}(\cdot, \cdot)$\\
      $D$                       & Distribution mismatch coefficient \\
    \multicolumn{2}{l}{{{NLP:}}}
    \\
    \midrule
      $\vrho$, $\rho(s)$        & Initial state distribution, probability that $s$ is the initial state\\
      $\vv$, $v(s)$             & value vector, value of state $s$ \\
      $r(s, \px, b)$             & Expected reward at state $s$ under $(\vx, b) \in \calX \times \calB$ \\
      $\pr( s' | s, \px, b)$       & Expected probability of transitioning to state $s'$ from $s$ under $(\vx, b) \in \calX \times \calB$ \\
      $\px_{k}$                 & Strategy of player $k$ of the team \\
      $\px_{k,s}$               & Strategy of player $k$ of the team at state $s$ \\
    \multicolumn{2}{l}{{{Additional notation:}}} 
    \\
    \midrule
      $\phi(\px)$               & Maximum of the value function w.r.t. $\px$: $\phi(\px) = \max_{\py \in \calY} V_{\vrho}(\px, \py)$ \\
      $\phi_{1/2\ell}(\px)$     & Moreau envelope of $\phi$ with parameter $\lambda \defeq \frac{1}{2\ell}$ (\Cref{def:moreau}) \\
      $\hat{\px}$               %
                                & $\epsilon$-nearly stationary point of $\phi(\cdot)$ (\Cref{def:nearst}) \\
      $\tilde{\px}$             & proximal point of $\frac{1}{2\ell} \phi$ w.r.t. $\hat{\vx}$ (\Cref{eq:prox}) \\
      \midrule[0.3pt]\bottomrule[1pt]
     \end{tabularx}
    \end{table}
    \newpage

\newcommand{\red}[1]{{\color{red}#1}}

\section{Proof of Extendibility to Nash Equilibria}
\label{sec:extendibility}

In this section, we demonstrate how a nearly stationary point $\hat{\px}$ of $\phi(\cdot) \defeq \max_{\py \in \calY} V_{\vrho}(\cdot, \py)$, returned by \IPGmax, can be extended to an approximate Nash equilibrium.


Our extension argument uses a nonlinear program that is in spirit similar to the one found in \citep[Chapter 3.9]{filar2012competitive}. But, unlike the program in~\citep[Chapter 3.9]{filar2012competitive}, ours is designed to capture adversarial team Markov games. In this context, there are two main challenges in the proof. First, even if we had an \emph{exact} stationary point of $\phi$, establishing the existence of nonnegative Lagrange multipliers that satisfy the KKT conditions is particularly challenging. This is unfortunate since it turns out that establishing the KKT conditions is crucial, and is at the heart of our extendibility argument. Indeed, the upshot is that an admissible policy for the adversary \emph{can be derived from a subset of the Lagrange multipliers}. Further, our algorithm only has access to an \emph{approximate} stationary point. As a result, our argument needs to be robust in terms of approximation errors.



To address the first issue, we consider a modified nonlinear program---namely, \eqref{prog:xinlp} introduced earlier in \Cref{sec:extendibility-poly-time}---that incorporates an additional quadratic term to the objective function. This allows us to show that the proximal point $\tilde{\px} \defeq \prox_{\phi/(2\ell)}(\hat{\px})$ is part of a global optimum for our new program. In turn, this is crucial to establish the existence of nonnegative Lagrange multipliers at that point. Moreover, we bypass the second issue we discussed above by studying a relaxed linear program, which serves as a proxy for the ideal linear program that uses knowledge of the global optimum of \eqref{prog:xinlp}. Our main argument establishes that any solution to the proxy linear program is basically as good as solving the ideal one---modulo factors that depend polynomially on the natural parameters of the game. In turn, that solution---which incidentally can be computed efficiently---induces a strategy profile $\hat{\vy} \in \calY$ so that $(\hat{\vx}, \hat{\vy})$ is an $O(\epsilon)$-approximate Nash equilibrium.

 \paragraph{Outline of the proof.} Below we sketch the main steps in our proof.

 \begin{enumerate}[(i)]
     \item In \Cref{sec:quad-nlp} we consider \eqref{prog:xinlp}, a nonlinear program that incorporates an additional quadratic term to the objective function of the natural MDP formulation~\eqref{prog:nlp}.
     \item In \Cref{sec:nlp+phi} we show that \eqref{prog:xinlp} attains a global optimum at $(\Tilde{\vx}, \tilde\vv)$ (\Cref{lemma:global}), where $\Tilde{\vx} \defeq \prox_{\phi/(2\ell)}(\hat{\vx})$ and $\Tilde{\vv}$ is the unique value vector associated with $\tilde{\vx}$ (\Cref{prop:phi-and-nlp}).
     \item In \Cref{sec:ahu-cq-sat} we show that any feasible point of \eqref{prog:xinlp} satisfies the Arrow-Hurwicz-Uzawa constraint qualification (\Cref{lemma:AHU}). In turn, this implies the existence of nonnegative Lagrange multipliers at $(\tilde{\vx}, \tilde{\vv})$ satisfying the KKT conditions (\Cref{corollary:KKT-sat}).
     \item In \Cref{sec:fosp-extend-to-ne} we introduce a linear program, namely \eqref{prog:lpadv}, to formulate the optimization problem faced by the adversary; \eqref{prog:lpadv} will be eventually used to compute an admissible policy for the adversary.
     \item In \Cref{lem:lpadv-feasible} we show that \eqref{prog:lpadv} is always feasible. This is shown by first constructing an ``ideal'' linear program~\eqref{prog:strictlpadv}, and arguing that the ideal program is feasible (\Cref{lemma:idealfeas}) using the KKT conditions. The transition to \eqref{prog:lpadv} leverages the fact that $\|\Tilde{\vx} - \hat{\vx}\| \leq \epsilon$ and the Lipschitz continuity of the underlying constraint functions to show that the introduced error is only $O(\epsilon)$.
     \item Finally, this section is culminated in \Cref{lem:lpadv-nash} and \Cref{theorem:approx-Nash}, which establish that any solution of~\eqref{prog:lpadv} induces a policy for the advesrary $\hat{\vy} \in \calY$ so that $(\hat{\vx}, \hat{\vy})$ is an $O(\epsilon)$-approximate Nash equilibrium.
  \end{enumerate}



 \subsection{The Quadratic NLP}
    \label{sec:quad-nlp}
    
    In this subsection, we describe in more detail the nonlinear program~\eqref{prog:xinlp} we introduced earlier in \Cref{sec:extendibility-poly-time}. For completeness, let us first describe the perhaps most natural nonlinear formulation used to solve the min-max problem $\min_{\px \in \calX} \max_{\py \in \calY}  V_\rho\big( \px, \py\big)$ (see \citep[Chapter 3]{filar2012competitive}), introduced below.
    
    \begin{subequations}
        \makeatletter
        \def\@currentlabel{\mathrm{NLP}_\calG}
        \makeatother
        \renewcommand{\theequation}{$\mathrm{NLP}_\calG$.\arabic{equation}}
        \begin{tagblock}[tagname={$\mathrm{NLP}_\mathcal{G}$},content={\label{prog:nlp}}]
        \begin{alignat}{3}
         &\mathrlap{ \min~ \sum_{s \in \calS} \rho(s) v(s) }   \label{eq:obj}   \\
         &\text{s.t.}~ & r(s, \px, b) + \gamma \sum_{s' \in \calS} \pr(s' | s, \px, b) v(s')  \leq v(s), &  \quad  \forall  (s, b) \in \calS \times \calB;   \label{eq:br} \\
        & & \px_{k,s}^\top \vone  = 1, & \quad \forall (k,s) \in [n] \times \calS; \text{ and} \label{eq:eqone} 
         \\
        & & x_{k,s,a}  \geq 0,  &  \quad \forall k \in [n], (s,a) \in \calS \times \calA_k.
         \label{eq:nonneg}
        \end{alignat}
        \end{tagblock}
    \end{subequations}
    The variables of this program correspond to a strategy profile for the team players $(\vx_1, \dots, \vx_n) \in \calX$, while the value vector $\vv$ captures the value at each state when the adversary is best responding. Before we proceed further, it will be useful to note that, for any $(s,b) \in \calS \times \calB$ and $s' \in \calS$, the functions $r(s,\vx,b)$ and $\pr(s'|s, \vx, b)$ are multinear in $\vx$, so that
    \begin{align}
        \left\{
        \begin{array}{lcl}
            r\big( s, (\px_k; \px_{-k}), b \big) &= & \sum_{a \in \calA_k} x_{k,s,a} r\big(s, (\ve_{k, s,a}; \px_{-k} ), b \big); \text{ and} \\
            \pr\big( s' | s, (\px_k; \px_{-k}), b \big) &=& \sum_{a \in \calA_k} x_{k,s,a} \pr\big(s' | s, (\ve_{k, s,a}; \px_{-k} ), b \big),
        \end{array}\right.
    \end{align}
    where $\ve_{k,s,a} \in \Delta(\calA_k)$ is such that its unique nonzero element corresponds to the action $a \in \calA_k$ of agent $k \in [n]$. An additional immediate consequence that will be useful in the sequel is the following property.
    \begin{gather}
        \left\{
        \begin{array}{lcl}
            \frac{\partial }{\partial x_{k,s,a}} r( s, \px, b ) &=& r\big(s, ( \ve_{k, s, a} ; \px_{-k} , b)\big); \text{ and} \\
            \frac{\partial }{\partial x_{k,s,a}} \pr( s' | s, \px, b ) &=& \pr \big(s' | s, ( \ve_{k, s, a} ; \px_{-k} ), b \big).
        \end{array}
        \right.
    \end{gather}
    Those multilinear (nonconvex-nonconcave) functions are part of the source of the complexity in our problem. We clarify that when all team players select a fixed strategy, \eqref{prog:nlp} retrieves the linear-programming formulation of the Bellman equation for the single-agent MDP~\citep{puterman2014markov}---as seen from the perspective of the adversary.
    
    Nevertheless, for our analysis it will be convenient to use a formulation that perturbs the objective function of \eqref{prog:nlp} with a quadratic term. In particular, let $\phi(\cdot) = \max_{\py \in \calY} V_{\vrho}( \cdot, \py)$ and $\hat{\px} \in \calX$ be a 
    point such that $\| \hat{\px} - \tilde{\px} \| \leq \epsilon$,
    where $\tilde{\px}\defeq \prox_{\phi/2\ell}(\hat{\px})$ is its proximal point; such a point $\hat{\vx}$ will be available after the termination of the first phase of \IPGmax, as implied by \Cref{lem:ipgmax-convergence-lemma}. Now the program we consider still has variables $(\px, \vv)$, but its objective function incorporates an additional quadratic term. This program was first introduced in \Cref{sec:extendibility-poly-time}, but we include it below for the convenience of the reader.
    
     \begin{subequations}
        \makeatletter
        \def\@currentlabel{\text{Q-}\mathrm{NLP}}
        \makeatother
        \renewcommand{\theequation}{$Q$\arabic{equation}}
        \begin{tagblock}[tagname={$\text{Q-}\mathrm{NLP}$},content={\label{prog:xxinlp}}]
        \begin{alignat}{3}
         & \mathrlap{ \min ~ \sum_{s \in \calS} \rho(s) v(s) + \ell
         \|\px - \hat{\px} \|^2} \label{eq:xiobj} \\
        \label{eq:xxibr}
        & \text{s.t.}~ & r(s, \px, b) + \gamma \sum_{s' \in \calS} \pr(s' | s, \px, b) v(s')  \leq v(s), & \quad  \forall  (s,b) \in \calS \times \calB;    \label{eq:xxibr} \\
         & & \px_{k,s}^\top \vone  = 1, & \quad \forall (k,s) \in [n] \times \calS; \text{ and}  \label{eq:xxieqone} 
         \\
         & & x_{k,s,a}  \geq 0, & \quad \forall k \in [n], (s,a) \in \calS \times \calA_k.
         \label{eq:xxinonneg}
        \end{alignat}
        \end{tagblock}
    \end{subequations}
    
 As we show in the following subsection, \eqref{prog:xinlp} attains a global minimum in the proximal point $\tilde{\px} \defeq \prox_{\phi/(2\ell)} (\hat{\px})$. First, let us point out that \eqref{prog:xinlp}---and subsequently \eqref{prog:nlp}---has nonempty feasibility set.
 
 \begin{lemma}
     The program \pref{prog:xinlp} is feasible.
     \label{lem:xinlp-feasible}
 \end{lemma}
 \begin{proof}
      Let $\vx \in \calX$ be any directly parameterized policy for the team and $\vv \defeq \frac{1}{1 - \gamma} \vone$, where recall that $\vone$ is the all-ones vector (with dimension $S$). Clearly, $\px_{k,s}^\top \vone = 1$, for all $(k,s) \in [n] \times \calS$, and $x_{k,s,a} \geq 0$ for all $k \in [n], (s,b) \in \calS \times \calB$. Further, for any $(s, b)\in \calS \times \calB$, we have
     \begin{align}
         r(s, \px, b)  + \gamma \sum_{s' \in \calS}\pr(s' | s, \px, b) \frac{1}{1-\gamma} = r(s, \px, b)  + \gamma  \frac{1}{1-\gamma}  \leq 1 + \gamma \frac{1}{1-\gamma}
         \leq \frac{1}{1-\gamma}.
     \end{align}
 \end{proof}

\subsection{The Global Minimum of \texorpdfstring{\pref{prog:xinlp}}{(Q-NLP)}}
\label{sec:nlp+phi}

Here we demonstrate that \pref{prog:xinlp} attains a global minimum under $\px = \tilde{\px} \defeq \prox_{\phi/(2\ell)} (\hat{\px})$. To do so, we first show that fixing $\vx$ yields a unique optimal value vector $\vv$ such that $\vrho^\top \vv = \phi(\px)$, where recall that $\phi$ is defined as $\phi(\cdot) = \max_{\py \in \calY} V_{\vrho}(\cdot, \py)$. Next, we prove that the objective function of \pref{prog:xinlp} is lower bounded by the minimum of the function
$\Psi(\vw) = \phi(\vw) + \ell \| \vw - \hat{\px} \|^2$; the latter function is $\ell$-strongly convex, which means that it has a unique minimizer, namely $\tilde{\px}\defeq \prox_{\phi/(2\ell)}(\hat{\px})$. In turn, this implies that the objective function of \eqref{prog:xinlp} is at least $\Psi(\px)$ for any fixed $\vx \in \calX$. Finally, we conclude the proof by showing that $\tilde{\vx}$ is part of a feasible solution of \eqref{prog:xinlp}.

First, we relate the optimal vector $\vv$ that arises by fixing $\px$ in \pref{prog:xinlp} and the function $\phi(\px)$:
 
\begin{proposition}
    Suppose that $\vrho \in \Delta(\calS)$ is full support. For any $\px\in\calX$ there exists a unique optimal vector $\vv^\star$ in \pref{prog:xinlp}. Further,
    \begin{equation}
        \vrho^\top \vv^\star = \phi({\px}).
    \end{equation}
    \label{prop:phi-and-nlp}
\end{proposition}

\begin{proof}
First, we observe that by fixing a feasible point $\vx \in \calX$ in \pref{prog:xinlp} we recover a linear program with variable $\vv\in \R^S$, which incidentally corresponds to the formulation of a single-agent MDP~\citep[Chapter 6]{puterman2014markov}. The reward function of this MDP is the expected reward of the adversary given that team plays $\px$, and the transition function is the expected transition function conditioned on the team playing $\px \in \calX$. Formally, we introduce this linear program below.

\begin{mini!}{}{\vrho^\top\vv}{}{}
\addConstraint{ {r}(s,\px, b)  + \gamma \sum_{s' \in \calS}{\pr}(s'|s, \px,  b)v(s') }{\leq v(s)}{,~ \forall (s,b) \in \calS \times \calB}.
\end{mini!}

We claim that the optimal solution $\vv^\star$ is unique for any given $\px \in \calX$. Indeed, this is a consequence of the fact that---when $\vrho$ is full-support---it is equivalent to the Bellman optimality equation, whose solutions can be in turn expressed as the fixed point of a contraction operator~\citep[Chapter 6.2~\&~6.4]{puterman2014markov}. Further, let us consider its dual linear program with variables $\vlambda \in \R^{S \times B}$:

\begin{maxi!}{}{\sum_{(s,b) \in \calS \times \calB} {r}(s, {\px},b) \lambda(s,b) }{}{}
\addConstraint{ \rho(\bar{s}) + \sum_{s \in \calS} \sum_{b \in \calB} \lambda(s,b)\gamma {\pr} (\bar{s}|s, {\px}, b)  - \sum_{b\in \calB}\lambda(\bar{s},b)}{= 0, }{~\forall \bar{s} \in \calS; \text{ and}}
\addConstraint{ \lambda(s,b)}{\geq 0, }{~\forall (s,b) \in \calS\times \calB}.
\end{maxi!}

The dual linear program is both feasible and bounded~\citep[Chapter 6.9]{puterman2014markov}. As such, it admits at least one optimal vector $\vlambda^\star$, with the additional property that $\sum_{b \in \calB} \lambda^\star(s,b) > 0$; the latter follows since $\vrho$ is full-support. Moreover, by~\citep[Theorem 6.9.1]{puterman2014markov}, we know that
\begin{itemize}
    \item[(i)] Any $\py\in\calY$ defines a feasible vector $\vlambda$ for the dual linear program; namely,
    \begin{equation}
        \lambda(s,b) = d^{\px, \py}_{\vrho}(s, b) \defeq \sum_{\bar{s}\in\calS} \rho(\bar{s}) \cdot \E_{ \py } \left[ \gamma^t \pr(s\step{t}=s, b\step{t}=b ~|~ \px, s\step{0} = \bar{s} ) \right].
    \end{equation}
    \item[(ii)] Any feasible vector of the dual linear program $\vlambda$ defines a feasible $\py\in\calY$; namely,
    \begin{equation}
        y_{s,b} \defeq \frac{\lambda(s,b)}{\sum_{b'\in\calB}\lambda(s,b')}, ~\forall (s,b) \in \calS \times \calB.
    \end{equation}
    Further, for any such $\py \in \calY$ it holds that $d^{\px, \py}_{\vrho}(s,b) = \lambda(s,b), ~\forall (s,b) \in \calS \times \calB$, where $d^{\px, \py}_{\vrho}(s,b)$ is the induced discounted state-action measure.
\end{itemize}

An implication of this theorem is a ``1--1'' correspondence between $\py\in\calY$ and feasible solutions $\vlambda$ of the dual program. Further, for a pair $(\vlambda, \py)$, the associated discounted state visitation measure is such that $d^{\px, \py}_{\vrho}(s) = \sum_{b\in\calB} \lambda(s,b), ~\forall s\in\calS$. Moreover, strong duality of linear programming implies that
\begin{equation}
    \vrho^\top\vv^\star = \sum_{(s,b) \in \calS\times\calB} \lambda^\star(s,b)r(s,\px,b) = \sum_{s\in\calS}d_{\vrho}^{\px,\py^\star}(s) r(s,\px,\py^\star).
\end{equation}
But, by \Cref{claim:visitation-and-value} we know that
\begin{equation}
    V_{\vrho}(\px,\py)  = \sum_{s\in\calS} d^{\px,\py}_{\vrho}(s) r(s,\px,\py).
\end{equation}
Thus, for an optimal pair $(\vlambda^\star, \py^\star)$, it holds that
\begin{equation}
    V_{\vrho}(\px,\py^\star) = \sum_{(s,b)\in\calS\times\calB}\lambda^\star(s,b) r(s,\px,b)  = \vrho^\top\vv^\star .
\end{equation}
Finally, the optimality of $\vlambda^\star$ in the dual program implies that for any correspondence pair $(\vlambda, \py)$,
\begin{align}
     \vrho^\top\vv^\star &=  \sum_{(s,b)\in\calS\times\calB}\lambda^\star(s,b) r(s,\px,b) \\ &\geq  \sum_{(s,b)\in\calS\times\calB}\lambda(s,b) r(s,\px,b)\\
     &= V_{\vrho}(\px, \py).
\end{align}
\end{proof}

\begin{lemma}
    \label{lemma:global}
    Let $\tilde{\px} \defeq \prox_{\phi/(2\ell)}(\hat{\px})$, and $\tilde{\vv}$ be the unique minimizer for \eqref{prog:xinlp} under a fixed $\px = \tilde{\vx}$. Then, $(\tilde{\px}, \tilde{\vv})$ is a global minimum of \eqref{prog:xinlp}.
\end{lemma}

\begin{proof}
    Consider a fixed $\vx \in \calX$. By \Cref{prop:phi-and-nlp}, we know that there is a unique optimal vector $\vv^\star$ in \eqref{prog:xinlp}, which also satisfies the equality
    \begin{equation}
            \vrho^\top\vv^\star = \max_{\py \in \calY} V_{\vrho}(\vx, \vy) = \phi(\vx).
            \label{eq:phi-nlp-eq}
    \end{equation}
    Now let us consider the function $\Psi(\vw)\defeq \phi(\vw) + \ell \| \vw - \hat{\px} \|^2$. $\Psi$ is $\ell$-strongly convex and its unique minimum value is attained at $\tilde{\px}\defeq \prox_{\phi/(2\ell)}(\hat{\px})$ (\Cref{cor:weak-conv}). By \eqref{eq:phi-nlp-eq}, it follows that for any feasible $(\px, \vv)$,
    $$\vrho^\top\vv + \ell \|\px - \hat{\px} \|^2 \geq \min_{\px \in \calX}\Psi(\px).$$ Finally, the value $\min_{\vx \in \calX} \Psi(\vx)$ is indeed attained by \eqref{prog:xinlp} when we set $\vx = \tilde{\vx}$, which is feasible for \eqref{prog:xinlp} (see \Cref{lem:xinlp-feasible} and \Cref{prop:phi-and-nlp}).
\end{proof}




\subsection{KKT Conditions for a Minimizer of \texorpdfstring{\pref{prog:xinlp}}{Q-NLP}}

 As we have shown in the previous subsection, $(\tilde{\px}, \tilde{\vv})$ is a minimum of the program
 \pref{prog:xinlp}. In this subsection, we leverage this fact to establish the existence of nonnegative Lagrange multipliers at $(\tilde{\vx}, \tilde{\vv})$ that satisfy the KKT conditions; this will be crucial for our extendibility argument. First, let us write the Lagrangian of the constrained minimization problem associated with~\eqref{prog:xinlp}:
 \begin{align}
     \calL \Big( (\px, \vv), (\vlambda, \vomega, \vpsi, \vzeta) \Big)& = 
     \vrho^\top \vv + \ell \|\px - \hat{\px} \|^2 + \sum_{(s,b)\in\calS\times\calB} \lambda(s,b) \left( r(s, \px, b) + \gamma\sum_{s' \in \calS} \pr(s'|s, \px, b)v(s') - v(s)  \right) \\
     &+ \sum_{(k,s)} \omega(k,s) \left( \px_{k,s}^\top \vone - 1 \right)
     + \sum_{(k,s)} \psi(k, s) \left( 1 - \px_{k,s}^\top \vone \right)
     + \sum_{(k,s,a)} \zeta(k,s,a) \left( -x_{k,s,a}\right), \label{eq:lag}
 \end{align}
 where 
 $$\{ \lambda(s, b) \}_{(s,b)} \cup \{ \omega(k, s) \}_{(k, s) } \cup \{ \psi(k,s) \}_{(k,s)} \cup \{ \zeta(k, s, a) \}_{(k,s,a)}$$
 are the associated Lagrange multipliers. Let us denote by $I$ set indexing the constraints of \eqref{prog:xinlp}. Before we proceed, we partition the set of constraints $I$ into $I = I_1 \cup I_2 \cup I_2' \cup I_3$, such that:

 %
    \begin{itemize}
        \item The constraints of \pref{eq:xibr}, corresponding to the subset of Lagrange multipliers $\{ \lambda(s,b)\}_{(s,b)}$, are assumed to lie in set $I_1$, so that every index $i \in I_1$ is uniquely associated with a pair $(s,b) \in \calS \times \calB$. In particular, for all $i \in I_1$, and the uniquely associated pair $(s,b) \in \calS \times \calB$, we let
            \begin{align}
                g_{i}( \px, \vv) \defeq r(s, \px, b) + \gamma \sum_{s' \in \calS} \pr(s' | s, \px, b) v(s') - v(s).
            \end{align}
            For any index $i \in I_1$, and the associated pair $(s, b) \in \calS \times \calB$, we have that
            \begin{itemize}
                \item For any $\bar{s} \in \calS$,
                    \begin{align}
                        \frac{\partial}{\partial v(\bar{s}) } g_{i}(\px, \vv) = 
                            \begin{cases}
                                \gamma \pr( \bar{s}|s, \px, b), & \text{if $\bar s \neq s$; and} \\
                                -1 + \gamma \pr( s|s, \px, b), & \text{if $\bar s = s$.} \\
                            \end{cases} ~ 
                            \label{eq:grad-constr-br-wrt-v}
                    \end{align}
                \item For any $\bar{k} \in [n], (\bar{s}, \bar{a}) \in \calS \times \calA_k$,
                    \begin{align}
                 \frac{\partial}{\partial x_{\bar{k},\bar s, \bar a}} g_{i}(\px, \vv) = 
                            \begin{cases}
                                0, & \text{if $ \bar s \neq s$; and} \\
                                r\big(s,  ( \ve_{\bar{k}, s, \bar{a}}; \px_{-\bar{k},s}) , b\big) + \gamma \sum_{s' \in \calS} \pr\big(s'|s, (\ve_{\bar{k}, s, \bar{a}}; \px_{-\bar{k}, s}) , b\big) v(s') & \text{if $\bar s = s$.} \\
                            \end{cases}                       
                    \end{align}
            \end{itemize}
        \item The constraints described by \eqref{eq:xieqone}, corresponding to the subset of Lagrange multipliers $\{ \omega(k, s) \}_{(k,s)} \cup \{\psi(k,s)\}_{(k,s)}$, are assumed to lie in the set $I_2 \cup I_2'$ as follows. Every equality constraint \eqref{eq:xieqone} is converted to a pair of inequality constraints corresponding to the sets $I_2$ and $I_2'$, respectively, so that every index $i \in I_2$ or $i \in I_2'$ is uniquely associated with a pair $(k,s) \in [n] \times \calS$. In particular, for all $i \in I_2$, and the associated pair $(k,s) \in [n] \times \calS$, we let
            \begin{align}
                g_{i}( \px, \vv) \defeq \px_{k,s}^\top \vone - 1,
            \end{align}
            and for all $i \in I_2'$
            \begin{align}
                g_{i}'( \px, \vv) \defeq 1 - \px_{k,s}^\top \vone.
            \end{align}   
            For any index $i \in I_2$ and the associated pair $(k, s) \in [n] \times \calS$, we have that
            \begin{itemize}
                \item For any $\bar{s} \in \calS$,
                    \begin{align}
                        \frac{\partial}{\partial v(\bar{s}) } g_{i}(\px,\vv) = 0.
                    \end{align}
                \item For any $\bar{k} \in [n], (\bar{s}, \bar{a}) \in \calS \times \calA$,
                    \begin{align}
                        \frac{\partial}{\partial x_{\bar k,\bar s, \bar a}}g_{i}(\px,\vv)=
                            \begin{cases}
                                1, & \text{ if $(k, s) = (\bar k , \bar s)$; and} \\
                                0, & \text{otherwise.}
                            \end{cases} 
                    \end{align}
            \end{itemize}
            For any index $i \in I_2'$ and the associated pair $(k,s) \in [n] \times \calS$, we have that
            \begin{itemize}
                \item For any $\bar{s} \in \calS$,
                    \begin{align}
                        \frac{\partial}{\partial v(\bar{s}) } g'_{i}(\px,\vv) = 0.
                    \end{align}
                \item For any $\bar{k} \in [n], (\bar{s}, \bar{a}) \in \calS \times \calA$,
                    \begin{align}
                        \frac{\partial}{\partial x_{\bar k,\bar s, \bar a}}g'_{i} (\px, \vv) = 
                            \begin{cases}
                                -1, & \text{ if $(k, s) = (\bar k , \bar s)$;} \\
                                0, & \text{otherwise.}
                            \end{cases}           
                    \end{align}
            \end{itemize}
    
        \item Finally, the constraints described by \eqref{eq:xinonneg}, corresponding to the subset of Lagrangian multipliers $\{\zeta(k, s, a)\}_{(k, s,a)}$, are assumed to lie in the set $I_3$, so that every index $i \in I_3$ is uniquely associated with a triple $(k, s, a)$. In particular, for each $i \in I_3$, and the associated triple $(k,s,a)$, we let
            \begin{align}
                g_{i}(\px, \vv) \defeq - x_{k,s,a}.
            \end{align}
            For any index $i \in I_3$ and the associated triple $(k,s, a)$, we have that
                        \begin{itemize}
                            \item For any $\bar{s} \in \calS$,
                                \begin{align}
                                    \frac{\partial}{\partial v(\bar{s}) } g_{i}(\px, \vv) = 0.
                                \end{align}
                            \item For any $\bar{k} \in [n], (\bar{s}, \bar{a}) \in \calS \times \calA$, \begin{align}
                                    \frac{\partial}{\partial x_{\bar k,\bar s, \bar a}}g_{i}(\px, \vv) =
                                        \begin{cases}
                                            -1, & \text{ if $(k, s, a) = (\bar k , \bar s, \bar a)$;} \\
                                            0, & \text{otherwise.}
                                        \end{cases}                        
                                \end{align}
                        \end{itemize}
    \end{itemize}

    We are now ready to determine the partial derivatives of the Lagrangian, as formalized below.
    \begin{claim}
    Consider the Lagrangian function $\calL$ of \eqref{prog:xinlp}, as introduced in \eqref{eq:lag}. Then, for any $\bar{s} \in \calS$, the partial derivative of $\calL$ with respect to $v(\bar{s})$ reads
    \begin{equation}
            \frac{\partial}{\partial v(\bar{s})}\calL = \rho(\bar s) + \sum_{s \in \calS} \sum_{b \in \calB}    \Big[  \lambda( s, b) \gamma  \pr\left( \bar s|s , \px, b \right) \Big]  - \sum_{b \in \calB} \lambda(\bar s, b).
            \label{eq:gradLwrtv}
    \end{equation}
    Further, for any $\bar{k} \in [n], (\bar{s}, \bar{a}) \in \calS \times \calA_k$,
    \begin{align}
        \frac{\partial}{\partial x_{\bar{k},\bar{s},\bar{a}} }\calL = 
            2\ell( x_{\bar{k},\bar{s},\bar{a}} - \hat{x}_{\bar{k},\bar{s},\bar{a}}) &+ \sum_{b \in \calB} \lambda(\bar{s}, b) \left[ r\left(\bar{s}, (\ve_{\bar{k},\bar{s},\bar{a}}; \vx_{-\bar{k}}) ,b\right) + \gamma \sum_{s \in \calS} \pr\left(s| \bar{s}, (\ve_{\bar{k},\bar{s},\bar{a}}; \vx_{-\bar{k}}), b \right) v(s) \right]
             \\ &+ \omega(\bar{k}, \bar{s}) - \psi(\bar{k}, \bar{s})  -\zeta(\bar{k}, \bar{s}, \bar{a}).
            \label{eq:gradLwrtx}
    \end{align}
    \end{claim}
    
    \begin{proof}
    
    Let us first establish \eqref{eq:gradLwrtv}. Fix any $\bar{s} \in \calS$. The partial derivative of the objective function of \eqref{prog:xinlp} with respect to $v(\bar{s})$ reads
        \begin{equation}
            \frac{\partial}{\partial v(\bar{s})} \Big( \vrho^\top \vv  + \ell \|\px - \hat{\px} \|_2^2 \Big) = \rho(\bar{s}).
        \end{equation}
        Further, \eqref{eq:xibr} is the only constraint that involves the variable $v(\bar{s})$, and we previously showed that for any $i \in I_1$,
        \begin{align}
            \frac{\partial}{\partial v(\bar{s}) } g_{i}(\px, \vv) = 
                \begin{cases}
                    \gamma \pr( \bar{s}|s, \px, b), & \text{if $\bar s \neq s$;} \\
                    -1 + \gamma \pr( s|s, \px, b), & \text{if $\bar s = s$,} \\
                \end{cases}
        \end{align}
        where $(s,b) \in \calS \times \calB$ is the pair associated with index $i \in I_1$. Thus,
        \begin{align}
            \sum_{(s,b) \in \calS \times \calB} \frac{\partial}{\partial v(\bar{s})} g_i(\Vec{x}, \Vec{v}) &= \sum_{b \in \calB} \sum_{s \neq \bar{s}} \left[ \lambda(s, b) \gamma \pr(\bar{s}|s, \px, b)\right]  + \sum_{b \in \calB} \lambda(\bar{s},b) \left( -1 + \pr(\bar{s}|\bar{s}, \px, b)\right) \\
            &= \sum_{b \in \calB} \sum_{s \in \calS} \left[ \lambda(s, b) \gamma \pr(\bar{s}|s, \px, b)\right]  - \sum_{b \in \calB} \lambda(\bar{s},b).
        \end{align}
        As a result, we conclude that
        \begin{equation}
          \frac{\partial}{\partial v(\bar{s})} \calL = \rho(\bar{s}) + \sum_{s \in \calS} \sum_{b \in \calB} \lambda(s,b) \gamma \pr(\bar{s}|s,\px,b) - \sum_{b \in \calB} \lambda(\bar{s},b),
        \end{equation}
        establishing \eqref{eq:gradLwrtv}. Next, we show \eqref{eq:gradLwrtx}. We first calculate the partial derivative of the objective function:
        \begin{equation}
            \label{eq:wrtx1}
            \frac{\partial}{\partial x_{\bar{k},\bar{s},\bar{a}}} \Big( \vrho^\top \vv  + \ell\|\px - \hat{\px} \|_2^2 \Big)  = 2\ell(x_{\bar{k},\bar{s},\bar{a}} - \hat{x}_{\bar{k},\bar{s},\bar{a}} ).
        \end{equation}
        Moreover, the summation of all the partial derivatives with respect to $x_{\bar{k},\bar{s},\bar{a}}$, for a fixed triple $(\bar{k},\bar{s}, \bar{a})$, of the constraints \eqref{eq:xibr}, \eqref{eq:xieqone}, and \eqref{eq:xinonneg}, multiplied by their respective Lagrange multipliers reads
        \begin{equation}
            \label{eq:wrtx2}
           \sum_{b \in \calB} \lambda(\bar{s},b) \Big( r\big(\bar{s}, (\ve_{\bar{k},\bar{s}, \bar{a}}; \px_{-\bar{k},s}), b\big) + \gamma \sum_{s \in \calS} \pr \big(s | \bar{s}, (\ve_{\bar{k},\bar{s}, \bar{a}}; \px_{-\bar{k},s}), b\big) v(s) \Big) + \omega(\bar{k},\bar{s}) - \psi(\bar{k},\bar{s}) - \zeta(\bar{k},\bar{s},\bar{a}).
        \end{equation}
        Combining \eqref{eq:wrtx1} and \eqref{eq:wrtx2} implies \eqref{eq:gradLwrtx}, concluding the proof.
    \end{proof}

\subsubsection{Local Optima Satisfy the KKT Conditions}
\label{sec:ahu-cq-sat}

Here we will show that for $(\tilde{\px}, \tilde{\vv}) \in \calX \times \R^S$, a global minimum of \eqref{prog:xinlp}, there exist (nonnegative) Lagrange multipliers that jointly satisfy the KKT conditions. We will first argue in \Cref{lemma:AHU} below that any feasible point of \eqref{prog:xinlp} satisfies the Arrow-Hurwicz-Uzawa constraint qualification. Then, we will leverage \Cref{cor:ahu-cq} to show that any local minimizer of \eqref{prog:xinlp}---and in particular $(\tilde{\vx}, \tilde{\vv})$---attains Lagrange multipliers that satisfy the KKT conditions. The following proof is analogous to~\citep[Ch. 4.4]{vrieze1987stochastic}.

\begin{lemma}
    \label{lemma:AHU}
    Let $(\vx, \vv) \in \calX \times \R^S$ be any feasible point of \eqref{prog:xinlp}. Then, the Arrow-Hurwicz-Uzawa constraint qualification is satisfied at $(\vx, \vv)$.
\end{lemma}

\begin{proof}
    Suppose that $A( \px, \vv) \subseteq I$ is the set of active constraints at a feasible point $(\px, \vv)$. Let us further denote by $d$ the dimension of $(\vx, \vv)$. To apply \Cref{theorem:AHU}, we have to establish the existence of a vector $\vw \in \R^d$, such that for any $i \in A( \px, \vv)$,
    \begin{align}
    \left\{\begin{array}{cl}
         \vw ^\top \nabla_{(\px,  \vv)} g_i(\px, \vv) > 0, & \text{if $g_i$ is nonconcave; and}  \\  
         \vw ^\top \nabla_{(\px,  \vv)} g_i(\px, \vv) \geq 0, & \text{if $g_i$ concave.} 
    \end{array}
    \right.
    \end{align}
    
    For convenience, we will index the entries of $\vw$ so that $\vw = ( \vw_{\vx} , \vw_{\vv} )$. For reasons that will shortly become clear, we set $\vw_{\vx} = \Vec{0}$. Now consider any active constraint $i$ (if any exists) from the set $I_2 \cup I_2' \cup I_3$. The corresponding constraint function $g_i$ is affine, and in particular concave. Further, it holds that for any $s \in \calS$,
    \begin{equation}
        \frac{\partial}{\partial v(s)} g_i(\Vec{x}, \vv) = 0.
    \end{equation}
    As a result, for our choice of vector $\vw = (\Vec{0}, \vw_{\vv})$, it immediately follows that 
    \begin{equation}
        \vw^\top \nabla_{(\px,  \vv)} g_i(\px, \vv) = 0,
    \end{equation}
    for any $i \in I_2 \cup I_2' \cup I_3$. Let us now treat (if any) active constraints $i \in I_1$. In particular, let $(s,b) \in \calS \times \calB$ be the pair associated with $i$, so that
    \begin{equation}
        g_i( \px, \vv) = r(s, \px, b ) + \gamma \sum_{s' \in \calS} \pr(s'|s, \px, b) v(s') - v(s).
    \end{equation}
    Then, 
    \begin{align}
        \vw ^\top   \nabla g_i (\px, \vv) &= \vw_{\vv}^\top \left.\nabla_{\vv} \Big[  r(s, \px, b) + \gamma \sum_{s' \in \calS} \pr(s'| s, \px, b ) \vv(s') - \vv(s) \Big] \right|_{(\px, \vv)} \\
        &= 
         \sum_{\bar{s} \neq s} w_{{v}(\bar{s})} \gamma \pr( \bar{s} | s, \px, b) + w_{{v}(s)}\Big(                                 -1 + \gamma \pr( s| s, \px, b) \Big) \\
         &= \sum_{\bar{s} \in \calS} w_{{v}(\bar{s})} \gamma \pr( \bar{s} |s, \px, b) - w_{v(s)}.
    \end{align}
    By virtue of \Cref{theorem:AHU}, it suffices to show that there exists $\vw_{\vv}$ so that for any $(s,b) \in \calS \times \calB$,
    \begin{equation}
        \gamma  \sum_{s' \in \calS} w_{{v}(s')} \pr( s' |s, \tilde{\px}, b) - w_{v(s)} > 0.
    \end{equation}
    We will show that this property holds for $\vw_{\vv} \defeq - \vv$. Indeed, since $(\vx, \vv)$ is feasible, we get that
    \begin{equation}
        \gamma  \sum_{s' \in \calS} w_{{v}(s')} \pr( s' |s, \tilde{\px}, b) - w_{v(s)} = - \gamma  \sum_{s' \in \calS} v(s') \pr( s' |s, \tilde{\px}, b) + v(s) \geq r(s, \vx, b) > 0,
    \end{equation}
    since we have assumed that $r(s, \Vec{a}, b) > 0$ for any $(\Vec{a}, b) \in \calA \times \calB$. This concludes the proof.
\end{proof}
    
Next, leveraging this lemma and \Cref{cor:ahu-cq}, we conclude that $(\tilde{\vx}, \tilde{\vv})$---in fact, any local minimum of \eqref{prog:xinlp}---attains nonnegative Lagrange multipliers that satisfy the KKT conditions.

\begin{corollary}
    \label{corollary:KKT-sat}
For any local minimum $(\tilde{\vx}, \tilde{\vv}) \in \calX \times \R^S$ of \eqref{prog:xinlp}, there exists (nonnegative) Lagrange multipliers satisfying the KKT conditions.
\end{corollary}

In particular, by the first-order stationarity condition and the complementary slackness condition (recall \Cref{def:KKT}) with respect to $(\tilde{\vx}, \tilde{\vv})$, we have

  \begin{subequations}
        \label{eq:kkt}
    {
    \begin{align}
        \def\arraystretch{4}
        \displaystyle
      &
       \nabla_{(\px, \vv)} 
      \calL \Big( (\tilde{\px}, \tilde{\vv}), (\tilde{\vlambda}, \tilde{\vomega}, \tilde{\vpsi}, \tilde{\vzeta}) \Big)= \vzero; \label{eq:fos}
        \\
        &
        \begin{array}{l}
         \tlambda(s,b) \Big( \textstyle r(s, \tilde{\px}, b) + \gamma \sum_{s'\in\calS} \pr(s' | s, \tilde{\px}, b) \tilde{v}(s') - \tilde{v}(s) \Big)  =0, 
         \quad \forall (s,b) \in \calS \times \calB;
         \\
        \displaystyle
         \tomega(k,s) \Big( \tilde{\px}_{k,s}^\top \vone - 1 \Big)   = 0,
         \quad \forall (k,s) \in [n] \times \calS;
         \\
        \displaystyle
         \tpsi(k,s) \Big( 1 - \tilde{\vx}_{k,s}^\top \vone \Big)  = 0,
         \quad \forall (k,s)  \in [n] \times \calS;
         \\
        \displaystyle
         \tzeta(k,s,a)\Big( -  \tilde{x}_{k,s,a}\Big)   = 0, 
         \quad \forall k \in [n], \forall (s, a) \in \calS \times \calA_k; \text{ and}
         \end{array}
         \label{eq:kkt-cs}
         \\
         &
         \tomega(k,s), \tpsi(k,s), \tzeta(k,s,a) \geq 0, ~\forall (k,s)  \in [n] \times \calS, ~ \text{and} ~ \forall k \in [n], (s, a) \in \calS \times \calA_k. \label{eq:kkt-nonneg}
        \end{align}
    }
    \end{subequations}

\subsubsection{Connecting the Lagrange Multipliers with the Visitation Measure}

Here we establish an important connection between a subset of the Lagrange multipliers and the \emph{visitation measure} under a specific policy of the adversary. This fact will be crucial later in the proof of \Cref{lem:lpadv-nash} for controlling the approximation error.

\begin{proposition}
    \label{proposition:Lan-vis}
    Suppose that the initial distribution $\vrho$ is full support. Let also $\tilde{\vlambda} \in \R_{\geq 0}^{S\times B}$ be the associated vector of Lagrange multipliers at $(\tilde{\px}, \tilde{\vv}) \in \calX \times \R^{S}$ that satisfy~\eqref{eq:kkt}. Then, it holds that $\sum_{b \in \calB} \tlambda(s,b) > 0$, for any $s \in \calS$. Further, if
        $$\tilde{y}_{s,b} \defeq \frac{\tlambda(s,b)}{\sum_{b' \in \calB} \tlambda(s,b')},$$
        for any $(s,b) \in \calS \times \calB,$ then it holds that
    \begin{equation}
        \sum_{b \in \calB} \tlambda(s,b) = d^{\tilde{\px},\tilde{\py}}_{\vrho} (s), \quad \forall s \in \calS,
    \end{equation}
    where $d^{\tilde{\px},\tilde{\py}}_{\vrho} (s)$ defines the visitation measure at state $s \in \calS$ induced by $(\tilde \vx, \tilde \vy)$.
\end{proposition}
\begin{proof}
    First of all, it follows directly from \eqref{eq:gradLwrtv} and the fact that the Langrange multipliers are nonnegative that $\sum_{b \in \calB} \tlambda(s,b) > 0$. Next, for convenience, let us define a vector $\vd \in \R_{> 0}^{S}$ such that
    \begin{equation}
        \label{eq:def-d}
        d(s) = \sum_{b \in \calB} \tlambda(s,b),
    \end{equation}
    for all $s \in \calS$. Then, starting from \eqref{eq:gradLwrtv}, we have that for any $\bar s \in \calS$,
    \begin{align}
        \rho(\bar s) + \sum_{s \in \calS} \sum_{b \in \calB} \left[ \frac{d(s)}{d(s)} \tlambda( s, b)   \gamma  \pr\left( \bar s|s , \tilde{\px},b \right) \right]  - d(\bar s) = 0 \label{eq:ddef} \\
        \rho(\bar s) + \sum_{s \in \calS} \sum_{b \in \calB} \left[ \frac{d(s)}{\sum_{b' \in \calB}\tlambda(s,b')} \tlambda( s, b)   \gamma  \pr\left( \bar s|s , \tilde{\px},b \right) \right]  - d(\bar s) = 0 \\
        \rho(\bar s) + \sum_{s \in \calS} \sum_{b \in \calB} \left[ d(s) \frac{\tlambda( s, b)}{\sum_{b' \in \calB} \tlambda(s,b')}    \gamma  \pr\left( \bar s|s , \tilde{\px},b \right) \right]  - d(\bar s) = 0 \\
        \rho(\bar s) + \gamma \sum_{s \in \calS} \sum_{b \in \calB} \Big[ d(s) \tilde{y}_{s,b}  \pr\left( \bar s|s , \tilde{\px},b \right) \Big]  - d(\bar s) = 0 \label{eq:ydef}
        \\
        \rho(\bar s) + \gamma \sum_{s \in \calS} \Big[ d(s) \pr \left( \bar s|s , \tilde{\px}, \tilde{\py} \right) \Big]  - d(\bar s) = 0, \label{eq:lan-vis-final}
    \end{align}
    where \eqref{eq:ddef} uses the definition of $\Vec{d}$ given in \eqref{eq:def-d}; \eqref{eq:ydef} follows from the definition of strategy $\Vec{y}$ in the statement of the proposition; and \eqref{eq:lan-vis-final} is derived since $\mathbb{P}(\bar{s} | s, \tilde{\vx}, \tilde{\vy} ) = \sum_{b \in \calB} \tilde{y}_{s, b} \mathbb{P}(\bar{s} | s, \tilde{\vx}, b)$ (law of total probability). Next, we observe that \eqref{eq:lan-vis-final} can be compactly expressed as $\Vec{\rho}^\top = \Vec{d}^\top \left( \mat{I} - \gamma \pr (\tilde{\px}, \tilde{\py}) \right)$ (recall the definition of matrix $\pr$), in turn implying that
    \begin{equation}
        \vd^\top =  \vrho^\top \left( \mat{I} - \gamma \pr(\tilde{\px}, \tilde{\py}) \right) ^{-1}.
    \end{equation}
    We note that $\left( \mat{I} - \gamma \pr (\tilde{\px}, \tilde{\py}) \right)$ is invertible (\Cref{claim:iminusp-invertible}). As a result, by virtue of \Cref{fact:visitation} we conclude that $\sum_{b \in \calB} \tlambda(s,b) = d^{\tilde{\px},\tilde{\py}}_{\vrho} (s)$, for all $s \in \calS$. This concludes the proof.
\end{proof}

    We also provide an additional auxiliary claim that will be useful in the sequel. The proof follows by carefully leveraging the KKT conditions, as we formalize below.
    \begin{claim}
        Let $(\tilde{\px}, \tilde{\vv}) \in \calX \times \R^S$ be a local optimum of the \eqref{prog:xinlp}, and $\{\tlambda(s,b)\},\{\tpsi(k,s)\},\{\tomega(k,s)\}$ be the associated Lagrange multipliers defined in \pref{eq:kkt}. Then, for any player $k \in [n]$,
        \begin{equation}
        \tilde {v}(s) 
        - \frac{  2 \ell ( \tilde{\px}_{k,s} - \hat{\px}_{k,s} )^\top \tilde{\px}_{k,s} }{\sum_{b \in \calB} \tlambda(s, b)} 
        =  \frac{  \tpsi(k, s) - \tomega(k, s)   }{\sum_{b \in \calB} \tlambda(s, b) }
        , \quad
        \forall s \in \calS.
        \label{eq:keyrelation}
        \end{equation}
    \label{claim:important-relation}
    \end{claim}
    \begin{proof}
    First, multiplying \Cref{eq:gradLwrtx} by $\tilde{x}_{k,s,a}$ we get that
    \begin{align}
        - 2 \ell( \tilde{x}_{k,s,a} - \hat{x}_{k,s,a} ) \tilde{x}_{k,s,a} &+
        \tilde{x}_{k,s,a}\sum_{b \in \calB} \tlambda(s, b) \Big[  r\big( s, (\ve_{k,s,a}; \tilde{\px}_{-k} ),b\big)
        + \gamma \sum_{s' \in \calS} \pr\big(s'| s, (\ve_{k,s,a}; \tilde{\px}_{-k} ), b\big) v(s') \Big]  \\ 
        &+ \tilde{x}_{k,s,a}\Big( \tomega(k, s) - \tpsi(k, s)\Big)  - \tilde{x}_{k,s,a} \tzeta(k, s, a) = 0, \quad \forall k \in [n], (s,a) \in \calS \times \calA.
    \end{align}
    By complementary slackness, it follows that $- \tilde{x}_{k,s,a} \tzeta(k, s, a) = 0$, for all $k \in [n], (s, a) \in \calS \times \calA_k$. Thus, the previously displayed equation can be simplified as
    \begin{align}
        - 2 \ell( \tilde{x}_{k,s,a} - \hat{x}_{k,s,a} ) \tilde{x}_{k,s,a} &+
        \tilde{x}_{k,s,a}\sum_{b \in \calB} \tlambda(s, b)\Big[  r\big( s, (\ve_{k,s,a} ; \tilde{\px}_{-k} ),b\big)
        + \gamma \sum_{s' \in \calS} \pr\big(s'| s, (\ve_{k,s,a} ; \tilde{\px}_{-k} ), b\big) v(s') \Big]  \\ 
        &+ \tilde{x}_{k,s,a}\Big( \tomega(k, s) - \tpsi(k, s)\Big) = 0, \quad \forall k \in [n], (s,a) \in \calS \times \calA_k.
    \end{align}
    Next, summing the previous equation over all $a \in \calA_k$ it follows that for any $(k, s) \in [n] \times \calS$,
    \begin{align}
        \sum_{a \in \calA_k} \tilde{x}_{k,s,a}\sum_{b \in \calB} \tlambda(s, b) \Big[  r\big( s, (\ve_{k,s,a} ; \tilde{\px}_{-k} ),b\big) 
        + \gamma \sum_{s' \in \calS}\pr\big(s'| s, (\ve_{k,s,a} ; \tilde{\px}_{-k} ), b\big) \tilde{v}(s') \Big] \\ 
        - 2 \ell \sum_{a \in \calA_k}( \tilde{x}_{k,s,a} - \hat{x}_{k,s,a} ) \tilde{x}_{k,s,a} + \sum_{a \in \calA_k} \tilde{x}_{k,s,a}\Big( \tomega(k, s) - \tpsi(k, s)\Big) = 0 \\
        \sum_{b \in \calB} \tlambda(s, b) \sum_{a \in \calA_k} \tilde{x}_{k,s,a} \Big[  r\big( s, (\ve_{k,s,a} ; \tilde{\px}_{-k} ),b\big)  
        + \gamma \sum_{s' \in \calS} \pr\big(s'| s, (\ve_{k,s,a} ; \tilde{\px}_{-k} ), b\big) \tilde{v}(s') \Big] \\
         - 2 \ell ( \tilde{\px}_{k,s} - \hat{\px}_{k,s} )^\top \tilde{\px}_{k,s} + \Big( \tomega(k, s) - \tpsi(k, s)\Big) = 0,
    \end{align}
    where the last derivation uses that $\sum_{a \in \calA_k} \tilde{x}_{k,s,a} = 1$ since $\vx_{k, s} \in \Delta(\calA_k)$. Further, using that 
    \begin{itemize}
        \item[(i)] $\sum_{a \in \calA_k} \tilde{x}_{k,s,a} r\big( s, (\ve_{k,s,a} ; \tilde{\px}_{-k} ),b \big) = r\big( s, \tilde{\px}, b \big)$, and 
        \item[(ii)] $\sum_{a \in \calA_k} \tilde{x}_{k,s,a} \pr\big(s'| s, (\ve_{k,s,a} ; \tilde{\px}_{-k} ), b\big) = \pr\big(s'| s, \tilde{\px}, b \big)$,
    \end{itemize}
    it follows that for any $(k, s) \in [n] \times \calS$,
    \begin{equation}
        \sum_{b \in \calB} \tlambda(s, b)\Big[  r\big( s, \tilde{\px}, b \big) 
        + \gamma \sum_{s' \in \calS} \pr\big(s'| s, \tilde{\px}, b \big) \tilde{v}(s') \Big]  
        - 2 \ell ( \tilde{\px}_{k,s} - \hat{\px}_{k,s} )^\top \tilde{\px}_{k,s} + \Big( \tomega(k, s) - \tpsi(k, s)\Big) = 0. \label{eq:pr1}
    \end{equation}
    Further, we know from the complementary slackness condition~\eqref{eq:kkt} that for any $(s, b) \in \calS \times \calB$,
    $$ \tlambda(s,b) \left(  r\big(s, \tilde{\vx}_s ,b\big) + \gamma \sum_{s' \in \calS} \pr\big(s' | s, \tilde{\vx}_s ,b\big) \tilde{v}(s') - \tilde{v}(s) \right) =0.$$ 
    In turn, summing over all actions $b \in \calB$ we get that for any $s \in \calS$,
    \begin{equation}
        \tilde{v}(s) \sum_{b \in \calB} \tlambda(s,b) = \sum_{b \in \calB} \tlambda(s, b)\Big[  r\big( s, \tilde{\px}, b \big) 
        + \gamma \sum_{s' \in \calS} \pr\big(s'| s, \tilde{\px}, b \big) \tilde{v}(s') \Big].
    \end{equation}
    Combining this equation with \eqref{eq:pr1}, and recalling that $\sum_{b \in \calB} \tlambda(s,b) > 0$ for any $s \in \calS$ (by \Cref{proposition:Lan-vis}), leads to the desired conclusion.
    \end{proof}

    \subsection{
    Efficient Extension to Nash Equilibria}
    \label{sec:fosp-extend-to-ne}
    
    This subsection completes the proof that an $\epsilon$-near stationary point $\hat{\vx}$ of $\phi$ can be extended to a strategy profile $(\hat{\vx}, \hat{\vy})$ that is an $O(\epsilon)$-approximate Nash equilibrium. Further, we provide a computationally efficient way for computing $\hat{\vy}$ based on an appropriate linear program, \eqref{prog:lpadv} introduced below. The upshot is that feasible solutions of \eqref{prog:lpadv} induce the appropriate strategy for the adversary $\hat{\vy} \in \calY$. In this context, we are ready to introduce \eqref{prog:lpadv}, a linear program with free variables $\vlambda \in \R^{S \times B}$:
    
   \begin{subequations}
        \makeatletter
        \def\@currentlabel{\mathrm{LP}_\text{adv}}
        \makeatother
        \renewcommand{\theequation}{$\mathrm{LP}_\text{adv}$.\arabic{equation}}
        \begin{tagblock}[tagname={$\mathrm{LP}_\text{adv}$},content={\label{prog:lpadv}},boxsep=-2.6em]{
        \arraycolsep=1.4pt\def\arraystretch{1.4}
       \begin{alignat}{4}
            &\mathrlap{ \max ~ \sum_{(s,b) \in \calS \times \calB} \lambda(s,b)r\left(s,  \hat{\px}, b \right)} \\
            &\text{s.t.}~ &
            \displaystyle
            \begin{array}{lr}
            {\sum_{b} \lambda(s, b)
                    \left[ r\left(s, (\ve_{k,s,a} ; \hat{\px}_{-k}), b\right) 
                    +
                    \gamma \sum_{s'} \pr\left(s'|s, (\ve_{k,s,a} ; \hat{\px}_{-k}), b \right)  \hat{v}(s')-\hat{v}(s)\right] }
                    &{\geq -c_1 \cdot \epsilon },\\
                    \multicolumn{2}{r}{ \quad \forall s \in \calS};
            \end{array}
            \label{eq:pertconstr2}
            \\
            & &
            \displaystyle
            {\lambda(s, b)   \left(
            \left[ r\left(s, \hat{\px}, b \right) 
                    + \gamma \sum_{s' \in \calS}\pr\left(s'|s,  \hat{\px}, b \right)  \hat{v}(s') \right] - \hat{v}(s) \right) }{\leq c_2\cdot\epsilon}, 
                    \label{eq:pertconstr3}
                     ~ \forall (s,b) \in \calS \times \calB;
                    \\
            \displaystyle
            & &{\lambda(s, b)   \left(
            \left[ r\left(s, \hat{\px}, b \right) 
                    + \gamma \sum_{s' \in \calS} \pr\left(s'|s,  \hat{\px}, b \right)  \hat{v}(s') \right] - \hat{v}(s) \right) }{\geq - c_2\cdot\epsilon}, 
                    \label{eq:pertconstr4} 
                    ~ \forall (s,b) \in \calS \times \calB;
                    \\
            \displaystyle
            &  &\sum_{b \in \calB} \lambda(s,b) \geq \rho(s), ~ \forall s \in \calS; \text{ and} \label{eq:pertconstr5}
                    \\
            \displaystyle
            &  &\sum_{b \in \calB} \lambda(s,b) \leq \frac{1}{1-\gamma}, ~ \forall s \in \calS. \label{eq:pertconstr6}
        \end{alignat}
        }
        \end{tagblock}
    \end{subequations}
    Here,
    \begin{align}
        &c_2 \defeq \frac{1}{1-\gamma} \left( {\rewlip}  + \gamma S {\matlip} \frac{1}{1-\gamma}  + \gamma S L  + L  \right), \\
        &c_1 \defeq 4 \ell + c_2.
    \end{align}
    Before we proceed, a few remarks are in order. First, let us relate \eqref{prog:lpadv} with \eqref{prog:xinlp}. As alluded to by our notation, the free variables of \eqref{prog:lpadv} are related to a subset of the Lagrange multipliers introduced in \eqref{eq:lag}. In light of this, \eqref{eq:pertconstr3} and \eqref{eq:pertconstr4} are related to the complementary slackness condition given in \eqref{eq:kkt-cs}, while \eqref{eq:pertconstr2} is related to the first-order stationary condition~\eqref{eq:fos}. An important point is that we previously established the KKT conditions only with respect to the pair $(\tilde{\vx}, \tilde{\vv})$, instead of $(\hat{\vx}, \hat{\vv})$. This partially explains the ``slackness'' we introduced in \eqref{eq:pertconstr2}, \eqref{eq:pertconstr3} and \eqref{eq:pertconstr4}. Correspondingly, the slackness parameters $c_1$ and $c_2$ were introduced to ``transfer'' the constraints from $(\tilde{\vx}, \tilde{\vv})$ to $(\hat{\vx}, \hat{\vv})$, in a sense that will become clear in the sequel. We stress that expressing \eqref{prog:lpadv} in terms of $(\hat{\vx}, \hat{\vv})$ is crucial since $(\tilde{\vx}, \tilde{\vv})$ is not actually available to the algorithm. We also remark that the objective function of \eqref{prog:lpadv} is not relevant for out argument; even a constant objective would suffice for our purposes.
    
    But first, we need to show that \eqref{prog:lpadv} is feasible. To do so, we construct an auxiliary linear program that, unlike \eqref{prog:lpadv}, depends on $(\tilde{\vx}, \Tilde{\vv})$, an \emph{exact} minimum of \eqref{prog:xinlp}. As such, the feasibility of this program, \eqref{prog:strictlpadv}, is established using the Lagrange multipliers $\tilde{\vlambda} \in \R^{S \times B}$ associated with $(\tilde{\px}, \tilde{\vv})$.
    
    \begin{lemma}
       The linear program \pref{prog:lpadv} with variables $\vlambda \in \R^{S\times B}$ is feasible.
       \label{lem:lpadv-feasible}
    \end{lemma}
    \begin{proof}
        We introduce the following auxiliary linear program with variables $\vlambda\in \R^{S\times B}$:

            \begin{subequations}
        \makeatletter
        \def\@currentlabel{\mathrm{LP}'_\text{adv}}
        \makeatother
        \renewcommand{\theequation}{$\mathrm{LP}'_\text{adv}$.\arabic{equation}}
        \begin{tagblock}[tagname={$\mathrm{LP}'_\text{adv}$},content={\label{prog:strictlpadv}},boxsep=-1.2em]
        {
        \arraycolsep=1.4pt\def\arraystretch{1.4}
        \begin{alignat}{4}
            & \mathrlap{ \max ~ \sum_{(s,b) \in \calS \times \calB} \lambda(s,b) r\left(s,  \tilde{\px}, b \right)} \\
            &\text{s.t.}
            ~
            &   
            \begin{array}{lr}
            {\rho(\bar{s}) + \sum_{s \in \calS} \sum_{b \in \calB} \Big[  \lambda( s, b) \gamma  \pr\left( \bar s|s , \tilde{\px},b \right) \Big]  - \sum_{b \in \calB} \lambda(\bar s, b) } & = 0, 
            \\
            \multicolumn{2}{r}{ \qquad \forall (s,b) \in \calS \times \calB};
            \end{array}
            \label{eq:constr1}
            \\
            &  & \begin{array}{lr}
            {       \sum_{b} \lambda(s, b)
                    \left[ r\left(s, (\ve_{k,s,a} ; \tilde{\px}_{-k}), b \right)
                    +
                    \gamma \sum_{s'} \pr\left(s'|s, (\ve_{k,s,a} ; \tilde{\px}_{-k}), b \right)  \tilde{v}(s')-\tilde{v}(s)\right] } & {\geq -4 \epsilon \ell }, 
            \\
            \multicolumn{2}{r}{\forall k \in [n], \forall (s,a) \in \calS \times \calA_k;}
            \end{array}
        \label{eq:constr2}
            \\
            \scriptstyle
            &  &{\lambda(s, b)   \left(
            \left[ r\left(s, \tilde{\px} , b\right) 
                    + \gamma \sum_{s' \in \calS} \pr\left(s'|s,  \tilde{\px}, b \right)  \tilde{v}(s') \right] - \tilde{v}(s) \right) }{ = 0 },  
                    \label{eq:constr3}~
            \forall (s,b) \in \calS \times \calB; 
            \\
            \scriptstyle
        & &{\lambda(s,b)}{\geq 0 }{}, ~ \forall (s,b) \in \calS \times \calB.
        \label{eq:constr4}
        \end{alignat}
        }
        \end{tagblock}
    \end{subequations}
    
        
    Again, the objective function of \eqref{prog:strictlpadv} is not relevant for our argument. For our purposes, it suffices to show that \eqref{prog:strictlpadv} is feasible.
    \begin{lemma}
        \label{lemma:idealfeas}
        Let $\tilde{\vlambda} \in \R^{S \times B}$ be a subset of Lagrange multipliers associated with $(\tilde{\Vec{x}}, \tilde{\vv}) \in \calX \times \R^S$ of \eqref{prog:xinlp}. Then, $\tilde{\vlambda}$ satisfies all the constraints of \pref{prog:strictlpadv}.
    \end{lemma}
    \begin{proof}
    First, \eqref{eq:constr1} is satisfied by the first-order stationarity condition \pref{eq:fos}; \eqref{eq:constr3} is satisfied by the complementary slackness condition \pref{eq:kkt-cs}; and \pref{eq:constr4} by the nonnegative of the Lagrange multipliers \pref{eq:kkt-nonneg}. The rest of the proof is devoted to showing that $\tilde{\vlambda}$ also satisfies \eqref{eq:constr2}. To this end, we first recall that, by \Cref{claim:important-relation}, we have that
    \begin{equation}
        \tilde{\omega}(k, s) - \tilde{\psi}(k, s)
        \label{eq:lag-val} = - \tilde {v}(s) \sum_{b \in \calB}  \tilde{\lambda}(s, b)
        + 2 \ell ( \tilde{\px}_{k,s} - \hat{\px}_{k,s} )^\top \tilde{\px}_{k,s},
    \end{equation}
for any $s \in \calS$. Combing this relation with \pref{eq:gradLwrtx} we get that for any $k \in [n], (s,a) \in \calS \times \calA_k$,
    \begin{align}
        \sum_{b \in \calB} \tilde{\lambda}(s, b)\Big[  r\left( s, (\ve_{k,s,a} ; \tilde{\px}_{-k, s} ), b\right) + \sum_{s' \in \calS} \pr\left(s'| s, (\ve_{k,s,a} ; \tilde{\px}_{-k, s} ), b\right) v(s') \Big] + 2\ell(\tilde{x}_{k,s,a} - \hat{x}_{k,s,a} ) \\
         - \tilde{v}(s) \sum_{b \in \calB} \tilde{\lambda}(s, b) + 2\ell\tilde{\px}_{k,s}^\top ( \tilde{\px}_{k,s} - \hat{\px}_{k,s} ) 
        -\tilde{\zeta}(k, s, a)
        & = 0  
        \\
        \sum_{b \in \calB} \tilde{\lambda}(s, b)\Big[  r\left( s, (\ve_{k,s,a} ; \tilde{\px}_{-k, s} ), b\right) + \sum_{s' \in \calS} \pr\left(s'| s, (\ve_{k,s,a} ; \tilde{\px}_{-k, s} ), b\right) v(s') \Big] - \sum_{b \in \calB} \tilde{\lambda}(s, b) \tilde{v}(s) \\ 
        + 2\ell(\tilde{x}_{k,s,a} - \hat{x}_{k,s,a} ) - 2\ell\tilde{\px}_{k,s}^\top ( \tilde{\px}_{k,s} - \hat{\px}_{k,s} )
        & 
        = \tilde{\zeta}(k, s, a).
    \end{align}
    As a result, we conclude that
    \begin{align}
        \sum_{b \in \calB} \tilde{\lambda}(s, b)\Big[  r\left( s, (\ve_{k,s,a} ; \tilde{\px}_{-k, s} ), b\right) + \sum_{s' \in \calS }\pr\left(s'| s, (\ve_{k,s,a} ; \tilde{\px}_{-k, s} ), b\right) v(s')  - \tilde{v}(s) \Big] 
        \geq - 4 \epsilon \ell,
    \end{align}
    since 
    \begin{itemize}
        \item[(i)] $\tilde{\zeta}(k, s, a) \geq 0$ by \Cref{eq:kkt-nonneg};
        \item[(ii)] $2\ell (\hat{x}_{k,s,a} - \tilde{x}_{k,s,a}) \geq -2 \ell | \hat{x}_{k,s,a} - \tilde{x}_{k,s,a}| \geq - 2 \ell \epsilon$ given that $\| \tilde{\vx} - \hat{\vx} \|_\infty \leq \| \tilde{\vx} - \hat{\vx} \|_2 \leq \epsilon$; and 
        \item[(iii)] $2\ell\tilde{\px}_{k,s}^\top ( \tilde{\px}_{k,s} - \hat{\px}_{k,s}) \geq - 2 \ell \| \tilde{\px}_{k,s}\|_2 \| \tilde{\px}_{k,s} - \hat{\px}_{k,s} \|_2 \geq -2\ell \epsilon$, by Cauchy-Schwarz inequality and the fact that $\| \tilde{\px}_{k,s}\|_2 \leq 1$ since $\tilde{\px}_{k,s} \in \Delta(\calA_k)$.
    \end{itemize}
     This concludes the proof of the lemma.
    \end{proof}
    We next leverage this lemma to establish that the original linear program is also feasible. To do so, we will leverage the Lipschitz continuity of the constraint functions. In particular, consider any $(s, b) \in \calS \times \calB$. We observe that
    \begin{align}
     r\left(s, \tilde{\px} , b\right) + \gamma \sum_{s' \in \calS}  \pr\left(s'|s,  \tilde{\px}, b \right)\tilde{v}(s') - \tilde{v}(s) =& r\left(s, \tilde{\px} , b\right)+ r\left(s, \hat{\px}, b \right) - r\left(s, \hat{\px}, b \right) & \nonumber \\
        +& 
        \gamma \sum_{s' \in \calS} \Big( \pr\left(s'|s,  \tilde{\px}, b \right)+ \pr\left(s'|s,  \hat{\px}, b \right) -\pr\left(s'|s,  \hat{\px}, b \right) \Big)
            \Big( \tilde{v}(s') + \hat{v}(s') - \hat{v}(s') \Big)  \nonumber \\
        -& \tilde{v}(s) + \hat{v}(s) - \hat{v}(s).
    \end{align}
    Thus,
    \begin{align}
    r\left(s, \tilde{\px} , b\right) + \gamma \sum_{s' \in \calS}  \pr\left(s'|s,  \tilde{\px}, b \right)\hat{v}(s') - \tilde{v}(s)  =& 
    r\left(s, \hat{\px} , b\right) + \gamma \sum_{s' \in \calS} \pr\left(s'|s,  \hat{\px}, b \right)\hat{v}(s') - \hat{v}(s)  \nonumber\\
        +& r\left(s, \tilde{\px} , b\right) - r\left(s, \hat{\px}, b \right) \nonumber \\
          +&\gamma \sum_{s' \in \calS} \Big( \pr\left(s'|s,  \tilde{\px}, b \right) -\pr\left(s'|s,  \hat{\px}, b \right) \Big)
             \tilde{v}(s')
        \\
        +& \gamma \sum_{s' \in \calS} \pr\left(s'|s,  \hat{\px}, b \right) \big( \tilde{v}(s') - \hat{v}(s') \big) 
        \\
        -& \tilde{v}(s) + \hat{v}(s).\label{eq:add-sub}
    \end{align}
    As a result, given that
    \begin{equation}
        \tilde{\lambda}(s, b) \left(
            \left[ r\left(s, \tilde{\px} , b\right) 
                    + \gamma \sum_{s' \in \calS} \pr\left(s'|s,  \tilde{\px}, b \right)  \tilde{v}(s') \right] - \tilde{v}(s) \right)  = 0,
    \end{equation}
    it follows that from \eqref{eq:add-sub} and the triangle inequality that
    \begin{gather}
        \left| {\tilde{\lambda}(s, b)   \left(
            \left[ r\left(s, \hat{\px} , b\right) 
                    + \gamma \sum_{s'} \pr\left(s'|s,  \hat{\px}, b \right)  \hat{v}(s') \right] - \hat{v}(s) \right) } \right| \leq  
        \frac{1}{1-\gamma} \left( {\rewlip}  + \gamma S {\matlip} \frac{1}{1-\gamma}  + \gamma S L + L  \right)\distht.
    \end{gather}
    This inequality uses that $\|\tilde{\px} - \hat{\px} \|\leq \epsilon$; the fact that $\tilde{\lambda}(s, b) \leq \frac{1}{1 - \gamma}$ \pref{eq:pertconstr6}; and the Lipschitz continuity bounds provided in \Cref{claim:various-bounds}:
   \begin{align}
   \begin{cases}
      \left| r\left(s, \tilde{\px} , b\right) - r\left(s, \hat{\px}, b \right) \right| \leq \rewlip \distht;
     \\
      \left| \sum_{s' \in \calS} \Big( \pr\left(s'|s, \tilde{\px}, b \right) -\pr\left(s'|s,  \hat{\px}, b \right) \Big)
             \tilde{v}(s')  \right| \leq  S  \matlip \frac{1}{1-\gamma} \distht;
    \\
     \left|  \sum_{s' \in \calS} \pr\left(s'|s,  \hat{\px}, b \right) \big( \tilde{v}(s') - \hat{v}(s') \big)  \right| \leq S L \distht; \text{ and}
    \\
     \left| \tilde{v}(s) - \hat{v}(s) \right| \leq  L \distht.
   \end{cases}
    \end{align}
    We proceed in a similar manner for \eqref{eq:pertconstr2}, yielding that
    \begin{align}
        \sum_{b \in \calB} \tlambda(s, b)
        \left[ r\left(s, (\ve_{k,s,a} ; \hat{\px}_{-k}), b\right) 
        + \gamma \sum_{s' \in \calS} \pr\left(s'|s, (\ve_{k,s,a} ; \hat{\px}_{-k}), b \right)  \hat{v}(s')-\hat{v}(s)\right] \geq \nonumber \\ \geq -4 \epsilon \ell - 
         \frac{1}{1-\gamma} \left( \rewlip + \gamma S \matlip \frac{1}{1-\gamma}  + \gamma S L + L  \right) \distht.
    \end{align}
    Thus, $\tilde{\vlambda}$ satisfies \eqref{eq:pertconstr2}. Finally, $\tilde{\vlambda}$ also satisfies \eqref{eq:pertconstr5} and \eqref{eq:pertconstr6}, implied directly by \Cref{proposition:Lan-vis} and \Cref{claim:rho-ul}.
    \end{proof}
    
 

      \begin{lemma}
          \label{lem:lpadv-nash}
      Let $\hat{\px}$ be an $\epsilon$-nearly stationary point of $\phi(\cdot) = \max_{\py \in \calY} V_{\vrho}( \cdot, \py )$. Any feasible solution $\vlambda \in \R^{S\times B}$ of \pref{prog:lpadv} induces an $O(\epsilon)$-approximate Nash equilibrium for the adversarial team Markov game.
    \end{lemma}

    \begin{proof}
    Consider any feasible solution $\vlambda \in \R^{S \times B}$ of \eqref{prog:lpadv}, and the induced strategy for the adversary defined as
    \begin{equation}
        \hat{y}_{s,b} \defeq \frac{\lambda(s,b)}{\sum_{b\in \calB} \lambda(s,b)},
    \end{equation}
    for any $(s,b) \in \calS \times \calB$; this is indeed well-defined since $\sum_{b \in \cal B} \lambda(s, b) \geq \rho(s) > 0$, which in turn follows since $\vrho$ has full support. We will show that $(\hat{\vx}, \hat{\vy})$ is an $O(\epsilon)$-approximate Nash equilibrium. Our proof proceeds in two parts. First, we show that, if the team is responding according to $\hat{\vx}$, then $\hat{y}$ is an $O(\epsilon)$-approximate best response for the adversary. Analogously, in the second part of the proof we argue about deviations from team players.
    %
    
    \paragraph{Controlling deviations of the adversary.} Fix any $\vec{y} \in \calY$. Given that $(\hat{\px}, \hat{\vv} )$ is a feasible solution of \pref{prog:xinlp}, it follows that for any $(s,b) \in \calS \times \calB$,
    \begin{equation}
        y_{s,b} \left( r(s, \hat{\vx}, b) + \gamma \sum_{s' \in \calS} \pr(s' | s, \hat{\vx}, b) \hat{v}(s') \right) \leq \hat{v}(s) y_{s, b},
    \end{equation}
    Summing over all $b \in \calB$ yields that \begin{equation}
         \sum_{b \in \calB} y_{s, b} \left( r(s, \hat{\vx}, b) + \gamma \sum_{s' \in \calS} \pr(s' | s, \hat{\vx}, b) \hat{v}(s') \right) \leq \hat{v}(s), 
    \end{equation}
    in turn implying that
    \begin{equation}
         r\left(s, \hat{\px}, \py \right) + \gamma \sum_{s' \in \calS} \pr \left(s' | s, \hat{\px}, \py \right)\hat{v}(s') \leq \hat{v}(s),
    \end{equation}
    for any $s \in \calS$. The last inequality can be succinctly expressed in the following vector (element-wise) inequality:
    \begin{equation}
        \vr\left(\hat{\px}, \py \right) + \gamma \pr \left(\hat{\px}, \py \right)\hat{\vv} \leq \hat{\vv}.
    \end{equation}
    From this inequality it follows that
    \begin{equation}
    \label{eq:advdev-rel1}
      \hat{\vv} \geq \sum_{t=0}^\infty \gamma^t \pr^t (\hat{\vx}, \hat{\vy}) \vr(\hat{\vx}, \vy) = (\mat{I} - \gamma \pr(\hat{\vx}, \vy) )^{-1} \vr(\hat{\vx}, \vy) = \Vec{V}(\hat{\vx}, \vy),
    \end{equation}
    where we used \Cref{claim:vec-inequality,fact:value-from-matrix,claim:iminusp-invertible}, and the notation $\Vec{V}(\hat{\vx}, \vy)$ to represent the value vector under $(\hat{\vx}, \vy)$---recall \eqref{eq:value-func-def}. Moreover, given that $\vlambda$ is a feasible solution of \eqref{prog:lpadv}, manipulating \eqref{eq:pertconstr4} yields that for any $(s,b) \in \calS \times \calB$,
    \begin{align}
        \lambda(s, b)   \left(
            \left[ r\left(s, \hat{\px}, b \right) 
                    + \gamma \sum_{s' \in \calS}\pr\left(s'|s,  \hat{\px}, b \right)  \hat{v}(s') \right] - \hat{v}(s) \right) 
        \geq - c_2 \epsilon
        \\
        \frac{1}{ \sum_{b' \in \calB} \lambda(s,b') } \lambda(s, b) \left(
            \left[ r\left(s, \hat{\px}, b \right) 
                    + \gamma \sum_{s' \in \calS} \pr\left(s'|s,  \hat{\px}, b \right)  \hat{v}(s') \right] - \hat{v}(s) \right) 
        \geq -  \frac{c_2 \epsilon}{ \sum_{b' \in \calB} \lambda(s,b') }
        \label{eq:b2-4-1}
        \\
        \hat{y}_{s,b} \left(
            \left[ r\left(s, \hat{\px}, b \right) 
                    + \gamma \sum_{s' \in \calS} \pr\left(s'|s,  \hat{\px}, b \right)  \hat{v}(s') \right] - \hat{v}(s) \right) 
        \geq - \frac{ c_2 \epsilon}{ \sum_{b' \in \calB} \lambda(s,b') },
        \label{eq:b2-4-2}
        \end{align}
        where \eqref{eq:b2-4-1} follows since $\sum_{b' \in \calB} \lambda(s,b') > 0$, while \eqref{eq:b2-4-2} follows from the definition of $\hat{y}_{s,b}$. Summing over all $b \in \calB$,
        \begin{align}
        \sum_{b \in \calB}\hat{y}_{s,b}   \left(
            \left[ r\left(s, \hat{\px}, b \right) 
                    + \gamma \sum_{s' \in \calS}\pr\left(s'|s,  \hat{\px}, b \right)  \hat{v}(s') \right] - \hat{v}(s) \right) 
        \geq - \sum_{b \in \calB}  \frac{ c_2 \epsilon}{ \sum_{b' \in \calB} \lambda(s,b') }
        \\
        \sum_{b \in \calB} \hat{y}_{s,b}   \left(
             r\left(s, \hat{\px}, b \right) 
                    + \gamma \sum_{s' \in \calS}\pr\left(s'|s,  \hat{\px}, b \right)  \hat{v}(s')  \right) - \hat{v}(s)
        \geq -B  \frac{c_2 \epsilon}{ \sum_{b \in \calB} \lambda(s,b) } \\
        r\left(s, \hat{\px}, \hat{\py} \right) 
                    + \gamma \sum_{s' \in \calS} \pr\left(s'|s,  \hat{\px}, , \hat{\py} \right)  \hat{v}(s')
        \geq \hat{v}(s) -  B \frac{c_2 \epsilon}{ \sum_{b \in \calB} \lambda(s,b)}. 
        \label{eq:b2-4-3}
    \end{align}
    Let us set $\xi_s \defeq \frac{c_2\cdot\epsilon}{\sum_b \lambda(s,b)}$ for each $s \in \calS$. Continuing from \eqref{eq:b2-4-3}, we have that
    
    $$\vr\left(\hat{\px}, \hat{\py} \right) 
                    + \gamma  \pr\left( \hat{\px}, \hat{\py} \right)  \hat{\vv} 
        \geq \hat{\vv} - B \vxi,$$
    which in turn implies that
    \begin{equation}
        \vec{V}(\hat{\px},  \hat{\py})
        \geq \vv - B  \left( \mat{I} - \gamma \pr (\hat{\px}, \hat{\py}) \right) ^{-1}\vxi,
    \end{equation}
    by \Cref{fact:value-from-matrix,claim:vec-inequality}. Thus,
    \begin{align}
        V_{\vrho}(\hat{\px},  \hat{\py})
        &\geq \vrho^\top \hat{\vv} - B \vrho ^\top \left( \mat{I} - \gamma \pr (\hat{\px}, \hat{\py}) \right) ^{-1}\vxi
        \label{eq:b2-1-4}
        \\
        &\geq \vrho^\top \hat{\vv} - B \vxi^\top \vd_{\vrho}^{\hat{\px}, \hat{\py} }
        \label{eq:b2-1-5}
        \\
        &\geq \vrho^\top \hat{\vv} - c_2 B \sum_{s \in \calS} \frac{d_{\vrho}^{\hat{\px}, \hat{\py}}(s)}{\sum_{b \in \calB} \lambda(s,b)}\epsilon
        \label{eq:b2-1-6}
        \\
        &\geq \vrho^\top \hat{\vv} - c_2 B \sum_{s \in \calS} \frac{d_{\vrho}^{\hat{\px}, \hat{\py}}(s)}{\rho(s)}\epsilon \label{eq:b2-1-7} \\
        &\geq \vrho^\top \hat{\vv} -c_2 B {S} D\epsilon,
        \label{eq:advdev-rel2}
    \end{align}
    where \eqref{eq:b2-1-5} follows from \Cref{fact:visitation}; \eqref{eq:b2-1-7} follows from the feasibility constraint $\sum_{b \in \calB} \lambda(s,b) \geq \rho(s)$; and \eqref{eq:advdev-rel2} uses the definition of mismatch coefficient (\Cref{def:mismatch}). As a result, combining \eqref{eq:advdev-rel1} and \pref{eq:advdev-rel2}, we conclude that for any $\vec{y} \in \calY$,
    \begin{equation}
        \label{eq:less-value}
        V_{\vrho}(\hat{\vx}, \hat{\vy}) \geq \vrho^\top \hat{\vv} - c_2 BSD  \epsilon \geq V_{\vrho}(\hat{\vx}, \vec{y}) -  c_2 BSD \epsilon.
    \end{equation}
    
    \paragraph{Controlling deviations of a team player.} Next, we show that any deviation from a single player can only yield a small improvement for the player. Fix any player $k \in [n]$ and strategy $\vec{x}_k \in \calX_k$. The proof proceeds analogously to our previous argument. In particular, for any state $s \in \calS$, multiplying \pref{eq:pertconstr2} by $\vec{x}_{k, s, a}$, and summing over all actions $a \in \calA_k$ yields that
    \begin{equation}
        \sum_{b \in \calB} \lambda(s, b)
                    \left[ r\left(s, (\vx_k; \hat{\vx}_{-k}), b\right) 
                    +
                    \gamma \sum_{s' \in \calS} \pr\left(s'|s, (\vec{x}_k; \hat{\px}_{-k}), b \right)  \hat{v}(s')\right]  \geq \hat{v}(s) \sum_{b \in \calB} \lambda(s,b) - c_1 \epsilon;
    \end{equation}
    here, we leveraged the feasibility of $\vlambda$. Further, given that $\sum_{b \in \calB} \lambda(s,b) > 0$,
    \begin{equation}
        r(s,(\vx_k; \hat{\vx}_{-k}), \hat{\vy}) + \gamma \sum_{s' \in \calS} \pr(s' | s, (\vx_k; \hat{\vx}_{-k}), \hat{\vy}) \geq \hat{v}(s) - c_1\cdot\epsilon \frac{1}{\sum_{b \in \calB} \lambda(s,b)} \geq \hat{v}(s) - c_1 \epsilon \frac{1}{\rho(s)},
    \end{equation}
    for any $s \in \calS$, since $\sum_{b \in \calB} \lambda(s,b) \geq \rho(s)$. Hence,
    \begin{align}
        \vr\left( ( \px_k ; \hat{\px}_{-k} ), \hat{\py} \right) + \gamma \pr\left( ( \px_k ; \hat{\px}_{-k} ), \hat{\py} \right)\hat{\vv} \geq \hat{\vv} - c_1 \epsilon \frac{1}{\vrho}.
    \end{align}
    In turn, by \Cref{claim:vec-inequality}, this implies that
    \begin{align}
        \vec{V}((\vec{x}_k, \hat{\vx}_{-k}), \hat{\vy}) \geq \hat{\vv} - c_1\cdot\epsilon (\mat{I} - \gamma \pr((\vx_k; \hat{\vx}_{-k}), \hat{\vy}) )^{-1} \frac{1}{\vrho}.
    \end{align}
    Thus, we conclude that
    \begin{equation}
        \label{eq:part1}
        V_{\vrho} (\vec{x}_k, \hat{\vx}_{-k}), \hat{\vy}) \geq \vrho^\top \hat{\vv} - c_1 D S \epsilon,
    \end{equation}
    where we used \Cref{fact:visitation} and \Cref{def:mismatch}. Next,
    using \eqref{eq:pertconstr3} we obtain that for all $(s,b) \in \calS \times \calB$,
    \begin{align}
        \lambda(s, b)   \left(
            \left[ r\left(s, \hat{\px}, b \right) 
                    + \gamma \sum_{s' \in \calS}\pr\left(s'|s,  \hat{\px}, b \right)  \hat{v}(s') \right] - \hat{v}(s) \right) &\leq c_2 \epsilon \\
        \frac{\lambda(s, b)}{\sum_{b' \in \calB} \lambda(s,b')} \left(
            \left[ r\left(s, \hat{\px}, b \right) 
                    + \gamma \sum_{s' \in \calS}\pr\left(s'|s,  \hat{\px}, b \right)  \hat{v}(s') \right] - \hat{v}(s) \right) &\leq \frac{c_2 \epsilon}{\sum_{b' \in \calB} \lambda(s,b')}.     
    \end{align}
    For convenience, let us set $\xi_s \defeq \frac{c_2 \epsilon}{\sum_{b' \in \calB} \lambda(s,b')}$. By definition of $\hat{\vy}$, we have
    \begin{align}
            \hat{y}_{s,b}   \left(
            \left[ r\left(s, \hat{\px}, b \right) 
                    + \gamma \sum_{s' \in \calS}\pr\left(s'|s,  \hat{\px}, b \right)  \hat{v}(s') \right] - \hat{v}(s) \right) \leq \xi_s 
            \\
            \sum_{b \in \calB} \hat{y}_{s,b}   \left(
            \left[ r\left(s, \hat{\px}, b \right) 
                    + \gamma \sum_{s' \in \calS}\pr\left(s'|s,  \hat{\px}, b \right)  \hat{v}(s') \right] - \hat{v}(s) \right) \leq B \xi_s 
            \\
            \sum_{b \in \calB} \hat{y}_{s,b}   \left(
             r\left(s, \hat{\px}, b \right) 
                    + \gamma \sum_{s' \in \calS}\pr\left(s'|s,  \hat{\px}, b \right)  \hat{v}(s')  \right) \leq   \hat{v}(s) + B \xi_s 
            \\
            r\left(s, \hat{\px}, \hat{\py} \right) 
                    + \gamma \sum_{s' \in \calS} \pr \left(s'|s,  \hat{\px}, \hat{\py} \right)  \hat{v}(s')  \leq  \hat{v}(s) + B \xi_s,
    \end{align}
    for any $s \in \calS$. Thus,
    \begin{align}
        \vr\left(\hat{\px}, \hat{\py} \right) 
                    + \gamma  \pr\left(\hat{\px}, \hat{\py} \right)  \hat{\vv} \leq  \hat{\vv} + B \vxi
                \label{eq:b2-3-1}
            \\
            \vec{V}\left(\hat{\px}, \hat{\py} \right)  \leq \hat{\vv} + B \left( \mat{I} -\gamma \pr(\hat{\px}, \hat{\py}) \right)^{-1} \vxi
                \label{eq:b2-3-2}
            \\
            V_{\vrho}(\hat{\px},  \hat{\py}) \leq
            \vrho^\top \hat{\vv} + B c_2 \sum_{s \in \calS} \frac{d_{\vrho}^{\hat{\px}, \hat{\py}}(s)}{\sum_{b \in \calB}\lambda(s,b)} \epsilon
                \label{eq:b2-3-3}
            \\
            V_{\vrho}(\hat{\px},  \hat{\py}) \leq
            \vrho^\top \hat{\vv} + c_2 B S D \epsilon, \label{eq:part2}
    \end{align}
    where \eqref{eq:b2-3-2} follows from \Cref{claim:vec-inequality}; \eqref{eq:b2-3-3} follows from \Cref{fact:visitation}; and \eqref{eq:part2} follows from the fact that $\sum_{b \in \calB} \lambda(s,b) \geq \rho(s)$ and \Cref{def:mismatch}. As a result,
    combining \eqref{eq:part1} and \eqref{eq:part2} we conclude that
    \begin{align}
        V_{\vrho} (\hat{\px}, \hat{\py}) \leq  V_{\vrho} \left( ( \px_k ; \hat{\px}_{-k} ), \hat{\py} \right) + c_2 BSD \epsilon + c_1 D S \epsilon.
        \label{eq:greater-value}
    \end{align}
    \end{proof}
    
    We state the precise version of \Cref{lem:lpadv-nash} in \Cref{theorem:approx-Nash} below. First, let us summarize \advnashpolicy, the algorithm for computing the policy for the adversary. \advnashpolicy, described in \Cref{alg:adv-nash-policy}, takes as input $\hat{\vx} \in \calX$, an $\epsilon$-nearly stationary point of $\phi(\vx) \defeq \max_{\vy \in \calY} V_{\vrho}(\vx, \vy)$. The algorithm then computes the best-response value vector $\hat{\vv}$. This is computed by fixing the strategy of the team $\hat{\vx} \in \calX$, and then solving the single-agent MDP problem so as to maximize the value at every state. Then, the pair $(\hat{\vx}, \hat{\vv})$ is used in order to determine the---polynomial number of---coefficients of $\mathrm{LP}_\text{adv}$, as introduced in~\eqref{prog:lpadv}. Then, any feasible solution $\vlambda \in \R^{S \times B}$ of \eqref{prog:lpadv} is used to determine the strategy of the adversary as follows.
    \begin{equation}
        \hat{y}_{s,b} \defeq \frac{\lambda(s,b)}{\sum_{b \in \calB} \lambda(s,b)}, \quad \forall (s,b) \in \calS \times \calB.
    \end{equation}

\begin{algorithm}
  \caption{\advnashpolicy \label{alg:adv-nash-policy}}
  \begin{algorithmic}[1]
  \Require{An $\epsilon$-nearly stationary point $\hat{\px} \in \calX$ of $\phi(\vx) \defeq \max_{\vy \in \calY} V_{\vrho}(\vx, \vy)$}
    \Let{$\hat{\vv}$}{The best-response value vector for the adversary} \label{line:vv}
    \Let{$\mathrm{LP}_\text{adv}$}{Compute the coefficients of the linear program~\eqref{prog:lpadv}} \label{line:lpadv}
    \Let{$\vlambda$}{Any feasible solution of $\mathrm{LP}_\text{adv}$}
    \Let{$\hat{y}_{s,b}$}{$\frac{\lambda(s,b)}{\sum_{b \in \calB}\lambda(s,b)}$}
    \Statex
        \Return $\hat{\py}$
  \end{algorithmic}
\end{algorithm}
    \begin{remark}
    \label{remark:oracle}
In Lines \ref{line:vv} and \ref{line:lpadv} of \Cref{alg:adv-nash-policy} we adopt the assumption of ``polynomially accessible games,’’ which is standard in computational game theory. More precisely, we assume that we have an oracle that outputs the expected reward and the expected transition probabilities in any state given the announced (mixed) strategies of every player. It is known that such an oracle for the expected utilities can be implemented in polynomial time for virtually all\footnote{Nevertheless, there are certain rather artificial succinct games for which computing the expected utilities is \#\P-hard~\citep[Proposition 1]{Daskalakis06:The}.} interesting classes of \emph{succinctly representable} normal-form games~\citep{Papadimitriou08:Computing}, and even for extensive-form games~\citep{Huang08:Computing}. Crucially, without a succinct representation or an oracle access to the game the input scales exponentially with the number of players, trivializing the computational aspects of the problem in multiplayer games; this issue has been discussed extensively for normal-form games with regards to Nash equilibria~\citep{Daskalakis09:The}, and correlated equilibria~\citep{Papadimitriou08:Computing}. 

In our case, we will essentially assume that the adversary has a ``polynomially accessible environment.'' More explicitly, the adversary should have access to the expected rewards and the expected transition probabilities of the single-agent MDP induced  if the rest of the players have announced and fixed their policies. Further, in \IPGmax we also make the standard assumption that players in the team have access to the gradients of the value function.
\end{remark}

\begin{theorem}[Near stationary points extend to approximate NE]
        \label{theorem:approx-Nash}
        Consider an adversarial team Markov game $\calG$, and suppose that $\hat{\px} \in \calX$ is an $\epsilon$-nearly stationary point of $\phi(\vx) \defeq \max_{\py} V_{\vrho}(\cdot, \py)$, where $V_{\vrho}$ is the value function of $\calG$~\eqref{eq:value-func-def}. Then, any feasible solution of \pref{prog:lpadv} $\hat{\vlambda} \in \R_{\geq 0}^{S \times B}$ induces a strategy $\hat{\vy}$, defined as
        \begin{equation}
            \hat{y}_{s,b} \defeq \frac{\hat{\lambda}(s,b) }{\sum_{b \in \calB}\hat{\lambda}(s,b) }, \quad \forall (s,b) \in \calS \times \calB,
        \end{equation}
        so that for any player $k \in [n]$ and any deviations $\vx_k \in \calX_k$ and $\vy \in \calY$,
        \begin{equation}
            \left\{
            \begin{array}{lcl}
                 V_{\vrho}(\hat{\px}, \hat{\py})& \leq &  V_{\vrho}\left( (\px_k; \hat{\px}_{-k} ), \hat{\py} \right) + \left(  B  {S} D  + 1 \right) \frac{1}{1-\gamma}     \left( \rewlip + \gamma S \matlip \frac{1}{1-\gamma}  + \gamma S L + L \right)\distht    \\
                 & & + \frac{1}{1-\gamma} 4 \epsilon {\ell} 
                 \\
                 V_{\vrho}(\hat{\px}, \hat{\py})& \geq & V_{\vrho} (\hat{\px}, \py) -  B {S} D \frac{1}{1-\gamma} \left( \rewlip + \gamma S \matlip \frac{1}{1-\gamma}  + \gamma S  L  + L  \right)\distht,
            \end{array}
            \right.
        \end{equation}
        Here, we recall that $D = \max_{\policy{} \in \Pi} \left\| \frac{ \vd_{\vrho}^{\policy{}} }{\vrho } \right\|_{\infty}$ is the mismatch coefficient, $L= \frac{\sqrt{\sum_k A_k +B}}{(1-\gamma)^2}$ is a Lipschitz constant of the value function, and $\ell = \frac{2 \left(\sum_{k} A_{k}+B\right)}{(1-\gamma)^3}$ is a smoothness constant of the value function (\Cref{lem:smoothness}).
        \label{them:epsilon-NE}
    \end{theorem}
    
    \begin{proof}    
    By \Cref{lem:lpadv-feasible}, we know that \eqref{prog:lpadv} is feasible. Further, $\hat{\vy}$ is a well-formed strategy since for any feasible $\vlambda \in \R_{\geq 0}^{S \times B}$ of \eqref{prog:lpadv} it holds that $\sum_{b \in \calB} \lambda(s,b) \geq \rho(s) > 0$, for any state $s \in \calS$, where the first bound follows by feasibility of $\vlambda$ and the second since $\vrho$ is assumed to have full support. Thus, the proof of the theorem follows from \Cref{lem:lpadv-nash}, and in particular \eqref{eq:less-value} and \eqref{eq:greater-value}.
    \end{proof}

\section{Convergence to a Nearly Stationary Point}
\label{sec:convergence}

In this section, we establish that \IPGmax reaches to an $\epsilon$-nearly stationary point---in the sense of \Cref{def:nearst}---after a number of iterations that is polynomial in all the natural parameters of the game, as well as $1/\epsilon$. The main result here is \Cref{lem:ipgmax-convergence-lemma}, which was first introduced in \Cref{sec:IPGmax}. First, we need to establish that the value function $V_{\vec{\rho}}(\vx,\vy)$ is Lipschitz continuous and smooth, as formalized below. We note that this property is by now fairly standard (\emph{e.g.}, see \citep{Agarwal20:Optimality}), and we therefore omit the proof.

\begin{lemma}
\label{lem:smoothness}
For any initial distribution $\vec{\rho}$, the value function $V_{\vrho}(\px,\py)$ is $\frac{\sqrt{\sum_{k}A_k +B}}{(1-\gamma)^2}$-Lipschitz continuous and
$\frac{2 \left(\sum_{k} A_{k}+B\right)}{(1-\gamma)^3}$-smooth:
\begin{align}
    \label{eq:smooth}
    | V_{\vrho}(\px,\py)- V_{\vrho}(\px',\py')| \leq \frac{\sqrt{\sum_{k=1}^n A_k +B}}{(1-\gamma)^2} \norm{(\px,\py) - (\px',\py')}; \text{ and} \\
    \norm{\nabla V_{\vrho}(\px,\py)-\nabla V_{\vrho}(\px',\py')} \leq \frac{2 \left(\sum_{k=1}^n A_{k}+B\right)}{(1-\gamma)^3} \norm{(\px,\py) - (\px',\py')},
\end{align}
for all $(\px,\py) , (\px',\py') \in \calX \times \calY$.
\end{lemma}
We convenience, we will let $L \defeq \frac{\sqrt{\sum_{k=1}^n A_k +B}}{(1-\gamma)^2}$ and $\ell \defeq \frac{2 \left(\sum_{k=1}^n A_{k}+B\right)}{(1-\gamma)^3}$. The next key result characterizes the iteration complexity required to reach an $\epsilon$-nearly stationary point of $\phi(\cdot)$. The following analysis follows~\citep{Jin20:What}.

\ipgmx*

\begin{proof}
By virtue of the $\ell$-smoothness of $V_{\vrho}(\px,\py)$ (\Cref{lem:smoothness}), it follows that for any $\vx \in \calX$ and $0 \leq t \leq T -1$,
\begin{equation}\label{eq:smooth1}
    \phi(\px) \geq V_{\vrho}(\px,\py^{(t+1)}) \geq V_{\vrho}(\px^{(t)},\py^{(t+1)}) +\langle\nabla_{\px} V_{\vrho}(\px^{(t)},\py^{(t+1)}),\px-\px^{(t)}\rangle-\frac{\ell}{2}\norm{\px-\px^{(t)}}^2,
\end{equation}
since $\phi(\vec{x}) = \max_{\vec{y} \in \calY} V_{\vrho}(\vec{x}, \vec{y}) \geq V_{\vrho}(\vec{x}, \vec{y}^{(t+1)})$. Now recall that
\begin{equation}
    \label{eq:xtil}
 \tilde{\vx}^{(t)} \defeq \arg\min_{\px' \in \calX} \left\{ \phi(\px')+\frac{1}{2\lambda}\norm{\px^{(t)}-\px'}^2 \right\},   
\end{equation}
for any $0 \leq t \leq T -1$, where $\lambda \defeq \frac{1}{2\ell}$. Using the definition of Moreau envelope (\Cref{def:moreau}), 
\begin{align}
    \hspace{-0.7cm}
        \phi_{\lambda}(\px^{(t+1)}) 
            &\leq \phi(\tilde{\px}^{(t)}) + \ell \norm{\px^{(t+1)}-\tilde{\px}^{(t)}}^2 \\
        &\leq \phi(\tilde{\px}^{(t)}) + \ell \norm{ \proj{\calX}{\vx^{(t)} - \eta \nabla_{\vec{x}} V_{\vec{\rho}} (\vec{x}^{(t)}, \vec{y}^{(t+1)}) } - \proj{\calX}{\tilde{\px}^{(t)}}}^2 \label{eq:proj} \\
            &\leq \phi(\tilde{\px}^{(t)}) + \ell \norm{\px^{(t)}-\eta \nabla_{\px} V_{\vrho}(\px^{(t)},\py^{(t+1)})-\tilde{\vx}^{(t)}}_2^2 \label{eq:nonexpansive}
        \\
        &\leq \phi(\tilde{\px}^{(t)}) + \ell \norm{\vx^{(t)} - \tilde{\vx}^{(t)}}^2 + \eta^2 \ell \norm{\nabla_{\px} V_{\vrho}(\px^{(t)},\py^{(t+1)})}^2 + 2 \eta \ell \langle\nabla_{\px} V_{\vrho}(\px^{(t)},\py^{(t+1)}),\tilde{\px}^{(t)}-\px^{(t)}\rangle \label{eq:pytha} \\
            &\leq \phi_{\lambda} (\px^{(t)})  + 2\eta \ell \left(\phi(\tilde{\px}^{(t)})-\phi(\px^{(t)})+\frac{\ell}{2}\norm{\px^{(t)}-\tilde{\px}^{(t)}}^2\right) + \eta^2 \ell L^2, \label{eq:multi}
\end{align}
where 
\begin{itemize}
    \item \eqref{eq:proj} uses the fact that $\vec{x}_k^{(t+1)} \defeq  \proj{\calX_k}{\vx_k^{(t)} - \eta \nabla_{\vec{x}_k} V_{\vec{\rho}} (\vec{x}^{(t)}, \vec{y}^{(t+1)}) }$ for all $k \in [n]$, as defined in \IPGmax, in turn implying that $\vec{x}^{(t+1)} = \proj{\calX}{\vx^{(t)} - \eta \nabla_{\vec{x}} V_{\vec{\rho}} (\vec{x}^{(t)}, \vec{y}^{(t+1)}) }$, as well as the fact that $\proj{\calX}{\tilde{\px}^{(t)}} = \tilde{\vx}^{(t)}$ since $\tilde{\vx}^{(t)} \in \calX$;
    \item \eqref{eq:nonexpansive} follows from the fact that the projection operator is nonexpansive (\Cref{fact:nonexpansive});
    \item \eqref{eq:pytha} uses the identity $\| \vec{a} + \vec{b} \|^2 = \|\vec{a}\|^2 + \|\vec{b}\|^2 + 2 \langle \vec{a}, \vec{b} \rangle$ for any $\vec{a}, \vec{b} \in \R^d$; and
    \item \eqref{eq:multi} follows since
    \begin{itemize}
        \item[(i)] $\phi(\tilde{\px}^{(t)}) + \ell \norm{\vx^{(t)} - \tilde{\vx}^{(t)}}^2 = \min_{\vec{x}' \in \calX} \left\{ \phi(\vx') + \ell \norm{\vx^{(t)} - \vx'}^2 \right\} = \phi_{\lambda} (\px^{(t)})$ by definition of $\tilde{\vec{x}}^{(t)}$ in \eqref{eq:xtil} and the definition of Moreau envelope with $\lambda = \frac{1}{2\ell}$ (\Cref{def:moreau});
        \item[(ii)] $ V_{\vrho}(\px^{(t)},\py^{(t+1)}) +\langle\nabla_{\px} V_{\vrho}(\px^{(t)},\py^{(t+1)}),\tilde{\px}^{(t)}-\px^{(t)}\rangle-\frac{\ell}{2}\norm{\tilde{\px}-\px^{(t)}}^2 \leq \phi(\tilde{\vx}^{(t)})$, which is an application of \eqref{eq:smooth1} for $\vx \defeq \tilde{\vx}^{(t)}$; and 
        \item[(iii)] $\norm{\nabla_{\px} V_{\vrho}(\px^{(t)},\py^{(t+1)})}^2 \leq L^2$ by $L$-Lipschitz continuity of $V_{\vrho}(\px^{(t)},\py^{(t+1)})$ (\Cref{lem:smoothness}) combined with \Cref{fact:boungrad}.
    \end{itemize}
\end{itemize}
As a result, taking a telescopic sum of \eqref{eq:multi} for all $0 \leq t \leq T - 1$ and rearranging the terms yields
\begin{equation}
    \label{eq:ave-gap}
    \frac{1}{T} \sum_{t=0}^{T-1} \left(\phi(\px^{(t)})-\phi(\tilde{\px}^{(t)})-\frac{\ell}{2}\norm{\px^{(t)}-\tilde{\px}^{(t)}}^2\right) \leq \frac{\phi_{\lambda}(\px^{(0)}) - \phi_{\lambda}(\vec{x}^{(T)}) }{2\eta \ell T} + \frac{\eta L^2}{2} \leq \frac{1}{2 (1 - \gamma) \eta \ell T} + \frac{\eta L^2}{2},
\end{equation}
since $\phi_{\lambda}(\vec{x}^{(T)})\geq 0$, directly by \Cref{def:moreau}, and $\phi_{\lambda}(\vec{x}^{(0)}) \leq \phi(\vec{x}^{(0)}) \leq \frac{1}{1 - \gamma}$, where the last inequality follows from \Cref{claim:V-bounds}. Therefore we conclude that there exists an iterate $\tstar$, with $0 \leq \tstar \leq T - 1$, so that
\begin{equation}
    \label{eq:phi-diff}
    \phi(\px^{(\tstar)})-\phi(\tilde{\px}^{(\tstar)})-\frac{\ell}{2}\norm{\px^{(\tstar)}-\tilde{\px}^{(\tstar)}}^2 \leq \frac{1}{2(1 - \gamma)\eta \ell T} + \frac{\eta L^2}{2}.
\end{equation}

Further, since $\phi(\px) + \ell \norm{\px - \px^{(\tstar)}}^2$ is $\ell$-strongly convex with respect to $\vec{x}$
(by \Cref{lem:max-weakly-convex} and \Cref{cor:weak-conv}), we get that
\begin{equation}
    \phi(\px^{(\tstar)}) - \phi(\tilde{\px}^{(\tstar)}) - \ell \norm{\px^{(\tstar)}-\tilde{\px}^{(\tstar)}}^2 \geq \frac{\ell}{2} \norm{\px^{(\tstar)}-\tilde{\px}^{(\tstar)}}^2,
\end{equation}
by definition of $\tilde{\vx}^{(\tstar)}$ in \eqref{eq:xtil}, in turn implying that
\begin{gather}
\phi(\px^{(\tstar)}) - \phi(\tilde{\px}^{(\tstar)}) - \frac{\ell}{2} \norm{\px^{(\tstar)}-\tilde{\px}^{(\tstar)}}^2 \geq \ell \norm{\px^{(\tstar)}-\tilde{\px}^{(\tstar)}}^2.
\end{gather}
Combing this bound with \eqref{eq:phi-diff} yields that

\begin{gather}
\norm{\px^{(\tstar)}-\tilde{\px}^{(\tstar)}}^2 \leq \frac{1}{2 (1 - \gamma)\eta \ell^2 T} + \frac{\eta L^2}{2\ell}.
\end{gather}
In particular, letting
\begin{equation}
    \eta = \epsilon^2 \cdot \frac{\ell}{L^2} = 2 \epsilon^2 \cdot (1-\gamma)
\end{equation}
and
\begin{equation}
   T = \frac{1}{\epsilon^2 (1-\gamma) \eta \ell^2 } = \frac{(1-\gamma)^4}{8 \epsilon^4 (\sum_{i=1}^n A_i + B)^2}
\end{equation}
implies that $\norm{\px^{(\tstar)}-\tilde{\px}^{(\tstar)}} \leq \epsilon$.
\end{proof}

A limitation of this proposition is that it only establishes a ``best-iterate'' guarantee. However, as we explained in \Cref{sec:IPGmax}, determining such an iterate could introduce a substantial computational overhead in the algorithm. For this reason, we provide a stronger guarantee below, showing that even a random iterate will also be nearly stationary with constant probability, leading to a practical implementation of \IPGmax.

\begin{corollary}
    \label{cor:highprob}
Consider any $\epsilon > 0$, and suppose that $\eta = \epsilon^2  ( 1-\gamma)$ and  $T=\frac{(1-\gamma)^4}{2 \epsilon^4 (\sum_{k}A_k+B)^2}$. For any $\delta > 0$, if we select uniformly at random (with repetitions) a set $\mathcal{T}$ of $\lceil \log(1/\delta) \rceil$ indexes from the set $\{0, 1, \dots, T-1\}$, then with probability at least $1 - \delta$ there exists a $t' \in \mathcal{T}$ such that $\norm{\vx^{(t')} - \tilde{\vx}^{(t')}} \leq \epsilon$, where $\tilde{\vx}^{(t')} \defeq \prox_{\phi/(2\ell)}(\vx^{(t')})$.
\end{corollary}

\begin{proof}
First, we claim that selecting uniformly at random an index $t'$ from the set $\{0, 1, \dots, T-1\}$ will satisfy
\begin{equation}
    \|\vx^{(t')} - \tilde{\vx}^{(t')} \|^2 \leq 2 \epsilon^2
\end{equation}
with probability at least $\frac{1}{2}$. To show this, let us define 
\begin{equation}
    g^{(t)} \defeq \phi(\vx^{(t)}) - \phi(\tilde{\vx}^{(t)}) - \frac{\ell}{2} \norm{\vec{x}^{(t)} - \tilde{\vx}^{(t)}}^2,
\end{equation}
for $t = 0, 1, \dots, T-1$. By definition of $\tilde{\vx}^{(t)}$ in \eqref{eq:xtil}, we have
\begin{equation}
    \label{eq:strong-con}
    g^{(t)} = \phi(\vx^{(t)}) - \phi(\tilde{\vx}^{(t)}) - \frac{\ell}{2} \norm{\vx^{(t)} - \tilde{\vx}^{(t)}}^2 \geq \ell \norm{\vx^{(t)} - \tilde{\vx}^{(t)}} \geq 0,
\end{equation}
for $0 \leq t \leq T - 1$. Further, by \eqref{eq:ave-gap} we have
\begin{equation}
    \label{eq:small-gap}
    \frac{1}{T} \sum_{t=0}^{T-1} g^{(t)} \leq \epsilon^2 \ell,
\end{equation}
where we used that $\eta = 2\epsilon^2  ( 1-\gamma)$ and $T=\frac{(1-\gamma)^4}{8 \epsilon^4 (\sum_{k=1}^n A_k+B)^2}$. As a result, we conclude that at least half of the indexes $t$ are such that $g^{(t)} \leq 2 \epsilon^2 \ell$. Indeed, the contrary case contradicts \eqref{eq:small-gap} given that $g^{(t)} \geq 0$ for all $t$. In turn, this implies our claim in light of \eqref{eq:strong-con}. Finally, the proof of the corollary follows from a standard boosting argument, as well as rescaling $\epsilon$ by $\frac{1}{\sqrt{2}}$.
\end{proof}

\begin{theorem}[Computing $\epsilon$-approximate NE]
    \label{theorem:formal}
    Consider an adversarial team Markov game $\calG$. Running \IPGmax for $
    T =
    \frac{ 512 S^8 D^4  \left( \sum_{k=1}^n A_k + B \right)^4 }{ \epsilon^4 (1-\gamma)^{12} }
    $
    number of iterations
    and learning rate
    $\eta = \frac{ \epsilon^2 (1 - \gamma )^9}{ 32 S^4 D^2  \left( \sum_{k=1}^n A_k + B \right)^3 }
    $
     yields a team strategy $\hat{\px} \in \calX$ that can be extended to an $\epsilon$-approximate Nash equilibrium in polynomial time through the routine $\advnashpolicy(\hat{\vx})$, assuming a polynomially accessible environment for the adversary (\Cref{remark:oracle}).
    \label{thm:main_full}
\end{theorem}
\begin{proof}
In place of $\epsilon$ of \Cref{lem:ipgmax-convergence-lemma} we set $$ \epsilon \gets \frac{\epsilon}{\frac{1}{1-\gamma}\left[ 4\ell + 
    (BSD + 1) \left( {\rewlip} + \gamma S \matlip  \frac{1}{1-\gamma} + \gamma S L + L \right)
    \right]},$$
    which allows us to compute an $\epsilon$-approximate Nash equilibrium by virtue of \Cref{them:epsilon-NE}. Then, the number of iterations reads

\begin{align}
T   &= 
    \frac{(1-\gamma)^4}{8  (1 - \gamma)^4 \epsilon^4 (\sum_{k=1}^n A_k+B)^2}
    \left[ 4\ell + 
    (BSD + 1) \left({\rewlip} + \gamma S \matlip  \frac{1}{1-\gamma} + \gamma S L + L \right)
    \right]^{4} \\
    &\leq
    \frac{ 8^3 S^8 D^4  \left( \sum_{k=1}^n A_k + B \right)^4 }{ \epsilon^4 (1-\gamma)^{12}},
\end{align}
with a learning rate
\begin{align}
    \eta 
    &=
    2 \epsilon^2 (1 - \gamma) (1 - \gamma)^2
    \left(  \frac{1}{1 - \gamma} 
  \left[ 4\ell + 
    (BSD + 1) \left( \rewlip + \gamma S \matlip \frac{1}{1-\gamma} + \gamma S L + L \right)
    \right]
        \right) ^{-2} \\
   &\geq
   \frac{\epsilon^2 (1 - \gamma )^9}{ 32 S^4 D^2  \left( \sum_{k=1}^n A_k + B \right)^3}.
\end{align}

    Further, assuming a polynomially accessible environment for the adversary, \advnashpolicy can be implemented in polynomial time via linear programming~\citep{ye2011simplex}.
\end{proof}

\section{Additional Auxiliary Claims}
\label{sec:marginalia}

For the sake of readability, this section contains some simple and standard claims we used earlier in our proofs, but are only stated here.

\begin{fact}
    \label{fact:boungrad}
    Let $f : \calX \ni \Vec{x} \mapsto \R$ be an $L$-Lipschitz continuous and differentiable function. Then,
    \begin{equation}
        \max_{\Vec{x} \in \calX} \|\nabla_{\Vec{x}} f(\Vec{x}) \| \leq L.
    \end{equation}
\end{fact}

\begin{fact}[Projection operator is nonexpansive]
    \label{fact:nonexpansive}
    Let $\calX \subseteq \R^d$ be a nonempty, convex and compact set. Further, let $\proj{\calX}: \R^d \rightarrow \calX $ be the Euclidean projection operator defined as $\proj{\calX}: \R^d \ni \vy \mapsto \frac{1}{2} \argmin_{\vx \in \calX}\| \vx - \vy \|^2 $. Then, for any $\vx, \vy \in \R^d$,
    \begin{equation}
        \| \proj{\calX}{\vx} - \proj{\calX}{\vy} \| \leq \| \vx - \vy\|.
    \end{equation}.
    \label{fact:proj-nonexpansive}
\end{fact}

In the rest of the claims, we are implicitly---for the sake of readability---fixing an adversarial team Markov game $(\calS, \calA, \calB, r, \pr , \gamma, \vrho)$.

\begin{claim}
    Consider any joint stationary policy $\policy{}\in \Pi$. For any $\gamma \in [0,1)$, the matrix $\mat{I} - \gamma \pr(\policy{})$ is invertible.
    \label{claim:iminusp-invertible}
\end{claim}

\begin{claim} Let $\policy{} \in \Pi$ be a joint stationary policy. The value vector $\vec{V} \in \R^{S}$ can be expressed as
    \begin{equation}
        \vec{V} = \big( \mat{I} - \gamma \mat{\pr}(\policy{}) \big)^{-1} \vr(\policy{}),
    \end{equation}
    where $\vr(\policy{})$ denotes the per-state reward under policy $\policy{}$.
    \label{fact:value-from-matrix}
\end{claim}
\begin{proof}
    For any state $s \in \calS$,
    \begin{equation}
        V_{s}(\policy{}) = \vr(\policy{}) + \gamma \pr(\policy{}) + \gamma^2 \pr^2(\policy{}) + \dots = \sum_{t=0}^\infty \gamma^t \pr^t(\policy{}) \vr(\policy{}). 
    \end{equation}
    But, given that the matrix $\mat{I} - \gamma \pr(\policy{})$ is invertible (\Cref{claim:iminusp-invertible}), we have
    \begin{equation}
        \sum_{t=0}^\infty \gamma^t \pr^t(\policy{}) = \big( \mat{I} - \gamma \pr (\policy{})\big)^{-1},
    \end{equation}
    and the claim follows.
\end{proof}

\begin{claim}
    \label{fact:visitation}
    Consider a stationary joint policy $\policy{} \in \Pi$. The discounted visitation measure $d_{\vrho}^{\policy{}}(s)$ can be expressed as
    \begin{equation}
        \left( {\vd^{\policy{}}_{\vrho}}\right)^\top = \vrho^\top \big(\mat{I} - \gamma \pr(\policy{}) \big)^{-1}.
    \end{equation}
\end{claim}

\begin{claim}
    \label{claim:visitation-and-value}
    Consider a stationary joint strategy $(\vx, \vy) \in \calX \times \calY$, and the visitation measure $\vd^{\px, \py}_{\vrho}$, under some initial distribution $\vrho \in \Delta(\calS)$. Then, the value function can be expressed as
    \begin{equation}
        V_{\vrho} = \sum_{s\in\calS} d_{\vrho}^{\px, \py}(s) r(s,\px,\py).
    \end{equation}
\end{claim}
\begin{proof}
    By definition of $\vd^{\px, \py}_{\vrho}$, we have that for any $s \in \calS$,
    \begin{equation}
        d_{\vrho}^{\px,\py}(s) = \sum_{\bar{s}\in\calS} \rho(\bar{s}) \sum_{t=0}^{\infty} \gamma^t \pr\left(s\step{t}=s ~\big|~ \px, \py, s\step{0} = \bar{s}\right).
    \end{equation}
    Similarly, the value function can be written as
    \begin{align}
        V_{\vrho}(\px, \py) = \sum_{s \in \calS} \sum_{\bar{s} \in \calS} \rho(\bar{s}) \sum_{t=0}^{\infty} \gamma^t \pr\left(s\step{t}=s ~\big|~ \px, \py, s\step{0} = \bar{s}\right) r(s, \px, \py) = \sum_{s\in\calS} d_{\vrho}^{\px, \py}(s) r(s, \px, \py).
    \end{align}
\end{proof}

\begin{claim}
    Let $\policy{} \in \Pi$ be a joint stationary policy, $\vr(\policy{})$ be the reward vector under $\policy{}$, and $\vv, \vc \in \R^S$. If $\vr (\policy{}) + \gamma \pr(\policy{}) \vv \leq \vv  + \vc$, then it holds that
    \begin{equation}
        \vec{V}( \policy{} ) \leq \vv + \big( \mat{I} - \gamma \pr(\policy{}) \big)^{-1}\vc.
    \end{equation}
    Similarly, if $\vr (\policy{}) + \gamma \pr(\policy{}) \vv \geq \vv  + \vc$, then it holds that 
    \begin{equation}
    \vec{V}( \policy{} ) \geq \vv + \big( \mat{I} - \gamma \pr(\policy{}) \big)^{-1}\vc.    
    \end{equation}
    \label{claim:vec-inequality}
\end{claim}

\begin{proof}
    Suppose that $\vr (\policy{}) + \gamma \pr(\policy{}) \vv \leq \vv  + \vc$. Applying recursively this inequality, it follows that
    \begin{equation}
        \sum_{t=0}^\infty \gamma^t \pr^{t}(\policy{}) \vr(\policy{}) -          \sum_{t=0}^\infty \gamma^t \pr^{t}(\policy{}) \vc \leq \vv.  
    \end{equation}
    Combining this bound with \Cref{claim:iminusp-invertible,fact:value-from-matrix} implies that
    \begin{equation}
        \vec{V}(\policy{}) - \Big( \mat{I} - \gamma \pr(\policy{}) \Big)^{-1} \vc \leq \vv.
    \end{equation}
    The case where $\vr (\policy{}) + \gamma \pr(\policy{}) \vv \geq \vv  + \vc$ admits an analogous proof.
\end{proof}

\begin{claim}
    \label{claim:stratequiv}
    Consider an adversarial team Markov game $\calG$. Altering all the rewards by adding an additive constant $c\in\R$ yields a \emph{strategically-equivalent} game $\calG'$: any $\epsilon$-approximate Nash equilibrium in $\calG'$ is also an $\epsilon$-approximate Nash equilibrium in $\calG$, and \emph{vice versa}.
\end{claim}
\begin{proof}
    By assumption, $r'(s,\va, b) = r(s,\va, b) + c $ for any $(s, \vec{a}, b) \in \calS \times \calA \times \calB$. Let $V'_{\vrho}$ be the value function in $\calG'$. Then, for all $(\px, \py) \in \calX \times \calY$,
    \begin{align}
        V'_{\vrho} (\px, \py )  &= \vrho^\top \left( \mat{I} - \gamma \pr(\px, \py) \right)^{-1} \vr'(\px, \py) \\
        &= \vrho^\top \left( \mat{I} - \gamma \pr(\px, \py) \right)^{-1} ( \vr(\px, \py) + c \cdot \vone ) \\
        &= V_{\vrho}(\px, \py)  + \frac{c}{1-\gamma}.
    \end{align}
    Thus, our claim follows immediately from the definition of Nash equilibria (\Cref{def:Nash}).
\end{proof}

\begin{claim}
    \label{claim:rho-ul}
    Let $\policy{} \in \Pi$ be a joint stationary policy, and $\vd^{\policy{}}_{\vrho}$ be the induced visitation measure. Then, for every $s \in \calS$,
         $$ \rho(s) \leq d^{\policy{}}_{\vrho}(s) \leq \frac{1}{1 - \gamma}.$$
\end{claim}
\begin{proof}
This is an immediate consequence of the definition of $d^{\policy{}}_{\vrho}$; in particular,
\begin{equation}
    d^{\policy{}}_{\vrho}(s) =\sum_{\bar{s} \in \calS }\rho(\bar{s}) \sum_{t=0}^{\infty} \gamma^t { \pr} (s\step{t} = s | \policy{}, s\step{0} = \bar{s}) \leq \sum_{\bar{s} \in \calS} \rho(\bar{s}) \sum_{t=0}^\infty \gamma^t = \frac{1}{1 - \gamma},
\end{equation}
and 
\begin{equation}
    d^{\policy{}}_{\vrho}(s) =\sum_{\bar{s} \in \calS }\rho(\bar{s}) \sum_{t=0}^{\infty} \gamma^t { \pr} (s\step{t} = s | \policy{}, s\step{0} = \bar{s}) \geq \rho(s) \sum_{t=0}^{\infty} \gamma^t { \pr} (s\step{t} = s | \policy{}, s\step{0} = s) \geq \rho(s).
\end{equation}
\end{proof}

\begin{claim}
    \label{claim:V-bounds}
    Suppose that the reward function takes values in $[m_r, M_r]$, for some $m_r, M_r > 0$. Then, for any stationary joint policy $\policy{} \in \Pi$ and every state $s \in \calS$,
    \begin{equation}
        \frac{\minrew}{1-\gamma}\leq V_s(\policy{}) \leq \frac{\maxrew}{1-\gamma}.
    \end{equation}
\end{claim}

\begin{proof}
    By the definition of the value function in~\eqref{eq:value-func-def}, we have
    \begin{equation}
        V_s(\policy{}) \leq \maxrew + \gamma \maxrew + \gamma^2 \maxrew + \cdots = \frac{1}{1-\gamma} \maxrew,
    \end{equation}
    for any $s \in \calS$. Similarly, we conclude that
    \begin{equation}
       V_s(\policy{}) \geq \frac{1}{1-\gamma} \minrew.
    \end{equation}
\end{proof}

\begin{claim}
    \label{claim:various-bounds}
    Let an adversarial team Markov game $\calG$, two team policies $\tilde{\px}, \tilde{\px}$ and quantities $R_b(\cdot, \cdot), P_b(\cdot| s, \cdot), v(s)$ quantities defined in \pref{prog:xinlp}.
    The following inequalities hold:
    \begin{enumerate}
        \item $\left| r\left(s, \tilde{\px}, b \right) - r\left(s, \hat{\px}, b \right) \right| \leq \sqrt{\sum_{k=1}^n A_k} \| \tilde{\px} - \hat{\px} \|$, for any $(s,b) \in \calS \times \calB$; \label{item:reward}
    \item \label{item:prob} $ \left|  \sum_{s' \in \calS} \Big( \pr\left(s'|s,  \tilde{\px}, b \right) - \pr \left(s'|s, \hat{\px},b \right) \Big)
             \tilde{v}(s')  \right| \leq  \frac{S}{1-\gamma} \sqrt{\sum_{k=1}^n A_k} \|\tilde{\vx} - \hat{\vx}\|$, for any $(s,b) \in \calS \times \calB$;
    \item \label{item:val} $\left| \tilde{v}(s) - \hat{v}(s) \right| \leq L \| \tilde{\px} - \hat{\px}\|$, for any $s \in \calS$; and
    \item \label{item:last} $\left|  \sum_{s' \in \calS}  \pr \left(s'|s, \hat{\px}, b \right) \big( \tilde{v}(s') - \hat{v}(s') \big)  \right| \leq S  L \| \tilde{\px} - \hat{\px}\|$, for any $(s,b) \in \calS \times \calB$.
    \end{enumerate}
\end{claim}
\begin{proof} We briefly note how the bounds are derived:

    \begin{itemize}
        \item We first establish \Cref{item:reward}. Fix any pair $(s,b) \in \calS \times \calB$. By definition, we have
        \begin{equation}
            r(s, \tilde{\vx}, b) = \E_{\vec{a} \sim \tilde{\vx}}[r(s, \vec{a}, b)] = \sum_{(a_1, \dots, a_n) \in \calA} r(s, \vec{a}, b) \prod_{k = 1}^n \tilde{x}_{k, s, a_k}.
        \end{equation}
        As a result, 
        \begin{align}
            |r(s, \tilde{\vx}, b) - r(s,\hat{\vx}, b)| &= \left| \sum_{(a_1, \dots, a_n) \in \calA} r(s, \vec{a}, b) \prod_{k = 1}^n \tilde{x}_{k, s, a_k} - \sum_{(a_1, \dots, a_n) \in \calA} r(s, \vec{a}, b) \prod_{k = 1}^n \hat{x}_{k, s, a_k} \right| \\
            &= \left| \sum_{(a_1, \dots, a_n) \in \calA} r(s, \vec{a}, b) \left( \prod_{k=1}^n \tilde{x}_{k, s, a_k} - \prod_{k=1}^n \hat{x}_{k, s, a_k} \right) \right|\\
            &\leq \sum_{(a_1, \dots, a_n) \in \calA} \left| \prod_{k=1}^n \tilde{x}_{k, s, a_k} - \prod_{k=1}^n \hat{x}_{k, s, a_k} \right| \label{eq:triangle} \\
            &\leq \sum_{k=1}^n \| \tilde{\vx}_{k,s} - \hat{\vx}_{k, s} \|_1 = \| \tilde{\vx}_s - \hat{\vx}_s \|_1 \leq \left( \sqrt{\sum_{k=1}^n A_k} \right) \| \tilde{\vx}_s - \hat{\vx}_s\|_2, \label{eq:sum-product}
        \end{align}
        where \eqref{eq:triangle} follows from the triangle inequality and the fact that $|r(s, \vec{a}, b)| \leq 1$, and \eqref{eq:sum-product} follows from the fact that the total variation distance between two product distributions is bounded by the sum of the total variations of each marginal distribution~\citep{Hoeffding58:Distinguishability}, as well as the fact that $\|\vx \|_1 \leq \sqrt{d} \|\vx\|_2$ for a vector $\vx \in \R^d$.
        
        \item \Cref{item:prob} follows analogously to \Cref{item:reward}, using the fact that $\tilde{v}(s') \leq \frac{1}{1 - \gamma}$ (by \Cref{claim:V-bounds} and \Cref{prop:phi-and-nlp}).
        \item For \Cref{item:val}, we begin by noting that $\hat{\vv}$ and $\tilde{\vv}$ are the unique optimal vectors of \pref{prog:xinlp} for $\hat{\px}$ and $\tilde{\px}$ respectively (recall \Cref{prop:phi-and-nlp}). Further, by \Cref{prop:phi-and-nlp}, we know that $\vrho^\top \hat{\vv} =\max_{\py \in \calY} V_{\vrho}(\hat{\px}, \py) = \phi( \hat{\px})$ and 
        $\vrho^\top \tilde{\vv} =\max_{\py \in \calY} V_{\vrho}(\tilde{\px}, \py) = \phi( \tilde{\px})$,
        for any $\vrho\in\Delta(\calS)$ of full support. As a result, \Cref{item:val} is a consequence of the fact that $\phi(\cdot)$ is $L$-Lipschitz continuous, which in turn follows since $V_{\vrho}$ is $L$-Lipschitz continuous (see \Cref{lem:smoothness} and \Cref{lem:max-weakly-convex}).
        
        \item Finally, \Cref{item:last} follows from \Cref{item:val} and the fact that 
        \begin{align}
            \left| \sum_{s' \in \calS} \pr(s' | s, \hat{\vx}, b) \right| &= \left| \sum_{s' \in \calS} \sum_{(a_1, \dots, a_n) \in \calA} \pr(s' | s, \vec{a}, b) \prod_{k=1}^n \hat{\vx}_{k, s, a_k} \right| \\
            &\leq \sum_{s' \in \calS} \sum_{(a_1, \dots, a_n) \in \calA} \prod_{k=1}^n \hat{\vx}_{k, s, a_k} = S,
        \end{align}
        for any fixed $(s,b) \in \calS \times \calB$, where the last bound follows from the triangle inequality and the normalization constraint of the product distribution: $\sum_{(a_1, \dots, a_n) \in \calA} \prod_{k=1}^n \hat{\vx}_{k, s, a_k} = 1$.
    \end{itemize}
\end{proof}



\end{document}